\documentclass[conference,romanappendices]{IEEEtran}

\usepackage[T1]{fontenc}
\usepackage[utf8x]{inputenc}
\usepackage[english]{babel}

\usepackage{graphicx}
\usepackage{latexsym}
\usepackage{mathpartir}
\usepackage{amsfonts,amssymb,mathtools,amsthm,extarrows}
\usepackage{thmtools,thm-restate}

\usepackage[usenames,dvipsnames]{xcolor}
\usepackage{listings}
\usepackage{courier}
\usepackage{hyperref}

\usepackage{url}

\usepackage{tikz}
\usetikzlibrary{shapes,arrows,positioning,matrix,automata,decorations.pathreplacing,calc}

\usepackage{xspace}
\usepackage{cmds}
\usepackage{mystyle}
\usepackage{bookmark}
\usepackage{booktabs}
\usepackage{stmaryrd}
\usepackage{subcaption}
\usepackage{thmtools}
\usepackage{placeins}
\usepackage{enumerate}
\usepackage{listings}

\usepackage{pifont}
\usepackage{multirow}

\usepackage{longfigure}

\usepackage[makeroom]{cancel}

\usepackage{algorithm}
\usepackage{algpseudocode}
\usepackage{mathpartir}
\usepackage[english]{babel}

\usepackage{amsfonts,amssymb,mathtools,amsthm}
\usepackage{thmtools,thm-restate}
\usepackage{hyperref}
\usepackage[usenames,dvipsnames]{xcolor}

\usepackage{tikz}
\usetikzlibrary{shapes,arrows,positioning,matrix,automata,decorations.pathreplacing}

\usepackage{xspace}
\usepackage{booktabs}
\usepackage{multirow}

\usepackage{cmds}
\usepackage{mystyle}
\usepackage{csquotes}

\usepackage{listings}
\lstdefinelanguage{msp430asm}{
	morekeywords={mov,add,addc,sub,subc,cmp,dadd,bit,bic,bis,
				  xor,and,rrc,rra,push,swpb,call,reti,sxt,jeq,
				  jz,jne,jnz,jnc,jlo,jc,jhs,jn,jge,jl,jmp,
				  adc,dadc,dec,decd,inc,incd,sbc,inv,rla,rlc,
				  br,dint,eint,nop,ret,clr,clrc,clrn,clrz,pop,
				  setc,setn,setz,tst,INST1,INST2},
	sensitive=false,
	morecomment=[l]{;},
	morecomment=[s]{/*}{*/},
	morestring=[b]{"},
}
\usepackage[inline, shortlabels]{enumitem}

\allowdisplaybreaks

\mprset{sep=1em} 
\mprset{vskip=0.1em}
\usepackage{xstring}

\title{Provably Secure Isolation for Interruptible Enclaved Execution on Small Microprocessors: Extended Version}

\author{
\IEEEauthorblockN{
  Matteo Busi\IEEEauthorrefmark{1},
  Job Noorman\IEEEauthorrefmark{2},
  Jo Van Bulck\IEEEauthorrefmark{2},\\
  Letterio Galletta\IEEEauthorrefmark{3},
  Pierpaolo Degano\IEEEauthorrefmark{1},
  Jan Tobias M\"uhlberg\IEEEauthorrefmark{2}
  and Frank Piessens\IEEEauthorrefmark{2}
}
\IEEEauthorblockA{\IEEEauthorrefmark{1} Dept. of Computer Science, Universit\`a di Pisa, Italy \\}
\IEEEauthorblockA{\IEEEauthorrefmark{2} imec-DistriNet, Dept. of Computer Science, KU Leuven, Belgium\\}
\IEEEauthorblockA{\IEEEauthorrefmark{3} IMT School for Advanced Studies Lucca, Italy}
}
\date{}

\newcommand{\onlinetechrep}{the appendices\xspace}
\newcommand{\techrep}{the appendices\xspace}
\newcommand{\algomem}{Appendix~\ref{app:security}, Algorithm~\ref{algo:memoryctx}}
\newcommand{\algodev}{Appendix~\ref{app:security}, Algorithm~\ref{algo:devicectx}}
\newcommand{\isolationlemmata}{Appendix~\ref{app:security}, Lemmata~\ref{lemma:umeqtraces} and~\ref{lemma:pmeqtraces}}

\begin{document}
\maketitle
\makeatletter
\def\ps@IEEEtitlepagestyle{
  \def\@oddfoot{\mycopyrightnotice}
  \def\@evenfoot{}
}
\def\mycopyrightnotice{
  {\footnotesize
  \begin{minipage}{\textwidth}
  \centering
  \copyright~2020 IEEE.
  Personal use of this material is permitted.
  Permission from IEEE must be obtained for all other uses, in any current or future media, including reprinting/republishing this material for advertising or promotional purposes, creating new collective works, for resale or redistribution to servers or lists, or reuse of any copyrighted component of this work in other works.
  \end{minipage}
  }
}
%

\begin{abstract}
  Computer systems often provide hardware support for isolation mechanisms like privilege levels, virtual memory,
  or enclaved execution. Over the past years, several successful software-based side-channel attacks have been developed that
  break, or at least significantly weaken the isolation that these mechanisms offer. Extending a processor with new
  architectural or micro-architectural features, brings a risk of introducing new such side-channel attacks.

  This paper studies the problem of extending a processor with new features {\em without} weakening the security of the isolation mechanisms that the processor offers.
  We propose to use full abstraction as a formal criterion for the security of a processor extension, and we
  instantiate that criterion to the concrete case of extending a microprocessor that supports enclaved execution with
  secure interruptibility of these enclaves. This is a very relevant instantiation as several recent papers have shown that
  interruptibility of enclaves leads to a variety of software-based side-channel attacks. We propose a design for interruptible
  enclaves, and prove that it satisfies our security criterion. We also implement the design on an open-source enclave-enabled
  microprocessor, and evaluate the cost of our design in terms of performance and hardware size.

  \emph{This is the extended version of the paper~\cite{csf} that includes both the original paper as well as the technical appendix with the proofs.}
\end{abstract}

\section{Introduction}\label{sec:introduction}
    Many computing platforms run programs coming from a number of different stakeholders that do not necessarily trust each other. Hence, these platforms provide mechanisms to prevent
code from one stakeholder to interfere with code from other stakeholders in undesirable ways.
These {\em isolation mechanisms} are intended to confine the interactions between two isolated programs to a well-defined communication interface.
Examples of such isolation mechanisms include process isolation, virtual machine monitors, or enclaved execution \cite{sgx}.

However, security researchers have shown that many of these isolation mechanisms can be attacked
by means of {\em software-exploitable side-channels}.
Such side-channels have been shown to violate integrity of
victim programs \cite{rowhammer,clockscrew,Murdock2019plundervolt}, as well as their confidentiality on both high-end processors \cite{cache-attacks,meltdown,spectre,foreshadow}
and on small microprocessors \cite{nemesis}.
In fact, over the past two years,
many major isolation mechanisms have been successfully attacked: Meltdown \cite{meltdown} has broken user/kernel isolation, Spectre \cite{spectre} has
broken process isolation and software defined isolation, and Foreshadow \cite{foreshadow} has broken enclaved execution on Intel processors.

The class of software-exploitable side-channel attacks is complex and varied. These attacks often exploit, or at least rely on, specific hardware features or hardware implementation
details. Hence, for complex state-of-the-art processors there is a wide potential attack surface that should be explored (see for instance \cite{transient} for an overview
of just the attacks that rely on transient execution). Moreover, the potential attack vectors vary with the attacker model that a specific isolation mechanism considers.
For instance, enclaved execution is designed to protect  enclaved code from malicious operating system software whereas process isolation assumes that the operating system
is trusted and not under control of the attacker. As a consequence, protection against software-exploitable side-channel attacks is much harder for enclaved
execution \cite{controlled-channels}.

Hence, no silver-bullet solutions against this class of attacks should be expected, and countermeasures will likely be as varied as the attacks.
They will depend on attacker model, performance versus security trade offs, and on the specific processor feature that is being exploited.

The objective of this paper is to study how to design and prove secure such countermeasures.
In particular, we rigorously study the resistance of enclaved execution on small microprocessors \cite{sancus,trustlite} against interrupt-based attacks \cite{nemesis,sgxlinger,sgx-step}.
%
%
This specific instantiation is important and challenging. First, interrupt-based attacks are very powerful against enclaved execution: fine-grained interrupts
have been a key ingredient in many attacks against enclaved execution \cite{branch-shadowing,foreshadow,sgxpectre,nemesis}. Second, to the best of our knowledge, all existing  implementations
of interruptible enclaved execution are vulnerable to software-exploitable side-channels, including implementations that specifically aim for secure interruptibility \cite{ruan,trustlite}.

We base our study on the existing open-source Sancus platform \cite{sancus1,sancus} that supports {\em non-interruptible} enclaved execution.
We illustrate that achieving security is non-trivial through a variety of attacks enabled by supporting interruptibility of enclaves.
Next, we provide a formal model of the existing Sancus and we then extend it with interrupts.
We prove that this extension does not break isolation properties by instantiating full abstraction~\cite{abadi1999protection}.

%
Roughly,
we show that what the attacker can learn from (or do to) an enclave is exactly the same \emph{before} and \emph{after} adding the support for interrupts.
In other words, adding interruptibility does not open new avenues of attack.
Finally, we implement the secure interrupt handling mechanism as an extension to Sancus, and we show that the cost of the mechanism is low, in terms of both hardware complexity and performance.

In summary, the novel contributions of this paper are:
\begin{itemize}[nosep]
\item 
We propose a specific design for extending Sancus, an existing enclaved execution system, with interrupts.
\item
We propose to use full abstraction~\cite{abadi1999protection} as a formal criterion of what it means to maintain the security of isolation mechanisms under processor extensions.
Also, we instantiate it for proving that the mechanism of enclaved execution, extended to support interrupts, complies with our security definition.
\item We implement the design on the open source Sancus processor, and evaluate cost in terms of hardware size and performance impact.%
\footnote{
Our implementation is available online at~\url{https://github.com/sancus-pma/sancus-core/tree/nemesis}.
}
\end{itemize}
The paper is structured as follows: in Section~\ref{sec:background} we provide background information on enclaved execution and interrupt-based attacks.
Section~\ref{sec:overview} provides an informal overview of our approach.
Section~\ref{sec:semantics} discusses our formalization and sketches the proof, pointing to~\onlinetechrep for full details.
Then, in Section~\ref{sec:implementation} we describe and evaluate our implementation.
Section~\ref{sec:discussion} and~\ref{sec:related} discuss limitations, and the connection to related work.
Finally, Section~\ref{sec:conclusion} offers our conclusions and plans for future work.

\section{Background}\label{sec:background}
\paragraph{Enclaved execution}
Enclaved execution is a security mechanism that enables {\em secure
remote computation} \cite{intel-sgx-explained}. It supports the
creation of {\em enclaves} that are initialized with a software
module, and that have the following security properties. First,
the software module in the enclave is isolated from all other
software on the same platform, including system software such as the
operating system. Second, the correct initialization of an enclave
can be {\em remotely attested}: a remote party can get cryptographic assurance that
an enclave was properly initialized with a specific software module
(characterized by a cryptographic hash of the binary module).
These security properties are guaranteed while relying on a small
trusted computing base, for instance trusting only the hardware \cite{sancus,sgx},
or possibly also a small hypervisor \cite{trustvisor,komodo}.

The remote attestation aspect of enclaved execution is important for
the secure initialization of enclaves, and
for setting up secure communication
channels to the enclave.
However, it does not play an important role for the
interrupt-driven attacks that we study in this paper, and hence we will focus
    here on the isolation aspect of enclaves only. Other papers
describe in detail how remote attestation and secure communication
work on large \cite{intel-sgx-explained} or small systems \cite{sancus,trustlite}.

The isolation guarantees offered to an enclaved software module are
the following. The module consists of two contiguous memory
sections, a {\em code section}, initialized with the machine code of the module,
and a {\em data section}.
The data section is initialized to zero, and loading of confidential data happens
through a secure channel to the enclave, after attesting
the correct initialization of the module. For instance, confidential data can be restored from cryptographically sealed
storage, or can be obtained from a remote trusted party.

The enclaved execution platform guarantees that: (1) the data section
of an enclave is {\em only} accessible while executing code from
the code section, and (2) the code section can only be entered through
one or more designated {\em entry points}.

These isolation guarantees are simple, but they offer the useful
property that {\em data of a module can only be manipulated by code
of the same module}, i.e., an encapsulation property similar to what
programming languages offer through classes and objects.
Untrusted code residing in the same address space as the enclave
but outside the enclave code and data sections can interact with
the enclave by jumping to an entry point. The enclave can return
control (and computation results) to the untrusted code by jumping back
out.


\paragraph{Interrupt-based attacks}\label{sec:background:attacks} Enclaved
execution is designed to be resistant against a very strong attacker that
controls all other software on the platform, including privileged
system software. While isolating enclaves is well-understood at the
architectural level, including even successful formal verification
efforts~\cite{komodo,vrased}, researchers have shown that it is challenging to
protect enclaves against side-channels.
Particularly, a recent line of work on {\em controlled channel}
attacks~\cite{controlled-channels,sgx-step,nemesis,stealthy-page-tables,branch-shadowing}
has demonstrated a new class of powerful, low-noise side-channels that leverage
the adversary's increased control over the untrusted operating system.

A specific consequence of this strong model is that the attacker also controls the scheduling
and handling of interrupts: the attacker can precisely schedule interrupts to arrive during enclaved execution,
and can choose the code to handle these interrupts.
This power has been put to use
for instance to single-step through an enclave~\cite{sgx-step}, or
to mount a new class of ingenious \emph{interrupt latency} attacks~\cite{nemesis,sgxlinger}
that derive individual enclaved instruction timings from the time it takes to
dispatch to the untrusted operating system's interrupt handler.
We provide concrete examples of interrupt-based attacks
in the next section, after detailing our model of enclaved execution.

While advanced CPU features such as virtual
memory~\cite{controlled-channels,stealthy-page-tables,foreshadow},
branch prediction~\cite{branch-shadowing,sgxpectre} or
caching~\cite{sgx-cache} are known to leak information on high-end processors,
pure interrupt-based attacks such as interrupt latency measurements
are the {\em only} known controlled-channel attack against low-end enclaved execution
platforms lacking these advanced features.
Moreover, they have been shown to be very powerful: e.g., Van Bulck et al.~\cite{nemesis} have shown how to
efficiently extract enclave secrets like passwords or PINs from embedded enclaves.

Some enclaved execution designs avoid the problem of interrupt-based attacks by completely
disabling interrupts during enclave execution \cite{sancus,vrased}. This has the important downside
that system software can no longer guarantee availability: if an enclaved module goes into
an infinite loop, the system cannot progress.
All 
designs that do support interruptibility of enclaves \cite{ruan,trustlite}
are vulnerable to these attacks.

\section{Overview of our approach}\label{sec:overview}

We set out to design an interruptible enclaved execution system that is provably resistant against
interrupt-based attacks.
This section discusses our approach informally, later sections discuss
a formalization with security proofs, and report on implementation and experimental
evaluation.

We base our design on Sancus \cite{sancus}, an existing open-source enclaved execution system.
We first describe our Sancus model, and discuss how extending Sancus with interrupts leads to the attacks mentioned in Section~\ref{sec:background:attacks}.
In other words, we show how extending Sancus with interrupts breaks some of the isolation guarantees provided by Sancus.

Then, we propose a formal security criterion that defines what it means for interruptibility to {\em preserve
the isolation properties}, and we illustrate that definition with examples.

Finally, we propose a design for an interrupt handling mechanism that is
resistant against the considered attacks and that
satisfies our security definition.
%
Crucial to our design is the assumption that the timing of individual instructions is predictable, which is typical of ``small'' microprocessors, like Sancus.
%
Although tailored here on a specific architecture and a specific class of attacks, we expect our approach of ensuring that the same attacks are possible before and after
an architecture extension to be applicable in other settings too.

\subsection{Sancus model}
\paragraph{Processor}
Sancus is based on the TI MSP430 16-bit microprocessor \cite{ti-msp430}, with a classic von Neumann architecture where code and data share the same address space.
We formalize the subset of instructions summarized in~\tablename~\ref{tab:op-summary} that is rich enough to model all the attacks we care about.
We have a subset of memory-to-register and register-to-memory transfer instructions; a comparison instruction; an unconditional and a conditional jump; and basic arithmetic instructions.
\begin{table}[tb]
    \centering
    \resizebox{\columnwidth}{!}{
    \begin{tabular}{@{}llll@{}}
    \toprule
    \textbf{Instr. $i$} &
    \textbf{Meaning} 
    & \textbf{Cycles}  & \textbf{Size} \\
    \midrule
    $\RETI$ & Returns from interrupt. & $5$ & $1$\\
    $\NOP$  & No-operation.           & $1$                 & $1$\\
    $\HLT$  & Halt.                   & $1$                 & $1$\\
    $\NOT{r}$  & $\rn{r} \leftarrow \lnot \rn{r}$. (Emulated in MSP430) & $2$ & $2$\\
    $\IN{r}$  & Reads word from the device and puts it in $\rn{r}$. & $2$ & $1$\\
    $\OUT{r}$  & Writes word in register $\rn{r}$ to the device. & $2$ & $1$\\
    $\AND{r_1}{r_2}$  & $\rn{r_2} \leftarrow \rn{r_1}\ \&\ \rn{r_2}$. & $1$ & $1$\\
    $\JMP{r}$ & Sets $\rpc$ to the value in $\rn{r}$. & $2$ & $1$\\
    $\JZ{r}$ & Sets $\rpc$ to the value in $\rn{r}$ if bit 0 in $\rsr$ is set. & $2$ & $1$\\
    $\MOV{r_1}{r_2}$  & $\rn{r_2} \leftarrow \rn{r_1}$. & $1$ & $1$\\
    $\MOVL{r_1}{r_2}$  & Loads in $\rn{r_2}$ the word in starting in location pointed by $\rn{r_1}$. & $2$ & $1$\\
    $\MOVS{r_1}{r_2}$  & Stores the value of $\rn{r_1}$ starting at location pointed by $\rn{r_2}$. & $4$ & $2$\\
    $\MOVI{w}{r_2}$  & $\rn{r_2} \leftarrow w$. & $2$ & $2$\\
    $\ADD{r_1}{r_2}$  & $\rn{r_2} \leftarrow \rn{r_1} + \rn{r_2}$. & $1$ & $1$\\
    $\SUB{r_1}{r_2}$  & $\rn{r_2} \leftarrow \rn{r_1} - \rn{r_2}$. & $1$ & $1$\\
    $\CMP{r_1}{r_2}$  & Zero bit in $\rsr$ set if $\rn{r_2} - \rn{r_1}$ is zero. & $1$ & $1$\\
    \bottomrule\\
    \end{tabular}}
\caption{Summary of the assembly language considered.}\label{tab:op-summary}
\end{table}
%


\paragraph{Memory}
Sancus has a byte addressable memory of at most {\footnotesize 64KB}, where a finite number of enclaves can be defined.
The bound on the number of enclaves is a parameter set at
processor synthesis time.
In our model, we assume that there is only a single enclave,
made of a {\em code section}, initialized with the machine code of the module, and a {\em data section}.
A data section is securely provisioned with data by relying on remote attestation and secure communication, not modeled here as they play no role in the interrupt-based attacks we care about in this paper.
Instead, our model allows direct initialization of the data section with confidential enclave data.
All the other memory is {\em unprotected memory}, and will be considered to be under
control of the attacker.

Enclaves have a single entry point; the enclave can only be entered by jumping to the first address of the code section.
Multiple {\em logical entry points} can easily be implemented on top of this single physical entry point. Control flow can leave the enclave
by jumping to any address in unprotected memory. Obviously, a compiler can implement higher-level abstractions such as enclave
function calls and returns, or out-calls from the enclave to functions in the untrusted code~\cite{sancus}.

Sancus enforces program counter (pc) based memory access control.
If the pc is in unprotected memory, the processor can not access any memory location within the enclave -- the only way
to interact with the enclave is to jump to the entry point. If the pc is within the code section of the enclave, the processor can only access the enclave data section for reading/writing and the enclave code
section for execution.
This access control is faithfully rendered in our model, via the predicate MAC in~\tablename~\ref{tab:mac}.


\paragraph{I/O devices}

Sancus uses memory-mapped I/O to interact with peripherals. One important example of a peripheral for the attacks we study is a cycle accurate timer, which allows software to measure
time in terms of the number of CPU cycles.
In our model, we include a single very general I/O device that behaves as a state machine running synchronously to CPU execution. In particular, it is trivial to instantiate
this general I/O device to a cycle-accurate timer.

Instead of modeling memory-mapped I/O, we introduce two special instructions that allow writing/reading a word to/from the device (see~\tablename~\ref{tab:op-summary}).
Actually these instructions are short-hands, which are easy to macro-expand, at the price of dealing with  special cases in the execution semantics for any memory operation.
For instance, software could read the current cycle timer value from a timer peripheral by using the $\IN{}\!\!$ instruction.

The I/O devices can request to interrupt the processor with single-cycle accuracy. The original Sancus disables interrupts during enclaved execution.
One of the key objectives of this paper is to propose a Sancus extension that does handle such interrupts without weakening security.
Hence, we will define two models of Sancus, one that ignores interrupts, and one that handles them even during enclaved execution.

\subsection{Security definitions}
\paragraph{Attacker model} An attacker
controls the entire {\em context} of an enclave, that is: he controls
$(1)$ all of unprotected memory (including code interacting with the enclave, as well as data in unprotected memory), and
$(2)$ the connected device.
This is the standard attacker model for enclaved execution.
In particular, it implies that the attacker has complete control over the Interrupt Service Routines.


\paragraph{Contextual equivalence formalizes isolation}
Informally, our security objective is extending the Sancus processor without
weakening the isolation it provides to enclaves.
What isolation achieves is that attackers can not see ``inside'' an
enclave,
so making it possible to ``hide'' enclave data or implementation
details from the attacker.
We formalize this concept of isolation precisely by using the notion
of  {\em contextual equivalence} or {\em contextual indistinguishability}
(as first proposed by Abadi \cite{abadi1999protection}).
Two enclaved modules $M_1$ and $M_2$ are contextually equivalent, if
the attacker can not distinguish them, i.e., if there exists no
context that tells them apart.
We discuss this on the following example.

\begin{example}[Start-to-end timing]\label{ex:start-to-end}
   The following enclave compares a user-provided password in $\reg{15}$
    with a secret in-enclave password at address $\mi{pwd\_adrs}$, and stores
    the user-provided value in $\reg{14}$ into the enclave location at
    $\mi{store\_adrs}$ if the user password was correct.
\begin{lstlisting}
enclave_entry:
    /* Load addresses for comparison */
    MOV #store_adrs, r10    ; 2 cycles
    MOV #access_ok, r11     ; 2 cycles
    MOV #endif, r12         ; 2 cycles
    MOV #pwd_adrs, r13      ; 2 cycles
    /* Compare user vs. enclave password */
    MOV @r13, r13           ; 2 cycles
    CMP r13, r15            ; 1 cycle
    JZ  &r11                ; 2 cycles
access_fail:   /* Password fail: return */
    JMP &r12                ; 2 cycles
access_ok:  /* Password ok: store user val */
    MOV r14, 0(r10)         ; 4 cycles
endif:  /* Clear secret enclave password */
    SUB r13, r13            ; 1 cycle
enclave_exit:
\end{lstlisting}
\end{example}

In the absence of a timer device, this enclave successfully hides the in-enclave password.
If we take enclaves $M_1$ and $M_2$ to be two instances of
Example~\ref{ex:start-to-end}, differing only in the value for the secret password, then
$M_1$ and $M_2$ are indistinguishable for any context that does not have
access to a cycle accurate timer: all such a context can do is call the entry point,
but the context does not get any indication whether the user-provided password
was correct. This formalizes that enclave isolation successfully ``hides'' the password.

However, with the help of a cycle accurate timer, the attacker can distinguish $M_1$ and $M_2$ as follows.
The attacker can create a context that measures the start-to-end execution
time of an enclave call: the context reads the timer right before
jumping to the enclave. On enclave exit, the context reads the timer again to
compute the total time spent in the enclave.

In order to reason about execution timing, we represent enclaved
executions as an ordered array of individual instruction timings.  (\tablename~\ref{tab:op-summary}
conveniently specifies how many cycles it takes to execute each instruction.)
Hence the two possible control flow paths of the above program are:
    \texttt{ok=[2,2,2,2,2,1,2,4,1]} for the \enquote{access\_ok} branch, or
    \texttt{fail=[2,2,2,2,2,1,2,2,1]} for the \enquote{access\_fail} branch.
Since \texttt{sum(ok) = 18} and \texttt{sum(fail) = 16}, the context can distinguish
the two control flow paths, and hence can distinguish $M_1$ and $M_2$ (and by
launching a brute-force attack~\cite{msp430-bsl-timing}, can also extract the secret password).

This example illustrates how contextual equivalence formalizes
isolation. It also shows that the original Sancus already has some side-channel
vulnerabilities under our attacker model. Since we assume the
attacker can use any I/O device, he can choose to use a timer device
and mount the start-to-end timing attack we discussed.

It is important to note that it is {\em not} our objective in this paper to close
these existing side-channel vulnerabilities in Sancus.
Our objective is to make sure that extending Sancus with interrupts
does not introduce {\em additional} side-channels, i.e., that this
does not {\em weaken} the isolation properties of Sancus.

For existing side-channels, like the start-to-end timing side-channel, countermeasures
can be applied by the enclave programmer. For instance, the programmer can
balance out the various secret-dependent control-flow paths as in Example~\ref{ex:latency}.

\begin{example}[Interrupt latency]\label{ex:latency}
    Consider the program of
    Example~\ref{ex:start-to-end}, balanced in terms of overall execution time by adding two
    $\NOP$ instructions at lines 13-14.
    The two possible control flow paths are:
    \texttt{ok=[2,2,2,2,2,1,2,4,1]} vs.
    \texttt{fail= [2,2,2,2,2,1,2,1,1,2,1]}.
    Since  \texttt{sum(ok)} is equal to \texttt{sum(fail)}, the start-to-end timing attack is mitigated.
\begin{lstlisting}
enclave_entry:
    /* Load addresses for comparison */
    MOV #store_adrs, r10    ; 2 cycles
    MOV #access_ok, r11     ; 2 cycles
    MOV #endif, r12         ; 2 cycles
    MOV #pwd_adrs, r13      ; 2 cycles
    /* Compare user vs. enclave password */
    MOV @r13, r13           ; 2 cycles
    CMP r13, r15            ; 1 cycle
    JZ  &r11                ; 2 cycles
access_fail:
    /* Password fail: constant time return */
    NOP                     ; 1 cycle
    NOP                     ; 1 cycle
    JMP &r12                ; 2 cycles
access_ok: /* Password ok: store user val */
    MOV r14, 0(r10)         ; 4 cycles
endif:   /* Clear secret enclave password */
    SUB r13, r13            ; 1 cycle
enclave_exit:
\end{lstlisting}
\end{example}

\paragraph{Interrupts can weaken isolation}
We now show that a straightforward implementation of interrupts in the Sancus processor
would significantly weaken isolation.
Consider an implementation of interrupts similar to the TI MSP430: on arrival of an interrupt,
the processor first completes the ongoing instruction, and then jumps to an interrupt service routine.

The program in Example~\ref{ex:latency} is secure on Sancus without interrupts.
However, it is not secure against a malicious context that can schedule interrupts
to be handled
while the enclave executes.
To see why, assume that an interrupt is scheduled by the malicious
context to arrive within the first cycle after the conditional jump at line 10.
If the jump was taken then the instruction being executed is the 4-cycle
\MOV{\!\!}{\!\!}~at line 18, otherwise the current instruction is the 1-cycle
\NOP\ at line 13. Now, since the attacker's interrupt handler will only
be called \emph{after} completion of the current instruction,
the adversary observes an interrupt latency difference of 3 cycles, depending
on the secret branch condition inside the enclave.
Researchers~\cite{nemesis} have shown how interrupt latency can be practically measured
to precisely reconstruct individual enclave instruction timings
on both high-end and low-end enclave processors.

Using this attack technique, a context can again distinguish two instances of the module with a different password, and hence the addition of interrupts
has {\em weakened} isolation.

A strawman solution to fix the above timing leakage is to modify the implementation of interrupt handling in the processor
to always dispatch interrupt service routines in constant time $\mtt{T}$, i.e.,
regardless of the execution time of the interrupted instruction.
We show in the two examples below, however, that this is a necessary but not sufficient condition.

\begin{example}[Resume-to-end timing]\label{ex:resume-to-end}
Consider the program from Example \ref{ex:latency} executed on a
processor which always dispatches interrupts in constant time $\mtt{T}$.
The attacker schedules an interrupt to arrive in the first
cycle after the \mtt{JZ} instruction, yielding constant interrupt latency $\mtt{T}$.
Next, the context resumes the enclave and measures the time it takes to
let the enclave run to completion \emph{without} further interrupts.
While interrupt latency timing differences are properly masked, the time to
complete enclave execution after resume from the interrupt is 1 cycle for the \texttt{ok} path
and 4 cycles for the \texttt{fail} path.
\end{example}


\begin{example}[Interrupt-counting attack]\label{ex:interrupt-counting}
An alternative way to attack the program from Example~\ref{ex:latency}
even when interrupt latency is constant, is to {\em count} how often
the enclave execution can be interrupted, e.g., by scheduling a new
interrupt 1 cycle after resuming from the previous one. Since interrupts
are handled on instruction boundaries, this allows the attacker to count
the number of instructions executed in the enclave, and hence to distinguish
the two possible control flow paths.
\end{example}


\paragraph{Defining the security of an extension}
The examples above show how a new processor feature (like interrupts) can weaken isolation of an existing isolation mechanism (like enclaved execution), and this is exactly what we want to avoid.
Here we propose and implement a provably secure defense against these attacks.
With this background, our security definition is now obvious.
Given an original system (like Sancus), and an extension of that
system (like interruptible Sancus), that extension is secure if and only if it does not change the contextual
equivalence of enclaves.
Enclaves that are contextually equivalent in the original system must be contextually equivalent in the
extended system and vice versa %
(we shall formalize this as a \emph{full abstraction} property later on). 



\subsection{Secure interruptible Sancus}\label{sec:overview:sl}

Designing an interrupt handling mechanism that is secure according to our definition above is quite subtle.
We illustrate some of the subtleties.
In particular, we provide an intuition on how an appropriate use of padding can handle the various attacks discussed above.
We also discuss how other design aspects are crucial for achieving security.
In this section, we just provide intuition and examples.
The ultimate argument that our design is secure is our proof, discussed later.


\paragraph{Padding} We already discussed that it is insufficient for security to naively pad interrupt latency to make it constant.
We need a padding approach that handles all kinds of attacks, including the example attacks discussed above.

The following padding scheme works (see~\figurename~\ref{padding}). Suppose the
attacker schedules the interrupt to arrive at $t_a$, during the
execution of instruction $I$ in the enclave. Let $\Delta t_1$ be the time needed to
complete execution of $I$. To make sure the attacker can not learn
anything from the interrupt latency, we introduce padding for $\Delta t_{p_1}$
cycles where $\Delta t_{p_1}$ is computed by the interrupt handling logic
such that $\Delta t_1 + \Delta t_{p_1}$ is a constant value $T$. This value $T$
should be chosen as small as possible to avoid wasting unnecessary
time, but must be larger than or equal to the maximal instruction
cycle time \MT\ (to make sure that no negative padding is required,
even when an interrupt arrives right at the start of an instruction with
the maximal cycle time).
This first padding ensures that
an attacker always measures a constant interrupt latency.
\begin{figure}[tb]
\begin{center}
\includegraphics[width=\columnwidth]{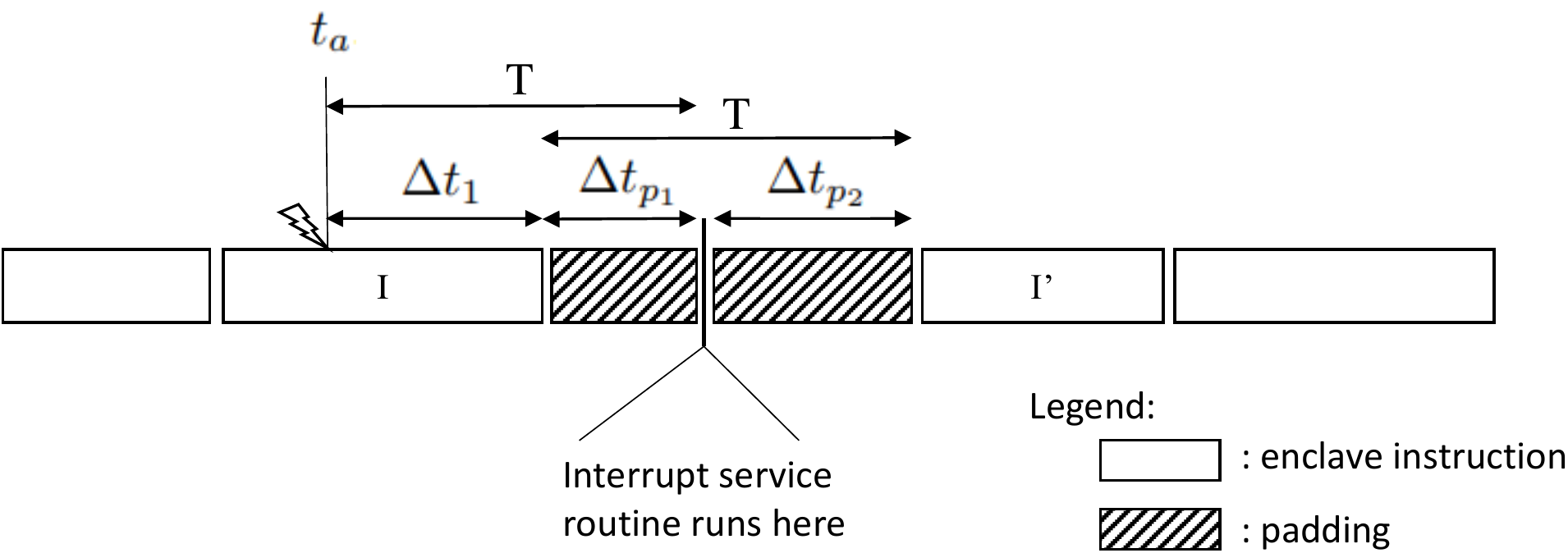}
\caption{\label{padding}The secure padding scheme.
}
\end{center}
\end{figure}

But this alone is not enough, as an attacker can now
measure resume-to-end time as in Example~\ref{ex:resume-to-end}.
Thus, we provide
a second kind of padding. On return from an interrupt, the interrupt
handling logic will pad again for $\Delta t_{p_2}$ cycles, ensuring that
$\Delta t_{p_1} + \Delta t_{p_2}$ is again the constant value~$T$ (i.e., $\Delta t_{p_2} = \Delta t_1$).
This makes sure that the resume-to-end time measured by the attacker
does not depend on the instruction being interrupted.

This description of our padding scheme is still incomplete: it is also important to specify what happens if a new interrupt arrives while the interrupt handling logic is still performing padding because of a previous interrupt. This is important to counter attacks like that of Example~\ref{ex:interrupt-counting}. We refer to the formal description for the complete definition.

Intuitively, the property we get is that (1) an attacker can schedule an interrupt at any time $t_a$ during enclave execution,
(2) that interrupt will always be handled with a constant latency $T$, (3) the resume-to-end time is always exactly the time the enclave still would have needed to complete
execution from point $t_a$ if it had not been interrupted.

This double padding scheme is a main ingredient of our secure interrupt
handling mechanism, but many other aspects of the design are important
for security. We briefly discuss a number of other issues that came
up during the security proof.


\paragraph{Saving execution state on interrupt}
When an enclaved execution is
interrupted, the processor state (contents of the registers) is saved (to
allow resuming the execution once the interrupt is handled) and is cleared (to
avoid leaking confidential register contents to the context).
A straightforward implementation
would be to store the processor state on the enclave stack.
%
However, the proof of our security theorem showed that storing the processor state
in enclave accessible memory is not secure: consider two enclaved modules that monitor the
content of the memory area where processor state is saved, and behave differently on
observing a change in the content of this memory area.
These modules are contextually
equivalent in the absence of interrupts (as the contents of this memory area will never change),
but become distinguishable in the presence of interrupts.
Hence, our design saves processor state in a storage area {\em inaccessible} to software.

\paragraph{No access to unprotected memory from within an enclave}
Most designs of enclaved execution allow an enclave to access unprotected memory (even if
this has already been criticized for security reasons \cite{gruss-sgx-malware}). However, for
a single core processor, interruptibility significantly weakens contextual equivalence for enclaves
that can access unprotected memory. Consider an enclave  $M_1$ that always returns a constant~0,
and an enclave $M_2$ that reads twice from the same unprotected address and returns the difference of
the values read. On a single-core processor without interrupts, $M_2$ will also always return 0, and
hence is indistinguishable from $M_1$. But an interrupt scheduled to occur between the two reads
from $M_2$ can change the value returned by the second read, and hence $M_1$ and $M_2$ become
distinguishable.
Hence, our design forbids enclaves to access unprotected memory.

For similar reasons, our design forbids an interrupt handler to reenter the enclave while it has been
interrupted, and forbids the enclave to directly interact with I/O devices.

Finally, we prevent the interrupt enable bit ($\mtt{GIE}$) in the status register from being changed by software in the enclave, as such changes are unobservable in the original Sancus and they would
be observable once interruptibility is added.

While the security proof is a significant amount of effort, an important benefit of this formalization is that
it forced us to consider all these cases and to think about secure ways of handling them.
We made our design choices to keep model and proof simple, and these choices may seem restrictive.
Section~\ref{sec:discussion} discusses the practical impact of these choices.


%
%



\section{Formalization and security proofs}\label{sec:semantics}
We proceed to formally define two Sancus models, one describing the original, uninterruptible Sancus (\SH,  Sancus-High) and one describing the secure interruptible Sancus (\SL, Sancus-Low).%
\footnote{The {\em high} and {\em low} terminology is inherited from the
field of \emph{secure compilation} of {\em high} source languages to
{\em low} target ones.
Also, for readability we hereafter highlight in \src{blue, sans\text{-}serif} font elements of \SH, in \trg{red, bold} font elements of \SL and in black those that are in common.
}
The two share most of their structure and just differ in the way they deal with interrupts.

Given the semantics of \SH and \SL, we formally show that the two versions of Sancus actually provide the same security guarantees, i.e., the isolation mechanism is not broken by adding a carefully designed interruptible enclaved execution.
Technically, this is done through the \emph{full abstraction} theorem between
\SH and \SL (Theorem~\ref{thm:fa}).
Note that, our theorem guarantees that the \emph{same} program has the \emph{same} security guarantees both in \SH and \SL.

Space limitations prevent us from discussing all the details of our formalization and we refer the reader to \techrep for all the missing details.


\subsection{Setting up our formal framework}
\paragraph{Memory and memory layout}
The memory is modeled as a (finite) function mapping $2^{16}$ locations to bytes $b$.
Given a memory $\M$, we denote the operation of retrieving the byte associated to the location $l$ as $\M(l)$.
On top of that, we define read and write operations on words (i.e., pairs of bytes) and we write $w = b_1 b_0$ to denote that the most significant byte of a word $w$ is $b_1$ and its least significant byte is $b_0$.

The read operation is standard: it retrieves two consecutive bytes from a given memory location $l$ (in a little-endian fashion, as in the MSP430):
\[
    \M [l] \triangleq b_1 b_0 \quad \text{if }  \M(l) = b_0 \land \M(l+1) = b_1
\]
    We define the write operation as follows
\begin{align*}
    (\M [l \mapsto b_1 b_0]) (l') &\triangleq
                                    \begin{cases}
                                        b_0 & \text{if } l' = l \\
                                        b_1 & \text{if } l' = l+1 \\
                                        \M (l') & \text{o.w.}\\
                                    \end{cases}
\end{align*}
Writing $b_0 b_1$ in location $l$ in $\M$ means to build an updated memory mapping $l$ to $b_0$, $l+1$ to $b_1$ and unchanged otherwise.

Note that reads and writes to $l = \mtt{0xFFFF}$ are undefined ($l+1$ would overflow hence it is undefined).
The memory access control explicitly forbids these accesses (see below).
    Also, the write operation deals with unaligned memory accesses (cfr. case $l' = l+1$).
    We faithfully model these aspects to prove that they do not lead to potential attacks.

A memory layout $\L \triangleq \langle
    \mi{ts}, \mi{te},
    \mi{ds}, \mi{de},
    \mi{isr}
    \rangle$
describes how the enclave and the \emph{interrupt service routine} (ISR) are placed in memory and is used to check memory accesses during the execution of each instruction (see below).
The protected code section is denoted by $[\mi{ts}, \mi{te})$,
$[\mi{ds}, \mi{de})$ is the protected data section, and
$\mi{isr}$ is the address of the ISR.
The protected code and data sections do not overlap
and the first address of the protected code section is the single entry point of the enclave.
Finally, we reserve the location $\mtt{0xFFFE}$ to store \emph{the address of} the first instruction to be executed when the CPU starts or when an exception happens%
, reflecting the behavior of MSP430.
Thus, $\mtt{0xFFFE}$ must be outside the enclave sections and different from $\mi{isr}$.

\paragraph{Registers}
There are sixteen $16$-bit registers, three of which $\reg{0}$, $\reg{1}$, $\reg{2}$ have dedicated functions, whereas the others are for general use.
($\reg{3}$ is a constant generator in the MSP430, but we ignore that use in our formalization.)
More precisely, $\reg{0}$ (hereafter denoted as $\rpc$) is the program counter and points to the next instruction to be executed.
Instruction accesses are performed by word and the $\rpc$ is aligned to even addresses.
The register $\reg{1}$ ($\rsp$ hereafter) is the stack pointer and is aligned to even addresses.
Since for the time being we do not model instructions for procedure calls, the only special use of the stack pointer in our model is to store the state while handling an interrupt (see below).
The register $\reg{2}$ ($\rsr$ hereafter) is the status register and contains different
pieces of information encoded as flags.
The most important for us is the fourth bit, called $\mtt{GIE}$, set to $\mtt{1}$ when interrupts are enabled.
Other bits signal, e.g., when an operation produces a carry or when an operation returns zero.

Formally, our \emph{register file} $\R$ is a function that maps each register $\rn{r}$ to a word.
While read operation is standard, the write operation models some invariants enforced by the hardware:
\begin{align*}
    \R [\rn{r}] \triangleq w \text{ if } \R(\rn{r}) = w
\end{align*}
%
%
\begin{align*}
    \R [\rn{r} \mapsto w] \triangleq \lambda [\rn{r'}].
                                    \begin{cases}
                                    w \& \mtt{0xFFFE} & \\
                                        \qquad\quad \text{if } \rn{r'} = \rn{r} \land (\rn{r} = \rpc \lor \rn{r} = \rsp)\\
                                    (w \& \mtt{0xFFF7}) \mid (\R[\rsr] \& \mtt{0x8}) & \\
                                        \qquad\quad \text{if } \rn{r'} = \rn{r} = \rsr \land \moderelPM{\R[\rpc]}\\
                                    w
                                       \  \quad\quad \text{if } \rn{r'} = \rn{r} \land (\rn{r} \neq \rpc \land \rn{r} \neq \rsp)\\
                                    \R [\rn{r'}]
                                        \quad \text{o.w.}
                                    \end{cases}
\end{align*}
More specifically, the least-significant bit of the program counter and of the stack pointer are \emph{always} masked to $0$ (as is also the case in the MSP430), and the $\mtt{GIE}$ bit of the status register is always masked to its previous value when in protected mode (i.e., it cannot be changed when the CPU is running protected code, cf. the discussion in Section~\ref{sec:overview}).
Note that in the definition above we use the relation $\moderel{\R[\rpc]}$, for $\mtt{m} \in \{\mtt{PM}, \mtt{UM}\}$ made precise below: roughly it denotes that the execution is in {\em protected} or in {\em unprotected} mode
(i.e, execution is within, respectively outside the enclave).

\paragraph{I/O Devices}
\emph{I/O devices} are (simplified) \emph{deterministic I/O automata} $\D \triangleq \langle \Delta, \deltainit, \xleadsto{a}_D \rangle$ over a common signature $A$ containing the following actions $a$ (below, $w$ is a word):
\begin{enumerate*}[(i)]
    \item $\epsilon$, a silent, internal action;
    \item $\mi{rd}(w)$, an output action (i.e., read request from the CPU);
    \item $\mi{wr}(w)$, an input action (i.e., write request from the CPU);
    \item $\mi{int?}$, an output action indicating an interrupt is raised.
\end{enumerate*}
The transition function $\delta \xleadsto{a}_D \delta'$ models the device in state $\delta$ performing action $a \in A$ and moving to state $\delta'$, and $\deltainit$ is the initial state.
\paragraph{Contexts, software modules and whole programs}
We call \emph{software module} a memory $\M_M$ containing both protected data and code sections.
A \emph{context} $C$ is a pair $\langle \M_C, \D \rangle$, where $\D$ is a device and $\M_C$ defines the contents of all memory locations \emph{outside}
the protected sections of the layout, thus disjoint from $\M_M$.
Intuitively, the context is the part of the whole program that can be manipulated by an attacker.
Given a context $C$ and a software module $\M_M$, we define a \emph{whole program} as $C[\M_M] = \langle \M_C \uplus \M_M, \D \rangle$.

\paragraph{Instruction set}
We consider a subset of the MSP430 instructions plus our I/O instructions; they are in~\tablename~\ref{tab:op-summary}.
For each instruction the table includes its operands, an informal description of its semantics, its duration and the number of words it occupies in memory.
The durations are used to define the function $\cycles{i}$.
In our model, we let $\MT = 6$, because the longest MSP430 instructions take $6$ cycles (typically those for moving words within memory~\cite{ti-msp430}, none of which are displayed in~\tablename~\ref{tab:op-summary}).
Instructions are stored in the memory $\M$. We use the meta-function $\decode{\M}{l}$ that decodes the contents of the cell(s) starting at location $l$, returning an instruction in the table if any and  $\bot$ otherwise.
\paragraph{Configurations}
Given an I/O device $\D$, the state of the Sancus system is described by configurations of the form:
\[
    c \triangleq \cfg{\B}{\delta}{t}{t_a}{\M}{\R}{\vopc} \in \mb{C}, \quad \text{where}
\]
\begin{enumerate*}[(i)]
    \item $\delta$ is the current state of the I/O device;
    \item $t$ is the current time of the CPU;
    \item $t_a$ is either the arrival time of the last pending interrupt, or $\bot$ if there are none
    (this value may persist across multiple instructions);
    \item $\M$ is the current memory;
    \item $\R$ is the current content of the registers;
    \item $\vopc$
        is the value of the program counter before executing the current instruction;
    \item $\B$ is called the \emph{backup}, is software inaccessible storage space to save enclave state (registers, the old program counter and the remaining padding time) while handling an interrupt raised in protected mode.
\end{enumerate*}
%

The initial configuration for a whole program $C[\M_M] = \langle \M, \D \rangle$ is:
\[
    \initconf{C}{\M_M} \triangleq \cfg{\BBot}{\deltainit}{0}{\bot}{\M}{\R^{\mi{init}}_{\M_C}}{\mtt{0xFFFE}} \text{ where}
\]
\begin{enumerate*}[(i)]
    \item the state of the I/O device $\D$ is $\deltainit$;
    \item the initial value of the clock is $0$ and no interrupt has arrived yet;
    \item the memory is initialized to the whole program memory $\M_C \uplus \M_M$;
    \item all the registers are set to $0$ except that $\rpc$ is set to $\mtt{0xFFFE}$ (the address from which the CPU gets the initial program counter),
    and that $\rsr$ is set to $\mtt{0x8}$ (the register is clear except for the $\mtt{GIE}$ flag);
    \item the previous program counter is also initialized to $\mtt{0xFFFE}$;
    \item the backup is set to $\BBot$ to indicate absence of any backup.
\end{enumerate*}

Dually, $\haltconf$ is the only configuration denoting termination, more specifically it is an opaque and distinguished configuration that indicates graceful termination.

Also, we define \emph{exception handling} configurations, that model what happens on soft reset of the machine (e.g. on a memory access violation, or a halt in protected mode).
On such a soft reset, control returns to the attacker by jumping to the address stored in location $\mtt{0xFFFE}$:
\begin{align*}
    &\exconf{\cfg{\B}{\delta}{t}{t_a}{\M}{\R}{\vopc}} \triangleq\\
    &\qquad\qquad\cfg{\BBot}{\delta}{t}{\bot}{\M}{\R_0 [\rpc \mapsto \M[\mtt{0xFFFE}]]}{\mtt{0xFFFE}}.
\end{align*}
\paragraph{I/O device wrapper}
Since the class of interrupt-based attacks requires a cycle-accurate timer, it is convenient to synchronize the
CPU and the device time by forcing the device to take as many steps as the number of cycles consumed for each instruction by the CPU.
The following ``wrapper'' around the device $\D$ models this synchronization:
%
\[
    \dwrap{k}{\delta, t, t_a}{\delta', t', t'_a}
\]
Assuming that the device was in state $\delta$, at time $t$, and the last pending interrupt was raised at time $t_a$, then this wrapper
defines for $k$ cycles later: the new time $t' = t+k$, the new last pending interrupt time $t'_a$, and the new device state $\delta'$.
When no interrupt has to be handled, $t_a$ and $t'_a$ are $\bot$.

\paragraph{CPU mode and memory access control}\label{subsubsec:mode-mac}
The last two relations used by the main transition systems are the \emph{CPU mode} and the \emph{memory access control}, MAC.
The first tells when a given program counter value, $\vpc$, is an address in the protected code memory ($\mtt{PM}$) or in the unprotected one ($\mtt{UM}$):
\[
    \moderel{\vpc} \text{, with } \mtt{m} \in \{ \mtt{PM}, \mtt{UM} \}
\]
(Also, for simplicity, the relation is lifted to configurations.)

\noindent
The second one
\[
    \macrel{i}{\vopc}{\R}{\B}
\]
holds whenever the instruction $i$ can be executed in a CPU configuration in which the previous program counter is $\vopc$, the registers are $\R$ and the backup is $\B$.
More precisely, it uses the predicate $\mac{f}{rght}{t}$ (see~\tablename~\ref{tab:mac}) that holds whenever from the location $f$ we have the rights $\mtt{rght}$ on location $t$.
The predicate checks that
$(1)$ the code we came from (i.e., that in location $\vopc$) can actually execute the instruction $i$ located at $\R[\rpc]$;
$(2)$ $i$ can be executed in current CPU mode;
and $(3)$ we have the rights to perform $i$ from $\R[\rpc]$, when $i$ is a memory operation.
\begin{table}[tb]
 \centering
        \resizebox{\columnwidth}{!}{
        \begin{tabular}{@{}llcccc@{}}
            & & \multicolumn{4}{c}{$t$}\\
            \cmidrule(l){3-6}
            & \multicolumn{1}{l}{} & Entry Point & Prot. code & Prot. Data & Other \\
            \midrule
            \multicolumn{1}{l|}{\multirow{2}{*}{$f$}} & \multicolumn{1}{l|}{Entry Point/Prot. code} & r-x & r-x & rw- & --x \\
            \multicolumn{1}{l|}{} & \multicolumn{1}{l|}{Other} & --x & --- & --- & rwx \\
            \bottomrule\\
        \end{tabular}}

    \caption{Definition of $\mac{f}{rght}{t}$, where $f$ and $t$ are locations.}\label{tab:mac}
\end{table}
\subsection{\SH: a model of the original Sancus}\label{sec:semantics:sh}
Our models of Sancus are defined by means of two transition systems: a main one and an auxiliary one.
The first system defines the operational semantics of instructions, and relies on
the auxiliary system to specify the behavior upon interrupts.
\paragraph{Main transition system}

The main transition system describes how the \SH configurations evolve during the execution, whose steps are represented by transitions of the following form, where $\D$ is an I/O device and $c, c' \in \mb{C}$:
\[
    \shmain{c}{c'}
\]

\figurename~\ref{fig:main-ts} reports some selected rules among those defining the main transition system.
\newcommand{\bnotspecial}[0]{
        \B \neq \langle \bot, \bot, t_\mi{pad} \rangle
        }
\begin{figure*}[tb]
    \footnotesize
    \begin{mathpar}
        \inferrule*[
            lab={\shrulename{(CPU-Violation-PM)}},
            right={$i = \decode{\M}{\R[\rpc])} \neq \bot$}
        ]
        {
            \bnotspecial\\
            \nmacrel{i}{\vopc}{\R}{\B}
        }
        {
            \shmain{\cfg{\B}{\delta}{t}{t_a}{\M}{\R}{\vopc}}{\exconf{\cfg{\B}{\delta}{t + \cycles{i}}{t_a}{\M}{\R}{\vopc}}}
        }

        \inferrule*[
        lab={\shrulename{(CPU-MovL)}},
        right={$i = \decode{\M}{\R[\rpc])} = \MOVL{r_1}{r_2}$}]
        {
            \bnotspecial\\
            \macrel{i}{\vopc}{\R}{\B}\\
            \R' = \R[\rpc \mapsto \R[\rpc] + 2][\rn{r_2} \mapsto \M[\R[\rn{r_1}]]]\\
            \dwrap{\cycles{i}}{\delta, t, t_a}{\delta', t', t'_a}\\
            \shint{\D}{\cfg{\B}{\delta'}{t'}{t'_a}{\M}{\R'}{\R[\rpc]}}{\cfg{\B'}{\delta''}{t''}{t''_a}{\M'}{\R''}{\R[\rpc]}}
        }
        {
            \shmain{\cfg{\B}{\delta}{t}{t_a}{\M}{\R}{\vopc}}{\cfg{\B'}{\delta''}{t''}{t''_a}{\M'}{\R''}{\R[\rpc]}}
        }

        \inferrule*[
        lab={\shrulename{(CPU-In)}},
        right={$i = \decode{\M}{\R[\rpc]} = \IN{r}$}]
        {
            \bnotspecial\\
            \macrel{i}{\vopc}{\R}{\B}\\
            \delta \xleadsto{\mi{rd(w)}}_D \delta'\\
            \R' = \R[\rpc \mapsto \R[\rpc] + 2][\rn{r} \mapsto w]\\
            \dwrap{\cycles{i}}{\delta', t, t_a}{\delta'', t', t'_a}\\
            \shint{\D}{\cfg{\B}{\delta''}{t'}{t'_a}{\M}{\R'}{\R[\rpc]}}{\cfg{\B'}{\delta'''}{t''}{t''_a}{\M'}{\R''}{\R[\rpc]}}
        }
        {
            \shmain{\cfg{\B}{\delta}{t}{t_a}{\M}{\R}{\vopc}}{\cfg{\B'}{\delta'''}{t''}{t''_a}{\M'}{\R''}{\R[\rpc]}}
        }
    \end{mathpar}

    \caption{Selected rules from the main transition system.}\label{fig:main-ts}
\end{figure*}
The first shows how the model deals with violations in protected mode: if an instruction can not be executed according to the memory-access control relation then a transition to the \emph{exception handling} configuration happens.
Rule~\shrulename{(CPU-MovL)} is for when the current instruction $i$ loads in $\rn{r_2}$ the word in memory at the position pointed by $\rn{r_1}$.
Its first premise checks if the instruction can be executed;
the second one increments the program counter by $2$ and loads in $\rn{r_2}$ the value $\M[\rn{r_1}]$;
the third premise registers in the device that $i$ requires $\cycles{i}$ cycles to complete;
and the last one executes the interrupt logic to check whether an interrupt needs to be handled or not (see comment below).
Another interesting rule is~\shrulename{(CPU-In)} that deals with the case in which the instruction reads a word from the device and puts the result in $\rn{r}$.
Its second premise holds when the device sends the word $w$ to the CPU; the others are similar to those of~\shrulename{(CPU-MovL)}.
\paragraph{Interrupt logic}
The auxiliary transition system for \SH specifies the interrupt logic, and has the form:
\[
\shint{\D}{\cfg{\B}{\delta}{t}{t_a}{\M}{\R}{\vopc}}{\cfg{\B'}{\delta'}{t'}{t'_a}{\M'}{\R'}{\vopc}}.
\]
Since \SH ignores all interrupts,
even in unprotected mode,
the transition system always leaves the
configuration unchanged.

Actually, one could remove the premise with the auxiliary transition system from all the rules defining the semantics of \SH, as it always holds.
However, it is convenient keeping them both to ease the presentation of the transition system of \SL, and for technical reasons, as well.

\subsection{\SL: secure interruptible Sancus}\label{sec:semantics:sl}
We now define the semantics of \SL, the \emph{secure interruptible Sancus}, formalizing the mitigation outlined in Section~\ref{sec:overview}.
We start by describing the main difference with that of \SH, i.e., the way interrupts are handled.

\paragraph{Interrupt logic}
\figurename~\ref{fig:int-ts-sl} shows the relevant rules of the auxiliary transition system describing the interrupt logic of \SL.
\begin{figure*}[tb]
    \footnotesize
    \begin{mathpar}
        \inferrule*[lab={\slrulename{(INT-UM-P)}}]
        {
            \moderelUM{\vopc}\\
            \R[\rsr].\mtt{GIE} = 1\\
            t_a \neq \bot\\
            \R' = \R[\rpc \mapsto \mi{isr}, \rsr \mapsto 0, \rsp \mapsto \R[\rsp]-4]\\
            \M' = \M[\R[\rsp]-2 \mapsto \R[\rpc], \R[\rsp]-4 \mapsto \R[\rsr]]\\
            \dwrap{6}{\delta, t, \bot}{\delta', t', t'_a}
        }
        {
            \slint{\D}{\cfg{\B}{\delta}{t}{t_a}{\M}{\R}{\vopc}}{\cfg{\B}{\delta'}{t'}{t'_a}{\M'}{\R'}{\vopc}}
        }

        \inferrule*[lab={\slrulename{(INT-PM-P)}}]
        {
            k = \MT - (t - t_a)\\
            \moderelPM{\vopc}\\
            \R[\rsr].\mtt{GIE} = \mtt{1}\\
            t_a \neq \bot\\
            \R' = \R_0[\rpc \mapsto \mi{isr}]\\
            \dwrap{6+k}{\delta, t, \bot}{\delta', t', t'_a}\\
            \B' = \langle \R, \vopc, t - t_a \rangle
        }
        {
            \slint{\D}{\cfg{\B}{\delta}{t}{t_a}{\M}{\R}{\vopc}}{\cfg{\B'}{\delta'}{t'}{\bot}{\M}{\R'}{\vopc}}
        }
    \end{mathpar}

\caption{Selected rules for the interrupt logic in \SL.}\label{fig:int-ts-sl}
\end{figure*}
Now interrupts are handled both in unprotected and protected mode, modeled by the rules~\slrulename{(INT-UM-P)} and~\slrulename{(INT-PM-P)}, resp.
For the first case there is the premise $\moderelUM{\vopc}$, for the second $\moderelPM{\vopc}$ (i.e., the mode in which the last instruction was executed).
Both rules have a premise requiring that the $\mtt{GIE}$ bit of the status register is set to $1$ and that an interrupt is on ($t_a \neq \bot$).
(If this is not the case, two further rules, not displayed, just leave the configuration untouched%
, and keep the value of $t_a$ unchanged.)
A premise of~\slrulename{(INT-UM-P)} concerns registers: the program counter gets the entry point of the handler;
the status register gets $0$;
and the top of the stack is moved $4$ positions ahead.
Accordingly, the new memory $\M'$ updates the locations pointed by the relevant elements of the stack with the current program counter and the contents of the status register.
The last premise specifies that this interrupt handling takes $6$ cycles.

The rule~\slrulename{(INT-PM-P)} is more interesting.
Besides assigning the entry point of the handler to the program counter, it computes the padding time for mitigation of interrupt-based timing attacks and saves the backup in $\B'$.
The padding $k$ is then used, causing interrupt handling to take $6 + k$ steps.
Such a padding is needed to implement the first part of the mitigation (see Section~\ref{sec:overview:sl}) and is computed so as to make the dispatching time of interrupts constant.
Note that the padding never gets negative.
When an interrupt arrives in protected mode two cases may arise.
Either $\mtt{GIE}=1$, and the padding is non-negative because the interrupt is  handled at the end of the current instruction; or $\mtt{GIE}=0$, and no padding is needed because the interrupt is handled as soon as $\mtt{GIE}$ becomes 1, which is only possible in unprotected mode.
The backup stores part of the CPU configuration ($\R$ and $\vopc$) and $t_\mi{pad} = t - t_a$.
The value of $t_\mi{pad}$ will then be used as further padding before returning, so fully implementing the mitigation (cf. Section~\ref{sec:overview:sl}).
The register file $\R_0$ is $\{ \rpc \mapsto 0, \rsp \mapsto 0, \rsr \mapsto 0, \reg{3} \mapsto 0, \ldots, \reg{15} \mapsto 0 \}$.

\paragraph{The main transition system}

The rules defining the main transition system of \SL are
those of \SH, with a non-trivial transition system for interrupt logic and mitigation
--- this explains why also \SH rules have the premise $\shint{\D}{\cdot}{\cdot}$ for interrupts.

There are new rules for the new $\RETI$ instruction, shown in~\figurename~\ref{fig:main-ts-sl}.
%
\begin{figure*}[tb]
    \footnotesize
    \begin{mathpar}
        \inferrule*[
            lab={\slrulename{(CPU-Reti)}},
            right={$i = \decode{\M}{\R[\rpc]} = \RETI$}]
            {
                \bnotspecial\\
                \macrel{i}{\vopc}{\R}{\bot}\\
                \R' = \R[\rpc \mapsto \M[\R[\rsp] + 2], \rsr \mapsto \M[\R[\rsp]], \rsp \mapsto \R[\rsp] + 4]\\
                \dwrap{\cycles{i}}{\delta, t, t_a}{\delta', t', t'_a}
            }
            {
                \slmain{\cfg{\BBot}{\delta}{t}{t_a}{\M}{\R}{\vopc}}{\cfg{\BBot}{\delta'}{t'}{t'_a}{\M}{\R'}{\R[\rpc]}}
            }

        \inferrule*[
        lab={\slrulename{(CPU-Reti-Chain)}},
        right={$i = \decode{\M}{\R[\rpc]} = \RETI$}]
        {
            \bnotspecial\\
            \B \neq \BBot \\
            \dwrap{\cycles{i}}{\delta, t, t_a}{\delta', t', t'_a}\\
            \R[\rsr.\mtt{GIE}] = \mtt{1}\\
            t'_a \neq \bot\\
            \slint{\D}{\cfg{\B}{\delta'}{t'}{t'_a}{\M}{\R}{\R[\rpc]}}{\cfg{\B}{\delta''}{t''}{t''_a}{\M'}{\R'}{\R[\rpc]}}
        }
        {
            \slmain{\cfg{\B}{\delta}{t}{t_a}{\M}{\R}{\vopc}}{\cfg{\B}{\delta''}{t''}{t''_a}{\M'}{\R'}{\R[\rpc]}}
        }

        \inferrule*[
        lab={\slrulename{(CPU-Reti-PrePad)}},
        right={$i = \decode{\M}{\R[\rpc]} = \RETI$}]
        {
            \B \neq \langle \bot, \bot, t_\mi{pad} \rangle\\
            \macrel{i}{\vopc}{\R}{\B}\\
            \B \neq \BBot \\
            \dwrap{\cycles{i}}{\delta, t, t_a}{\delta', t', t'_a}\\
            (\R[\rsr.\mtt{GIE}] = \mtt{0} \ \lor\ t'_a = \bot)\\
        }
        {
            \slmain{\cfg{\B}{\delta}{t}{t_a}{\M}{\R}{\vopc}}
            {
                \cfg{\langle \bot, \bot, \B.t_\mi{pad} \rangle}{\delta'}{t'}{t'_a}{\M}{\B.\R}{\B.\vopc}
            }
        }

        \inferrule*[lab={\slrulename{(CPU-Reti-Pad)}}]
        {
            \B = \langle \bot, \bot, t_\mi{pad} \rangle \\
            \dwrap{t_\mi{pad}}{\delta, t, t_a}{\delta', t', t'_a}\\
            \slint{\D}{\cfg{\BBot}{\delta'}{t'}{t'_a}{\M}{\R}{\vopc}}{\cfg{\B'}{\delta''}{t''}{t''_a}{\M}{\R'}{\vopc}}\\
        }
        {
            \slmain{\cfg{\B}{\delta}{t}{t_a}{\M}{\R}{\vopc}}{\cfg{\B'}{\delta''}{t''}{t''_a}{\M}{\R'}{\vopc}}
        }
    \end{mathpar}

    \caption{Some rules from the operational semantics of \SL.}\label{fig:main-ts-sl}
\end{figure*}
Rule~\slrulename{(CPU-Reti)} deals with a return from an interrupt that was handled in unprotected mode, i.e., when $i = \decode{\M}{\R[\rpc]} = \RETI$ and there is no backup.
Its first premise checks that the $\RETI$ instruction is indeed permitted.
The second one requires that the program counter is set to the contents of the memory location pointed by the second element from the top of the stack (that grows downwards); that the status register is set to the contents of the memory location pointed by the top of the stack; and that
two words are popped
from the stack.
Finally, the third one registers that $\cycles{i}$ steps are needed to complete this task.
Rule~\slrulename{(CPU-Reti-Chain)} executes the interrupt handler in unprotected
mode when the CPU discovers that another interrupt arrived, while returning
from a handler whose interrupt was raised in protected mode (via the interrupt logic).
%
The most interesting rules are the last two.
They deal with the case in which the CPU is returning from the handling of an interrupt raised in protected mode, but no new interrupt arrived afterwards (or the $\mtt{GIE}$ bit is off, cf. fourth premise of rule~\slrulename{(CPU-Reti-PrePad)}).
First, rule~\slrulename{(CPU-Reti-PrePad)} restores registers and $\vopc$ from the backup $\B$, then rule~\slrulename{(CPU-Reti-Pad)} 
 (which is the only one applicable after~\slrulename{(CPU-Reti-PrePad)})
applies the remaining padding (recorded in the backup) to rule out resume-to-end timing attacks (note that this last padding is interruptible, as witnessed by the last premise).
We model the mechanism of restoring registers, $\vopc$ and of applying the remaining padding with two rules instead of just one for technical reasons (see \techrep for details).
Note that this last padding is applied even if the configuration reached through rule~\slrulename{(CPU-Reti-PrePad)} is in unprotected mode (i.e., the interrupted instruction was a jump out of protected mode).
Indeed, if it was not the case, the attacker would be able to discover the value of the padding applied \emph{before} the interrupt service routine.

\subsection{Security theorem}\label{sec:semantics:fa}
Our security theorem states that what an attacker can learn from an enclave is exactly the same before and after adding the support for interrupts.
Technically, we show that the semantics of \SL is \emph{fully abstract} w.r.t.\ the semantics of \SH, i.e., all the attacks that can be carried out in \SL can also be carried out in \SH, and viceversa.
    Even though the technical details are specific to our case study, the security definition applies also to other architectures.
Before stating the full abstraction theorem and giving the sketch of its proof, we introduce some further notations.

Recall that a whole program $C[\M_M]$ consists
of a module $\M_M$ and a context $C = \langle \M_C, \D \rangle$, where
$\M_C$ contains the unprotected program and data and $\D$ is the I/O device.

Let $\sconv{C[\M_M]}$ denote a \emph{converging computation in \SH}, i.e., a sequence of transitions of the whole program that reaches the halting configuration from the initial one.
Also, let two software modules $\M_M$ and $\M_{M'}$ be \emph{contextually equivalent in \SH}, written $\M_M \seq \M_{M'}$, if and only if for all contexts $C$, $\sconv{C[\M_M]} \iff \sconv{C[\M_{M'}]}$.
Similarly, we define $\tconv{C[\M_M]}$ and $\M_M \teq \M_{M'}$ for \SL.
Roughly, the notion of contextual equivalence formalizes the intuitive notion of \emph{indistinguishability}:
two modules are contextually equivalent if they behave in the same way under any attacker (i.e., context).
Due to the quantification over {\em all}
contexts, it suffices to consider just terminating and non-terminating
executions as distinguishable, since
any other distinction can be reduced to it.
We can state the theorem that guarantees the absence of interrupt-based attacks:
\begin{restatable}[Full abstraction]{thm}{fullabstraction}\label{thm:fa}\hfill\newline
    $\forall \M_M, \M_{M'}.\ (\M_M \seq \M_{M'} \iff \M_M \teq \M_{M'})$.
\end{restatable}

\noindent
First we prove
$\M_M \teq \M_{M'} \Rightarrow \M_M \seq \M_{M'}$ and then
$\M_M \seq \M_{M'} \Rightarrow \M_M \teq \M_{M'}$.
Below we only intuitively describe the proof steps (all the details are in \techrep).


\paragraph{Proof sketch for $\M_M \teq \M_{M'} \Rightarrow \M_M \seq \M_{M'}$}
Since programs in \SH behave like those in \SL with no interrupts, proving this implication is not too hard.
It suffices to introduce the notion of \emph{interrupt-less} context $\nI{C}$ for \SL
that behaves as $C$, but never raises  interrupts.
The thesis follows because an enclave hosted in a interrupt-less context terminates in \SL whenever it does in \SH, as interrupt-less contexts are a strict subset of all the contexts.

%

\paragraph{Proof sketch for $\M_M \seq \M_{M'} \Rightarrow \M_M \teq \M_{M'}$}
%
%
We first introduce the notion of observable behavior, in terms of the traces that $C[\M_M]$ can perform according to the \SL semantics.
Traces are built using three observables: $(i)$  $\cnv$ denotes that the computation halts; $(ii)$ $\jmpin{\R}$ denotes that the CPU enters the protected mode, where $\R$ are the observed registers and $(iii)$ $\jmpout{\dt}{\R}$ denotes the exit from protected mode with observed registers $\R$ and with $\dt$ representing the end-to-end time measured by an attacker for code running in protected mode.

The proof then follows the steps
in~\figurename~\ref{fig:strategy}, where $\M_M \ttreq \M_{M'}$ means that $\M_M$ and $\M_{M'}$ have the same 
traces.
Implication $(i)$ shows that the attacker in \SL at most observes \emph{as much as} traces say; implication $(ii)$ shows that the attacker in \SH is \emph{at least as powerful as} described by traces;
finally implication $(iii)$ is our thesis that follows by transitivity.
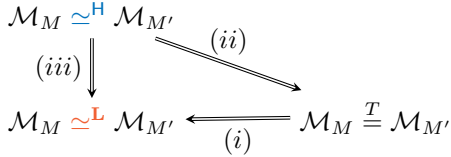
\begin{figure}[tb]
    \centering
    \begin{tikzpicture}
        \matrix (m) [matrix of math nodes,row sep=2em,column sep=4em,minimum width=2em] {
        \M_M \seq \M_{M'} & \\
        \M_M \teq \M_{M'} & \M_M \ttreq \M_{M'} \\};
        \path[-stealth]
          (m-2-2) edge [double] node [below] {$(i)$} (m-2-1)
          (m-1-1) edge [double] node [above] {$(ii)$} (m-2-2)
          (m-1-1) edge [double] node [left] {$(iii)$} (m-2-1);
    \end{tikzpicture}
    \caption{The steps for proving preservation of behavior.}\label{fig:strategy}
\end{figure}
The proof of $(i)\ \M_M \ttreq \M_{M'} \Rightarrow \M_M \teq \M_{M'}$ roughly goes as follows.
First the mitigation is shown to guarantee
that the behavior of the context
(in unprotected mode) does not depend on the behavior of the enclave (in
protected mode) and vice versa
(\isolationlemmata).
The thesis follows because if $\M_M \ttreq \M_{M'}$ and $C[\M_M]$ has a trace
$\btr$, then also $C[\M_M']$ has the same trace $\btr$.

The proof of $(ii)\ \M_M \seq \M_{M'} \Rightarrow \M_M \ttreq \M_{M'}$
is by contraposition: if two modules have different traces, there exists a context that distinguishes them,
and we build such a context through a \emph{backtranslation}.
Because of the strong limitations -- for instance because only 64KB of memory is available -- building such a context in unprotected memory only is infeasible and the strong attacker model that enclaved execution is built for is actually helpful here.
The backtranslation defines and uses both the unprotected memory (\algomem), and the I/O device, which has unrestricted memory (\algodev).
Very roughly, the idea is to take a trace of $\M_M$ and one of $\M_{M'}$ that differ for one observable, and build a context $C$ such that $\M_M$ converges and $\M_{M'}$ does not, so contradicting the hypothesis $\M_M \seq \M_{M'}$.



\section{Implementation and evaluation}\label{sec:implementation}

\newcommand{\ctr}[1]{\ensuremath{\mathtt{C}_\mathtt{#1}}}

We provide a full implementation of our approach based on the Sancus~\cite{sancus} architecture which, in turn, is based on the openMSP430, an open source implementation of the TI MSP430 ISA.
Our implementation can be divided in two parts.
First, we adapted the execution unit's state machine to add padding cycles whenever an interrupt happens in protected mode and when we return from such interrupts.
Second, we added a protected storage area corresponding to $\B$.


\paragraph{Cycle padding}
\label{sec:cycle-padding}


To implement cycle padding, we added three counters to the processor's frontend.
The first, \ctr{reti\_nxt}, tracks the number of cycles
to be padded on the next \RETI.
Whenever an \emph{interrupt request} (IRQ) occurs,
this counter is initialized to zero and is subsequently incremented every cycle until the current instruction
completes.
Thus, at the end of an instruction, this counter holds $t - t_a$, which corresponds to $t_{\mi{pad}}$ in $\B$ (cf. the \slrulename{(INT-PM-P)} rule in~\figurename~\ref{fig:int-ts-sl}).

The second counter, \ctr{irq}, holds the number of cycles that needs to be padded when an IRQ occurs.
It is initialized to $\MT - \ctr{reti\_next}$
($\MT$ is $6$ in our case)
when the instruction during which an IRQ occurred finishes execution.
That is, it holds the value $k$ from rule \slrulename{(INT-PM-P)} in~\figurename~\ref{fig:int-ts-sl} after the instruction finishes.
From this point on, the counter is decremented every cycle and the execution unit's state machine
is kept in a wait state until the counter reaches zero.
Only then is it allowed to progress and start handling the IRQ.

Lastly, a third counter, \ctr{reti}, is added that holds the number of cycles that needs to be padded for the current \RETI{} instruction.
Whenever a \RETI{} is executed while handling an IRQ from protected mode, this counter is initialized with the value of \ctr{reti\_nxt}.
Then, after restoring the processor state from $\B$ (see Section~\ref{sec:save-restore-state}), this counter is decremented every cycle until it reaches zero.
After these padding cycles, the next instruction is fetched, from $\R[\rpc]$ restored from $\B$, and executed.
Note that these padding cycles behave as any $t_{\mi{pad}}$-cycle instruction from the perspective of the padding logic.
That is, they can be interrupted and, hence, padded as well.
This is the reason why we need two counters to hold padding information for \RETI: \ctr{reti} is used to pad the current \RETI{} instruction and \ctr{reti\_nxt} is used -- concurrently, if an IRQ occurs -- to count $t_{\mi{pad}}$ for the \emph{next} \RETI.

\paragraph{Saving and restoring processor state}
\label{sec:save-restore-state}

Whenever an IRQ in protected mode occurs, the processor's register state needs to be saved in a location inaccessible from software.
Our current implementation uses a shadow register file to this end.
We duplicate all registers $\reg{0},\ldots,\reg{15}$ (except $\reg{3}$, the constant generator, which does not store state).
On an IRQ, all registers are first copied to the shadow register file and then cleared.
When a subsequent \RETI{} is executed, registers are restored from their copies.
For the other values in $\B$, $\vopc$ is handled the same as registers, and $t_{\mi{pad}}$ is saved from \ctr{reti\_nxt} and restored to \ctr{reti}, as explained in Section~\ref{sec:cycle-padding}.
Besides the values in $\B$, we add a single bit to indicate if we are currently handling an IRQ from protected mode, \mbox{allowing us to test if $\B  \neq \bot$.}

The current implementation allows to save or restore the processor state in a single cycle at the cost of approximately doubling the size of the register file.
If this increase in area is unacceptable, the state could be stored in a protected memory area.
Implementing this directly in hardware would increase the number of cycles needed to save and restore a state to one cycle per register.
Of course, one should make sure that this memory area is inaccessible from software by adapting the
memory access control logic of the processor accordingly.


\paragraph{Evaluation}
\label{sec:evaluation}

To evaluate the performance impact of our implementation, we only need to quantify the overhead on handling interrupts and returning from them,
as an uninterrupted flow of instructions is not impacted by our design.

When an IRQ occurs, as well as when the subsequent \RETI{} is executed, there is a maximum of $\MT$
padding cycles executed.
This variable part of the overhead is thus bounded by
$\MT$
cycles for both cases.
The fixed part -- saving and restoring the processor's state -- turns out to be 0 in our current implementation: since the fetch unit's state machine
needs at least one extra cycle to do a jump in both cases, copying the state is done during this cycle and
causes no extra overhead.
Of course, if the register state is stored in memory, as described in Section~\ref{sec:save-restore-state}, the fixed overhead grows accordingly.


To evaluate the impact on area, we synthesized our implementation on a Xilinx XC6SLX25 Spartan-6 FPGA with a speed grade of $-2$ using Xilinx ISE Design Suite optimizing for area.
The baseline is an unaltered Sancus 2.0 core configured with support for a single protected module and 64-bit keys
for remote attestation.
The unaltered core could be synthesized using 1239 slice registers and 2712 slice LUTs.
Adding support for saving and restoring the processor state increases the area to 1488 slice registers and 2849 slice LUTs and the implementation of cycle padding further increases it to 1499 slice registers and 2854 slice LUTs.
It is clear that the largest part of the overhead comes from saving the processor state which is necessary for any implementation of secure interrupts and can be optimized as discussed in Section~\ref{sec:save-restore-state}.
The implementation of cycle padding, on the other hand, does not have a significant impact on the processor's area.


\section{Discussion}\label{sec:discussion}

\subsection{On the use of full abstraction a security objective}
The security guarantee that our approach offers is quite strong: an attack is possible in \SH if and only if it is possible at \SL.
Full abstraction fits naturally with our goal, because isolation is defined in terms of contextual equivalence, and full abstraction
specifies that contextual equivalence is preserved and reflected.

The \emph{if}-part, namely preservation, guarantees that extending \SH with interrupts opens no new vulnerabilities.
Reflection, i.e., the \emph{only if}-part is needed because otherwise two enclaves that are distinguishable in \SH \emph{become} indistinguishable in \SL.
Although this mainly concerns functionality and not security,
a problem emerges: adding interrupts is not fully ``backwards compatible.''
Indeed, reflection rules out mechanisms that while closing the interrupt side-channels also close other channels.
We believe the situation is very similar for other extensions: adding caches, pipelining, etc.\ should not strengthen existing isolation mechanisms either.

Actually, full abstraction enables us to take the security guarantees of \SH as the specification of the isolation required after an extension is added.

An alternative approach to full abstraction would be to require (a non interactive version of) robust preservation of timing-sensitive non-interference~\cite{abate2018exploring}.
This can also guarantee resistance against the example attacks in Section~\ref{sec:overview}.
However, this approach offers a strictly weaker guarantee: our full abstraction result implies that timing-sensitive non-interference properties of \SH programs are preserved in \SL, as far as non-interference takes as secret the whole enclave, i.e., its memory and code, and the initial state, as well.
In addition, full abstraction implies that isolation properties that rely on code confidentiality are preserved, and this matters for enclave systems that guarantee code confidentiality, like the Soteria system~\cite{soteria}.
An advantage however might be that robust preservation of timing-sensitive non-interference might be easier to prove.

In case full abstraction is considered too strong as a security criterion,
it is possible to selectively weaken it by modifying \SH. For instance,
to specify that code confidentiality is not important, one can modify
\SH to allow contexts to read the code of an enclave.

\subsection{The impact of our simplifcations}
The model and implementation we discussed in this paper make several simplifying assumptions.
A first important observation that we want to make is that some simplifications of our model with respect to our implementation
are straightforward to remove.
For instance, supporting more MSP430 instructions in our model would not affect the strong security guarantees offered by our approach, and only requires straightforward, yet tedious technical work.

However, there are also other assumptions that are more essential, and removing these would require additional research.
Here, we discuss the impact of these assumptions on the applicability of our results to real systems.

First, we scoped our work to only consider ``small'' microprocessors.
The essential assumption our work relies on is that the timing of individual instructions is
predictable (as shown, e.g., in~\tablename~\ref{tab:op-summary} for the MSP430).
This is typically only true of small microprocessors.
As soon as a processor implements performance enhancing features like caching or speculation, the timing
of an individual instruction will be variable, e.g., a load will be faster if can be served from the cache. Our model and proof do not apply to such more advanced processors.
However, we do believe that the padding countermeasure that we proved to be secure on simple processors is a very
good {\em heuristic} countermeasure, also for more advanced processors.
It has been shown that for instance interrupt-latency attacks are relevant for modern
Intel Core processors supporting SGX enclaves \cite{nemesis}.
Interrupt latency is not deterministic on these processors, but is instead a complex function of the
micro-architectural state at the point of interruption, and it is hard to determine
an upper bound on the maximal latency that could be observed.
Still, padding to a fixed length on interrupt and complementary padding on resume will significantly
raise the bar for interrupt latency attacks.
We are aware that it would be very hard, if not impossible at all, to carry over to these settings the strong security guarantees offered by full abstraction for ``small'' microprocessors.
Consider for instance the leaks made possible by the persistent micro-architectural state that we do not model
in this paper.
However, implementing our countermeasure will likely make attacks harder also in high-end microprocessors.

Second, our model made some simplifying assumptions about the enclave-based
isolation mechanism.
We did not model support for cryptographic operations and for attestation.
This  means that we
assume that the loading and initialization of an enclave
can be done as securely in \SL as it can be done in \SH.
Our choice separates concerns, and it is independent of the security criterion adopted.
Modelling both memory access control and cryptography would only increase the complexity of the model, as
two security mechanisms rather than one would be in order.
Also their interactions should be considered to prevent, e.g., leaks of cryptographic keys unveiling secrets protected by
memory access control, and viceversa.
Also, we assumed the simple setting where only a single
enclave is supported.
We believe these simplifications are acceptable, as they
reduce the complexity of the model significantly, and as none of the known
interrupt-driven attacks relies on these features.
It is also important to
emphasize that these are model-limitations, and that an implementation can
easily support attestation and multiple enclaves. However, for implementations
that do this, our
current
proof does not rule out the presence of attacks that
rely on these features.

A more fundamental limitation of the model is that it forbids reentering an
enclave that has been interrupted, via $\vdash_{\mi{mac}}$.
Allowing reentrancy essentially causes the same complications as allowing multi-threaded enclaves, and these
are substantial complications that also lead to new kinds of attacks \cite{asyncshock}.
We leave investigation of these issues to future work.

Third, our model and implementation make
other
simplifications that we
believe to be non-essential and that could be removed with additional work
but without providing important new insights.
For instance, we assumed that enclaves have no read/write access to untrusted
memory. A straightforward alternative is to allow these accesses, but to also
make them observable to the untrusted context in \SH. Essentially, this alternative
forces the enclave developer to be aware of the fact that accessing untrusted
memory is an interaction with the attacker. A better alternative (putting
less security responsibility with the enclave developer) is to rely on a
trusted run-time that can access unprotected memory to
copy in/out parameters and results, and then turn off access to unprotected memory
before calling enclaved code.
This is very similar to how
Supervisor Mode Access Prevention prevents the kernel from the security risks of accessing user memory.
Our model could easily be extended to model such a trusted run-time by considering memory copied in/out as
a large CPU register.
It is important to emphasize however that the implementation of such 
trusted enclave runtime environments has been shown to be error-prone \cite{tale2worlds}.

Another such non-essential limitation is the fact that we do not support
nested interrupts, or interrupt priority. It is straightforward to extend
our model with the possibility of multiple pending interrupts and a policy
to select which of these pending interrupts to handle. One only has to
take care that the interrupt arrival time used to compute padding is the
arrival time of the interrupt that will be handled first.

In summary, to provide hard mathematical security guarantees, one often abstracts from some details and provable security only provides assurance to the extent that the assumptions made are valid and the simplifications non-essential.
The discussion above shows that this is the case for a relevant class of attacks and systems, and hence that
our countermeasure for these attacks is well-designed.
Since there is no 100\% security, attacks remain possible for more complex
systems (e.g. including caches and speculation), or for more powerful
attackers (e.g. with physical access to the system).




\section{Related Work}\label{sec:related}
    
    Our work is motivated by the recent wave of software-based side-channel attacks and controlled-channel attacks that rely on architectural or micro-architectural processor features.
    The area is too large to survey here, but good recent surveys include Ge et al.~\cite{cache-attacks} for timing attacks, Gruss' PhD thesis~\cite{gruss-thesis} for software-based microarchitectural attacks before Spectre/Meltdown, and~\cite{transient} for transient execution based attacks.
    The attacks most relevant to this paper are the pure interrupt-based attacks. Van Bulck et al.~\cite{nemesis} were the first to show how just measuring interrupt latency can be a powerful attack vector against both high-end enclaved execution systems like Intel SGX, and against low-end systems like the Sancus system that we based our work on.
    Independently, He et al.~\cite{sgxlinger} developed a similar attack for Intel SGX.

    There is an extensive body of work on defenses against software-based side-channel attacks.
    The three surveys mentioned above (\cite{cache-attacks,gruss-thesis,transient}) also survey defenses, including both software-based defenses like the constant-time programming model and hardware-based defenses such as cache-partitioning.
    To the best of our knowledge, our work proposes the first defense specifically designed and proved to protect against pure interrupt-based side-channel attacks.
    De Clerck et al.~\cite{ruan} have proposed a design for secure interruptibility of enclaved execution, but they have not considered side-channels -- their main concern is to make sure that there are no direct leaks of, e.g., register contents on interrupts.
    Most closely related to ours is the work on SecVerilog~\cite{secverilog} that also aims for formal assurances.
    To guarantee timing-sensitive non-interference properties, SecVerilog uses a security-typed hardware description language.
    However, this approach has not yet been applied to the issue of interrupt-based attacks.
    Similarly, Zagieboylo et al.~\cite{zagieboylo2019using} describe an ISA with information-flow labels and use it to guarantee timing-insensitive information flow at the architectural level.


    An alternative approach to interruptible secure remote computation is pursued by VRASED \cite{vrased}. In contrast to enclaved execution, their design only relies on memory access control for the attestation key, not for the software modules being attested. They prove that a carefully designed hardware/software co-design can securely do remote
    attestation.

    Our security criterion is directly influenced by a long line of work that considers {\em full abstraction} as a criterion for secure compilation. The idea was first coined
    by Abadi~\cite{abadi1999protection}, and has been applied in many settings, including compilation to JavaScript~\cite{fournet2013fully}, various intermediate compiler
    passes~\cite{ahmed2008typed,ahmed2011equivalence}, and compilation to platforms that support enclaved execution~\cite{agten2012secure,patrignani2015fully,patrignani2015secure}.
    But none of these works consider timing-sensitivity or interrupts: they study compilations higher up the software stack than what we consider in this paper.
    Patrignani et al.~\cite{patrignani2019formal} have provided a good survey of this entire line of work on secure compilation.


    \section{Conclusions and future work}\label{sec:conclusion}

    We have proposed an approach to formally assure that extending a microprocessor with
    a new feature does not weaken the isolation mechanisms that the processor offers.
    We have shown that the approach is applicable to an IoT-scale microprocessor, by showing
    how to design interruptible enclaved execution that is as secure as uninterruptible enclaved
    execution.
    Despite this successful case study, some limitations of the approach remain, and we
    plan to address them in future.

    First, as discussed in Section~\ref{sec:discussion}, our approach currently applies only to
    ``small'' micro-processors for which we can define a cycle-accurate operational semantics.
    While this obviously makes it possible to rigorously reason about timing-based side-channels, it is
    also difficult to scale to larger processors. To handle larger processors, we need models that
    can abstract away many details of the processor implementation, yet keeping enough detail to
    model relevant micro-architectural attacks. A very recent and promising example of such a
    model was proposed by Disselkoen et al.~\cite{CodeThatNeverRan}.
    An interesting avenue
    for future work is to consider such models for our approach instead of the cycle-accurate models.

    Second, the security criterion we proposed is binary: an extension is either secure, or it is
    not. The criterion does not distinguish {\em low bandwidth} side-channels from {\em high-bandwidth}
    side-channels. An important challenge for future work is to introduce some kind of {\em quantification}
    of the weakening of security, so that it becomes feasible to allow the introduction of some bounded
    amount of leakage.


\section*{Acknowledgements}
    We would like to thank the anonymous referees and the paper shepherd for their insightful comments and detailed suggestions that helped to greatly improve our presentation.
Matteo Busi and Pierpaolo Degano have been partially supported by the University of Pisa project PRA\_2018\_66
\emph{DECLware: Declarative methodologies for designing and deploying applications}.
This research is partially funded by the Research Fund KU Leuven, by the Agency for Innovation and Entrepreneurship (Flanders), and by a gift from Intel Corporation.
Jo Van Bulck is supported by a grant of the Research Foundation -- Flanders~(FWO).
Letterio Galletta has been partially supported by EU Horizon 2020 project
No 830892 \emph{SPARTA} and by MIUR project PRIN 2017FTXR7S \emph{IT MATTERS} (Methods and Tools for Trustworthy Smart Systems).


\bibliographystyle{IEEEtran}
\bibliography{paper}

\begin{thebibliography}{10}
\providecommand{\url}[1]{#1}
\csname url@samestyle\endcsname
\providecommand{\newblock}{\relax}
\providecommand{\bibinfo}[2]{#2}
\providecommand{\BIBentrySTDinterwordspacing}{\spaceskip=0pt\relax}
\providecommand{\BIBentryALTinterwordstretchfactor}{4}
\providecommand{\BIBentryALTinterwordspacing}{\spaceskip=\fontdimen2\font plus
\BIBentryALTinterwordstretchfactor\fontdimen3\font minus
  \fontdimen4\font\relax}
\providecommand{\BIBforeignlanguage}[2]{{%
\expandafter\ifx\csname l@#1\endcsname\relax
\typeout{** WARNING: IEEEtran.bst: No hyphenation pattern has been}%
\typeout{** loaded for the language `#1'. Using the pattern for}%
\typeout{** the default language instead.}%
\else
\language=\csname l@#1\endcsname
\fi
#2}}
\providecommand{\BIBdecl}{\relax}
\BIBdecl

\bibitem{csf}
M.~Busi, J.~Noorman, J.~V. Bulck, L.~Galletta, P.~Degano, J.~T. M\"uhlberg, and
  F.~Piessens, ``Provably secure isolation for interruptible enclaved execution
  on small microprocessors,'' in \emph{{To appear at the 33rd IEEE Computer
  Security Foundations Symposium (CSF'20)}}, 2020.

\bibitem{sgx}
F.~McKeen, I.~Alexandrovich, A.~Berenzon, C.~V. Rozas, H.~Shafi, V.~Shanbhogue,
  and U.~R. Savagaonkar, ``Innovative instructions and software model for
  isolated execution,'' in \emph{{HASP} 2013, The Second Workshop on Hardware
  and Architectural Support for Security and Privacy, Tel-Aviv, Israel, June
  23-24, 2013}, R.~B. Lee and W.~Shi, Eds.\hskip 1em plus 0.5em minus
  0.4em\relax {ACM}, 2013, p.~10.

\bibitem{rowhammer}
Y.~Kim, R.~Daly, J.~Kim, C.~Fallin, J.~Lee, D.~Lee, C.~Wilkerson, K.~Lai, and
  O.~Mutlu, ``Flipping bits in memory without accessing them: An experimental
  study of {DRAM} disturbance errors,'' in \emph{{ACM/IEEE} 41st International
  Symposium on Computer Architecture, {ISCA} 2014, Minneapolis, MN, USA, June
  14-18, 2014}.\hskip 1em plus 0.5em minus 0.4em\relax {IEEE} Computer Society,
  2014, pp. 361--372.

\bibitem{clockscrew}
\BIBentryALTinterwordspacing
A.~Tang, S.~Sethumadhavan, and S.~J. Stolfo, ``{CLKSCREW:} exposing the perils
  of security-oblivious energy management,'' in \emph{26th {USENIX} Security
  Symposium, {USENIX} Security 2017, Vancouver, BC, Canada, August 16-18,
  2017.}, E.~Kirda and T.~Ristenpart, Eds.\hskip 1em plus 0.5em minus
  0.4em\relax {USENIX} Association, 2017, pp. 1057--1074. [Online]. Available:
  \url{https://www.usenix.org/conference/usenixsecurity17/technical-sessions/presentation/tang}
\BIBentrySTDinterwordspacing

\bibitem{Murdock2019plundervolt}
K.~Murdock, D.~Oswald, F.~D. Garcia, J.~Van~Bulck, D.~Gruss, and F.~Piessens,
  ``{Plundervolt}: Software-based fault injection attacks against intel sgx,''
  in \emph{{Proceedings of the 41st IEEE Symposium on Security and Privacy
  (S\&P'20)}}, 2020.

\bibitem{cache-attacks}
Q.~Ge, Y.~Yarom, D.~Cock, and G.~Heiser, ``A survey of microarchitectural
  timing attacks and countermeasures on contemporary hardware,'' \emph{J.
  Cryptographic Engineering}, vol.~8, no.~1, pp. 1--27, 2018.

\bibitem{meltdown}
\BIBentryALTinterwordspacing
M.~Lipp, M.~Schwarz, D.~Gruss, T.~Prescher, W.~Haas, A.~Fogh, J.~Horn,
  S.~Mangard, P.~Kocher, D.~Genkin, Y.~Yarom, and M.~Hamburg, ``Meltdown:
  Reading kernel memory from user space,'' in \emph{27th {USENIX} Security
  Symposium, {USENIX} Security 2018, Baltimore, MD, USA, August 15-17, 2018.},
  W.~Enck and A.~P. Felt, Eds.\hskip 1em plus 0.5em minus 0.4em\relax {USENIX}
  Association, 2018, pp. 973--990. [Online]. Available:
  \url{https://www.usenix.org/conference/usenixsecurity18/presentation/lipp}
\BIBentrySTDinterwordspacing

\bibitem{spectre}
P.~Kocher, J.~Horn, A.~Fogh, , D.~Genkin, D.~Gruss, W.~Haas, M.~Hamburg,
  M.~Lipp, S.~Mangard, T.~Prescher, M.~Schwarz, and Y.~Yarom, ``Spectre
  attacks: Exploiting speculative execution,'' in \emph{40th IEEE Symposium on
  Security and Privacy (S\&P'19)}, 2019.

\bibitem{foreshadow}
\BIBentryALTinterwordspacing
J.~V. Bulck, M.~Minkin, O.~Weisse, D.~Genkin, B.~Kasikci, F.~Piessens,
  M.~Silberstein, T.~F. Wenisch, Y.~Yarom, and R.~Strackx, ``Foreshadow:
  Extracting the keys to the intel {SGX} kingdom with transient out-of-order
  execution,'' in \emph{27th {USENIX} Security Symposium, {USENIX} Security
  2018, Baltimore, MD, USA, August 15-17, 2018.}, W.~Enck and A.~P. Felt,
  Eds.\hskip 1em plus 0.5em minus 0.4em\relax {USENIX} Association, 2018, pp.
  991--1008. [Online]. Available:
  \url{https://www.usenix.org/conference/usenixsecurity18/presentation/bulck}
\BIBentrySTDinterwordspacing

\bibitem{nemesis}
\BIBentryALTinterwordspacing
J.~Van~Bulck, F.~Piessens, and R.~Strackx, ``Nemesis: Studying
  microarchitectural timing leaks in rudimentary {CPU} interrupt logic,'' in
  \emph{Proceedings of the 2018 ACM SIGSAC Conference on Computer and
  Communications Security}, ser. CCS '18.\hskip 1em plus 0.5em minus
  0.4em\relax New York, NY, USA: ACM, 2018, pp. 178--195. [Online]. Available:
  \url{http://doi.acm.org/10.1145/3243734.3243822}
\BIBentrySTDinterwordspacing

\bibitem{transient}
C.~Canella, J.~V. Bulck, M.~Schwarz, M.~Lipp, B.~von Berg, P.~Ortner,
  F.~Piessens, D.~Evtyushkin, and D.~Gruss, ``A systematic evaluation of
  transient execution attacks and defenses,'' in \emph{28th {USENIX} Security
  Symposium, {USENIX} Security 2019}, 2019.

\bibitem{controlled-channels}
Y.~Xu, W.~Cui, and M.~Peinado, ``Controlled-channel attacks: Deterministic side
  channels for untrusted operating systems,'' in \emph{2015 {IEEE} Symposium on
  Security and Privacy, {SP} 2015, San Jose, CA, USA, May 17-21, 2015}.\hskip
  1em plus 0.5em minus 0.4em\relax {IEEE} Computer Society, 2015, pp. 640--656.

\bibitem{sancus}
\BIBentryALTinterwordspacing
J.~Noorman, J.~V. Bulck, J.~T. M\"{u}hlberg, F.~Piessens, P.~Maene, B.~Preneel,
  I.~Verbauwhede, J.~G\"{o}tzfried, T.~M\"{u}ller, and F.~Freiling, ``Sancus
  2.0: A low-cost security architecture for iot devices,'' \emph{ACM Trans.
  Priv. Secur.}, vol.~20, no.~3, pp. 7:1--7:33, Jul. 2017. [Online]. Available:
  \url{http://doi.acm.org/10.1145/3079763}
\BIBentrySTDinterwordspacing

\bibitem{trustlite}
P.~Koeberl, S.~Schulz, A.~Sadeghi, and V.~Varadharajan, ``Trustlite: a security
  architecture for tiny embedded devices,'' in \emph{Ninth Eurosys Conference
  2014, EuroSys 2014, Amsterdam, The Netherlands, April 13-16, 2014}, D.~C.~A.
  Bulterman, H.~Bos, A.~I.~T. Rowstron, and P.~Druschel, Eds.\hskip 1em plus
  0.5em minus 0.4em\relax {ACM}, 2014, pp. 10:1--10:14.

\bibitem{sgxlinger}
W.~He, W.~Zhang, S.~Das, and Y.~Liu, ``{SGXlinger}: {A} new side-channel attack
  vector based on interrupt latency against enclave execution,'' in \emph{36th
  {IEEE} International Conference on Computer Design, {ICCD} 2018, Orlando, FL,
  USA, October 7-10, 2018}.\hskip 1em plus 0.5em minus 0.4em\relax {IEEE}
  Computer Society, 2018, pp. 108--114.

\bibitem{sgx-step}
J.~V. Bulck, F.~Piessens, and R.~Strackx, ``Sgx-step: {A} practical attack
  framework for precise enclave execution control,'' in \emph{Proceedings of
  the 2nd Workshop on System Software for Trusted Execution, SysTEX@SOSP 2017,
  Shanghai, China, October 28, 2017}.\hskip 1em plus 0.5em minus 0.4em\relax
  {ACM}, 2017, pp. 4:1--4:6.

\bibitem{branch-shadowing}
\BIBentryALTinterwordspacing
S.~Lee, M.~Shih, P.~Gera, T.~Kim, H.~Kim, and M.~Peinado, ``Inferring
  fine-grained control flow inside {SGX} enclaves with branch shadowing,'' in
  \emph{26th {USENIX} Security Symposium, {USENIX} Security 2017, Vancouver,
  BC, Canada, August 16-18, 2017.}, E.~Kirda and T.~Ristenpart, Eds.\hskip 1em
  plus 0.5em minus 0.4em\relax {USENIX} Association, 2017, pp. 557--574.
  [Online]. Available:
  \url{https://www.usenix.org/conference/usenixsecurity17/technical-sessions/presentation/lee-sangho}
\BIBentrySTDinterwordspacing

\bibitem{sgxpectre}
G.~Chen, S.~Chen, Y.~Xiao, Y.~Zhang, Z.~Lin, and T.~H. Lai, ``Sgxpectre
  attacks: Stealing intel secrets from sgx enclaves via speculative
  execution.''

\bibitem{ruan}
R.~de~Clercq, F.~Piessens, D.~Schellekens, and I.~Verbauwhede, ``Secure
  interrupts on low-end microcontrollers,'' in \emph{{IEEE} 25th International
  Conference on Application-Specific Systems, Architectures and Processors,
  {ASAP} 2014, Zurich, Switzerland, June 18-20, 2014}.\hskip 1em plus 0.5em
  minus 0.4em\relax {IEEE} Computer Society, 2014, pp. 147--152.

\bibitem{sancus1}
\BIBentryALTinterwordspacing
J.~Noorman, P.~Agten, W.~Daniels, R.~Strackx, A.~V. Herrewege, C.~Huygens,
  B.~Preneel, I.~Verbauwhede, and F.~Piessens, ``Sancus: Low-cost trustworthy
  extensible networked devices with a zero-software trusted computing base,''
  in \emph{Proceedings of the 22th {USENIX} Security Symposium, Washington, DC,
  USA, August 14-16, 2013}, S.~T. King, Ed.\hskip 1em plus 0.5em minus
  0.4em\relax {USENIX} Association, 2013, pp. 479--494. [Online]. Available:
  \url{https://www.usenix.org/conference/usenixsecurity13/technical-sessions/presentation/noorman}
\BIBentrySTDinterwordspacing

\bibitem{abadi1999protection}
M.~Abadi, ``Protection in programming-language translations,'' in \emph{Secure
  Internet Programming, Security Issues for Mobile and Distributed Objects},
  ser. Lecture Notes in Computer Science, J.~Vitek and C.~D. Jensen, Eds., vol.
  1603.\hskip 1em plus 0.5em minus 0.4em\relax Springer, 1999, pp. 19--34.

\bibitem{intel-sgx-explained}
\BIBentryALTinterwordspacing
V.~Costan and S.~Devadas, ``Intel {SGX} explained,'' \emph{{IACR} Cryptology
  ePrint Archive}, vol. 2016, p.~86, 2016. [Online]. Available:
  \url{http://eprint.iacr.org/2016/086}
\BIBentrySTDinterwordspacing

\bibitem{trustvisor}
J.~M. McCune, Y.~Li, N.~Qu, Z.~Zhou, A.~Datta, V.~D. Gligor, and A.~Perrig,
  ``Trustvisor: Efficient {TCB} reduction and attestation,'' in \emph{31st
  {IEEE} Symposium on Security and Privacy, S{\&}P 2010, 16-19 May 2010,
  Berleley/Oakland, California, {USA}}.\hskip 1em plus 0.5em minus 0.4em\relax
  {IEEE} Computer Society, 2010, pp. 143--158.

\bibitem{komodo}
A.~Ferraiuolo, A.~Baumann, C.~Hawblitzel, and B.~Parno, ``Komodo: Using
  verification to disentangle secure-enclave hardware from software,'' in
  \emph{Proceedings of the 26th Symposium on Operating Systems Principles,
  Shanghai, China, October 28-31, 2017}.\hskip 1em plus 0.5em minus 0.4em\relax
  {ACM}, 2017, pp. 287--305.

\bibitem{vrased}
I.~O. Nunes, K.~Eldefrawy, N.~Rattanavipanon, M.~Steiner, and G.~Tsudik,
  ``Vrased: A verified hardware/software co-design for remote attestation,'' in
  \emph{28th {USENIX} Security Symposium, {USENIX} Security 2019}, 2019.

\bibitem{stealthy-page-tables}
\BIBentryALTinterwordspacing
J.~V. Bulck, N.~Weichbrodt, R.~Kapitza, F.~Piessens, and R.~Strackx, ``Telling
  your secrets without page faults: Stealthy page table-based attacks on
  enclaved execution,'' in \emph{26th {USENIX} Security Symposium, {USENIX}
  Security 2017, Vancouver, BC, Canada, August 16-18, 2017.}, E.~Kirda and
  T.~Ristenpart, Eds.\hskip 1em plus 0.5em minus 0.4em\relax {USENIX}
  Association, 2017, pp. 1041--1056. [Online]. Available:
  \url{https://www.usenix.org/conference/usenixsecurity17/technical-sessions/presentation/van-bulck}
\BIBentrySTDinterwordspacing

\bibitem{sgx-cache}
M.~Schwarz, S.~Weiser, D.~Gruss, C.~Maurice, and S.~Mangard, ``Malware guard
  extension: Using sgx to conceal cache attacks,'' in \emph{International
  Conference on Detection of Intrusions and Malware, and Vulnerability
  Assessment}.\hskip 1em plus 0.5em minus 0.4em\relax Springer, 2017, pp.
  3--24.

\bibitem{ti-msp430}
T.~Instruments, ``{MSP430x1xx Family: User Guide},''
  \url{http://www.ti.com/lit/ug/slau049f/slau049f.pdf}.

\bibitem{msp430-bsl-timing}
T.~Goodspeed, ``Practical attacks against the {MSP}430 {BSL},'' in
  \emph{Twenty-Fifth Chaos Communications Congress.}, 2008.

\bibitem{gruss-sgx-malware}
\BIBentryALTinterwordspacing
M.~Schwarz, S.~Weiser, and D.~Gruss, ``Practical enclave malware with intel
  {SGX},'' \emph{CoRR}, vol. abs/1902.03256, 2019. [Online]. Available:
  \url{http://arxiv.org/abs/1902.03256}
\BIBentrySTDinterwordspacing

\bibitem{abate2018exploring}
C.~Abate, R.~Blanco, D.~Garg, C.~Hritcu, M.~Patrignani, and J.~Thibault,
  ``Journey beyond full abstraction: Exploring robust property preservation for
  secure compilation,'' in \emph{32nd {IEEE} Computer Security Foundations
  Symposium, {CSF} 2019, Hoboken, NJ, USA, June 25-28, 2019}, 2019, pp.
  256--271.

\bibitem{soteria}
\BIBentryALTinterwordspacing
J.~G\"{o}tzfried, T.~M\"{u}ller, R.~de~Clercq, P.~Maene, F.~Freiling, and
  I.~Verbauwhede, ``Soteria: Offline software protection within low-cost
  embedded devices,'' in \emph{Proceedings of the 31st Annual Computer Security
  Applications Conference}, ser. ACSAC 2015.\hskip 1em plus 0.5em minus
  0.4em\relax New York, NY, USA: ACM, 2015, pp. 241--250. [Online]. Available:
  \url{http://doi.acm.org/10.1145/2818000.2856129}
\BIBentrySTDinterwordspacing

\bibitem{asyncshock}
N.~Weichbrodt, A.~Kurmus, P.~R. Pietzuch, and R.~Kapitza, ``Asyncshock:
  Exploiting synchronisation bugs in intel {SGX} enclaves,'' in \emph{Computer
  Security - {ESORICS} 2016 - 21st European Symposium on Research in Computer
  Security, Heraklion, Greece, September 26-30, 2016, Proceedings, Part {I}},
  2016, pp. 440--457.

\bibitem{tale2worlds}
J.~V. Bulck, D.~Oswald, E.~Marin, A.~Aldoseri, F.~D. Garcia, and F.~Piessens,
  ``A tale of two worlds: Assessing the vulnerability of enclave shielding
  runtimes,'' in \emph{Proceedings of the 2019 {ACM} {SIGSAC} Conference on
  Computer and Communications Security, {CCS} 2019, London, UK, November 11-15,
  2019}, 2019, pp. 1741--1758.

\bibitem{gruss-thesis}
D.~Gruss, ``Software-based microarchitectural attacks,'' Ph.D. dissertation,
  Graz University of Technology.

\bibitem{secverilog}
D.~Zhang, Y.~Wang, G.~E. Suh, and A.~C. Myers, ``A hardware design language for
  timing-sensitive information-flow security,'' in \emph{Proceedings of the
  Twentieth International Conference on Architectural Support for Programming
  Languages and Operating Systems, {ASPLOS} '15, Istanbul, Turkey, March 14-18,
  2015}, {\"{O}}.~{\"{O}}zturk, K.~Ebcioglu, and S.~Dwarkadas, Eds.\hskip 1em
  plus 0.5em minus 0.4em\relax {ACM}, 2015, pp. 503--516.

\bibitem{zagieboylo2019using}
\BIBentryALTinterwordspacing
D.~Zagieboylo, G.~E. Suh, and A.~C. Myers, ``Using information flow to design
  an {ISA} that controls timing channels,'' in \emph{32nd {IEEE} Computer
  Security Foundations Symposium, {CSF} 2019, Hoboken, NJ, USA, June 25-28,
  2019}, 2019, pp. 272--287. [Online]. Available:
  \url{https://doi.org/10.1109/CSF.2019.00026}
\BIBentrySTDinterwordspacing

\bibitem{fournet2013fully}
C.~Fournet, N.~Swamy, J.~Chen, P.~Dagand, P.~Strub, and B.~Livshits, ``Fully
  abstract compilation to javascript,'' in \emph{The 40th Annual {ACM}
  {SIGPLAN-SIGACT} Symposium on Principles of Programming Languages, {POPL}
  '13, Rome, Italy - January 23 - 25, 2013}, R.~Giacobazzi and R.~Cousot,
  Eds.\hskip 1em plus 0.5em minus 0.4em\relax {ACM}, 2013, pp. 371--384.

\bibitem{ahmed2008typed}
A.~Ahmed and M.~Blume, ``Typed closure conversion preserves observational
  equivalence,'' in \emph{Proceeding of the 13th {ACM} {SIGPLAN} international
  conference on Functional programming, {ICFP} 2008, Victoria, BC, Canada,
  September 20-28, 2008}, 2008, pp. 157--168.

\bibitem{ahmed2011equivalence}
------, ``An equivalence-preserving {CPS} translation via multi-language
  semantics,'' in \emph{Proceeding of the 16th {ACM} {SIGPLAN} international
  conference on Functional Programming, {ICFP} 2011, Tokyo, Japan, September
  19-21, 2011}, 2011, pp. 431--444.

\bibitem{agten2012secure}
P.~Agten, R.~Strackx, B.~Jacobs, and F.~Piessens, ``Secure compilation to
  modern processors,'' in \emph{25th {IEEE} Computer Security Foundations
  Symposium, {CSF} 2012, Cambridge, MA, USA, June 25-27, 2012}, S.~Chong,
  Ed.\hskip 1em plus 0.5em minus 0.4em\relax {IEEE} Computer Society, 2012, pp.
  171--185.

\bibitem{patrignani2015fully}
M.~Patrignani and D.~Clarke, ``Fully abstract trace semantics for protected
  module architectures,'' \emph{Computer Languages, Systems {\&} Structures},
  vol.~42, pp. 22--45, 2015.

\bibitem{patrignani2015secure}
M.~Patrignani, P.~Agten, R.~Strackx, B.~Jacobs, D.~Clarke, and F.~Piessens,
  ``Secure compilation to protected module architectures,'' \emph{{ACM} Trans.
  Program. Lang. Syst.}, vol.~37, no.~2, pp. 6:1--6:50, 2015.

\bibitem{patrignani2019formal}
\BIBentryALTinterwordspacing
M.~Patrignani, A.~Ahmed, and D.~Clarke, ``Formal approaches to secure
  compilation: A survey of fully abstract compilation and related work,''
  \emph{ACM Comput. Surv.}, vol.~51, no.~6, 2019. [Online]. Available:
  \url{https://doi.org/10.1145/3280984}
\BIBentrySTDinterwordspacing

\bibitem{CodeThatNeverRan}
C.~Disselkoen, R.~Jagadeesan, A.~S.~A. Jeffrey, and J.~Riely, ``The code that
  never ran: Modeling attacks on speculative evaluation,'' in \emph{Proc. IEEE
  Symp. Security and Privacy}, 2019.

\end{thebibliography}

\onecolumn
\appendices
\section{Common definitions for \SH and \SL}
\subsection{Memory and memory layout}
The memory is modeled as a (finite) function mapping $2^{16}$ locations to bytes $b$ (just like in the original Sancus);
Given a memory $\M$, we denote the operation of retrieving the byte associated to the location $l$ as $\M(l)$.

On top of that, and for simplicity, we define read and write operations that work on words (i.e., pair of bytes) and we write $w = b_1 b_0$ to denote that the most significant byte of a word $w$ is $b_1$ and its least significant byte is $b_0$.

The read operation is standard, except that it retrieves two consecutive bytes from a given memory location $l$ (in a little-endian fashion, as in the MSP430):
\begin{align*}
    \M [l] &\triangleq b_1 b_0 \quad \text{if }  \M(l) = b_0 \land \M(l+1) = b_1
\end{align*}

The write operation is more complex because it deals with unaligned memory accesses.
We faithfully model detailed aspects of Sancus, like unaligned accesses, because we want to prove that these detailed aspects do not lead to potential attacks.
\begin{align*}
    (\M [l \mapsto b_1 b_0]) (l') &\triangleq
                                    \begin{cases}
                                        b_0 & \text{if } l' = l \\
                                        b_1 & \text{if } l' = l+1 \\
                                        \M (l') & \text{o.w.}\\
                                    \end{cases}
\end{align*}
Indeed writing $b_0 b_1$ in location $l$ in $\M$ means to build an updated memory that maps $l$ to $b_0$, $l+1$ to $b_1$ and is unchanged otherwise.

Note that reads and writes to $l = \mtt{0xFFFF}$ are undefined operations ($l+1$ would overflow hence it is undefined).
The memory access control relation explicitly forbids these accesses (see below).

Since modeling the memory as a function gives no clues on how the enclave is organized, we assume a fixed \emph{memory layout} $\L$ throughout the whole formalization that describes how the enclave is laid out in memory.
The protected code and the protected data are placed in consecutive, non-overlapping memory sections.
The memory layout $\L$ is used to regulate how the protected data are accessed: actually, it permits only protected code to manipulate protected data, and to jump to a protected address and to execute the instruction stored therein.
The first address of the protected code section also works as the entry point of the software module.
Note that memory operations enforce no memory access control w.r.t.\ $\L$, since these checks are performed during the execution of each instruction (see below).
In addition, the memory layout defines the entry point $\mi{isr}$ of the interrupt service routine, out of the protected sections.
Also, we assume the location $\mtt{0xFFFE}$ to be reserved to store the address of the first instruction to be executed when the CPU starts.
Formally, a memory layout is defined as
\[
    \L \triangleq \langle
        \mi{ts}, \mi{te},
        \mi{ds}, \mi{de},
        \mi{isr}
        \rangle
\]
where:
\begin{itemize}
    \item $[\mi{ts}, \mi{te})$ is the protected code section
    \item $[\mi{ds}, \mi{de})$ is the protected data section
    \item $\mi{isr}$ is the entry point for the ISR
\end{itemize}

\noindent
Also, we assume that:
\begin{itemize}
    \item $\mtt{0xFFFE} \not\in [\mi{ts}, \mi{te}) \cup [\mi{ds}, \mi{de})$
    \item $[\mi{ts}, \mi{te}) \cap [\mi{ds}, \mi{de}) = \emptyset$
    \item $\mi{isr} \not\in [ts, te) \cup [ds, de)$
\end{itemize}

\subsection{Register files}
\SH, just like the original Sancus, has sixteen $16$-bit registers three of which $\reg{0}$, $\reg{1}$, $\reg{2}$ are used for dedicated functions, whereas the others are for general use.
($\reg{3}$ is a constant generator in the real machine, but we ignore that use in our formalization.)
More precisely, $\reg{0}$ (hereafter denoted as $\rpc$) is the program counter and points to the next instruction to be executed.
Instruction accesses are performed by word and the $\rpc$ is aligned to even addresses.
The register $\reg{1}$ ($\rsp$ hereafter) is the stack pointer and it is used, as usual, by
the CPU to store the pointer to the activation record of the current procedure.
Also the stack pointer is aligned to even addresses.
The register $\reg{2}$ ($\rsr$ hereafter) is the status register and contains different
pieces of information encoded as flags.
For example, the fourth bit, called $\mtt{GIE}$, is set to 1 when interrupts are enabled.
Other bits are set, e.g., when an operation produces a carry or when the result
of an operation is zero.

Formally, our \emph{register file} $\R$ is a function that maps each register $\rn{r}$ to a word.
The read operation is standard:
\begin{align*}
    \R [\rn{r}] \triangleq w \text{ if } \R(\rn{r}) = w
\end{align*}

Instead, the write operation requires accommodating the hardware itself and our security requirements (see Section III in the paper for motivation and intuition):
\begin{align*}
    \R [\rn{r} \mapsto w] \triangleq \lambda [\rn{r'}].
                                    \begin{cases}
                                    w \& \mtt{0xFFFE} & \text{if } \rn{r'} = \rn{r} \land (\rn{r} = \rpc \lor \rn{r} = \rsp)\\
                                    (w \& \mtt{0xFFF7}) \mid (\R[\rsr] \& \mtt{0x8}) & \text{if } \rn{r'} = \rn{r} = \rsr \land \moderelPM{\R[\rpc]}\\
                                    w & \text{if } \rn{r'} = \rn{r} \land (\rn{r} \neq \rpc \land \rn{r} \neq \rsp)\\
                                    \R [\rn{r'}] \quad \text{o.w.}
                                    \end{cases}
\end{align*}
In the definition above we use the relation $\moderel{\R[\rpc]}$, for $m \in \{\mtt{PM}, \mtt{UM}\}$ that is defined in subsection~\ref{subsec:pm-um}.
It indicates that the execution is carried on in protected or in unprotected mode.
Note that the least-significant bit of the program counter and of the stack pointer are \emph{always} masked to $0$ (as it happens in MSP430), and that the $\mtt{GIE}$ bit of the status register is always masked to its previous value when in protected mode (i.e., it cannot be changed when the CPU is running protected code).

\subsubsection{Special register files}

We define the following special register files:
\begin{align*}
    \R_0 &\triangleq \{ \rpc \mapsto 0, \rsp \mapsto 0, \rsr \mapsto 0, \reg{3} \mapsto 0, \ldots, \reg{15} \mapsto 0 \}\\
    \R^{\mi{init}}_{\M} &\triangleq \{ \rpc \mapsto \M[\mtt{0xFFFE}], \rsp \mapsto 0, \rsr \mapsto \mtt{0x8}, \reg{3} \mapsto 0, \ldots, \reg{15} \mapsto 0 \}
\end{align*}

where
\begin{itemize}
    \item $\rpc$ is set to $\M[\mtt{0xFFFE}]$ as it does in the MSP430
    \item $\rsp$ is set to $0$ and we expect untrusted code to set it up in a setup phase, if any
    \item $\rsr$ is set to $\mtt{0x8}$, i.e., register is clear except for the $\mtt{GIE}$ flag
\end{itemize}

\subsection{I/O Devices}
We formalize Sancus \emph{I/O devices} as (simplified) \emph{deterministic I/O automata} $\D \triangleq \langle \Delta, \deltainit, \xleadsto{a}_D \rangle$ over a common signature $A$:
\begin{itemize}
    \item $A$ includes the following actions (below $w$ is a word):
        \begin{itemize}
            \item $\epsilon$, a silent, internal action;
            \item $\mi{rd}(w)$, an output action (i.e., read request from the CPU);
            \item $\mi{wr}(w)$, an input action (i.e., write request from the CPU);
            \item $\mi{int?}$ an output action telling that an interrupt was raised in the last state.
        \end{itemize}
    \item $\emptyset \neq \Delta$ is the \emph{finite} set of internal states of the device
    \item $\deltainit \in \Delta$ is the \emph{single} initial state
    \item $\delta \xleadsto{a}_D \delta' \subseteq \Delta \times A \times \Delta$ is the transition function that takes one step in the device while doing action $a \in A$, starting in state $\delta$ and ending in state $\delta'$.
    (We write $\overline{a}$ for a string of actions and we omit $\epsilon$ when unnecessary.)
    The transition function is such that $\forall \delta$ either $\delta \xleadsto{\epsilon}_D \delta'$ or
    $\delta \xleadsto{\mi{int?}}_D \delta''$ (i.e., one and only one of the two transitions must be possible), also at most one $\mi{rd}(w)$ action must be possible starting from a given state.
\end{itemize}
Note: to keep the presentation simple we assume to have a special state which is the destination of any action not explicitly defined.

\subsection{Contexts, software modules and whole programs}
\begin{definition}
    We call \emph{software module} a memory $\M_M$ containing both protected data and code sections.
\end{definition}

Intuitively, the context is the part of the whole program that can be manipulated by an attacker:
\begin{definition}
    A \emph{context} $C$ is a pair $\langle \M_C, \D \rangle$, where $\D$ is a device and $\M_C$ defines the contents of all memory locations \emph{outside} the protected sections of the layout.
\end{definition}

\begin{definition}
    Given a context $C = \langle \M_C, \D \rangle$ and a software module $\M_M$ such that $\dom{\M_C} \,\cap\,\dom{\M_M} = \emptyset$, a whole program is
    \[
        C[\M_M] \triangleq \langle \M_C \uplus \M_M, \D \rangle.
    \]
\end{definition}

\subsection{Instruction set}
The instruction set $ \mi{Inst} \ni i $ is the same for both \SL and \SH and is (almost)  that of the MSP430.
An overview of the instruction set is in Table~\ref{tab:op-summary-rep}.
For each instruction the table includes its operands, an  intuitive meaning of
its semantics, its duration and the number of words it occupies in memory.
The durations are used to define the function $\cycles{i}$ and implicitly determine a value $\MT$, greater than or equal to the duration of longest instruction.
Here we choose $\MT = 6$, in order to maintain the compatibility with the real MSP430 (whose longest instruction takes $6$ cycles).
Since instructions are stored in either the unprotected or in the protected code section of the memory $\M$, for
getting them we use the meta-function $\decode{\M}{l}$ that decodes the contents of the cell(s) starting at location $l$, returning an instruction in the table if any and  $\bot$ otherwise.
\begin{table}[tb]
    \centering
    \begin{tabular}{@{}llll@{}}
    \toprule
    \textbf{Instr. $i$} &
    \textbf{Meaning} 
    & \textbf{Cycles}  & \textbf{Size} \\
    \midrule
    $\RETI$ & Returns from interrupt. & $5$ & $1$\\
    $\NOP$  & No-operation.           & $1$                 & $1$\\
    $\HLT$  & Halt.                   & $1$                 & $1$\\
    $\NOT{r}$  & $\rn{r} \leftarrow \lnot \rn{r}$. (Emulated in MSP430) & $2$ & $2$\\
    $\IN{r}$  & Reads word from the device and puts it in $\rn{r}$. & $2$ & $1$\\
    $\OUT{r}$  & Writes word in register $\rn{r}$ to the device. & $2$ & $1$\\
    $\AND{r_1}{r_2}$  & $\rn{r_2} \leftarrow \rn{r_1}\ \&\ \rn{r_2}$. & $1$ & $1$\\
    $\JMP{r}$ & Sets $\rpc$ to the value in $\rn{r}$. & $2$ & $1$\\
    $\JZ{r}$ & Sets $\rpc$ to the value in $\rn{r}$ if bit 0 in $\rsr$ is set. & $2$ & $1$\\
    $\MOV{r_1}{r_2}$  & $\rn{r_2} \leftarrow \rn{r_1}$. & $1$ & $1$\\
    $\MOVL{r_1}{r_2}$  & Loads in $\rn{r_2}$ the word in starting in location pointed by $\rn{r_1}$. & $2$ & $1$\\
    $\MOVS{r_1}{r_2}$  & Stores the value of $\rn{r_1}$ starting at location pointed by $\rn{r_2}$. & $4$ & $2$\\
    $\MOVI{w}{r_2}$  & $\rn{r_2} \leftarrow w$. & $2$ & $2$\\
    $\ADD{r_1}{r_2}$  & $\rn{r_2} \leftarrow \rn{r_1} + \rn{r_2}$. & $1$ & $1$\\
    $\SUB{r_1}{r_2}$  & $\rn{r_2} \leftarrow \rn{r_1} - \rn{r_2}$. & $1$ & $1$\\
    $\CMP{r_1}{r_2}$  & Zero bit in $\rsr$ set if $\rn{r_2} - \rn{r_1}$ is zero. & $1$ & $1$\\
    \bottomrule
    \end{tabular}

\caption{Summary of the assembly language considered.}\label{tab:op-summary-rep}
\end{table}

\subsection{Configurations}
Given an I/O device $\D$, the internal state of the CPU is described by configurations
of the form:
\[
    c \triangleq \cfg{\B}{\delta}{t}{t_a}{\M}{\R}{\vopc} \in \mb{C}, \quad \text{where}
\]
\begin{itemize}
    \item $\delta$ is the current state of the I/O device;
    \item $t$ is the current time of the CPU;
    \item $t_a$ is either the arrival time of the last pending interrupt, or $\bot$ if there are none;
    \item $\M$ is the current memory;
    \item $\R$ is the current content of the registers;
    \item $\vopc$
    is the value of the program counter before executing the current instruction
    \item $\B$ is called the \emph{backup} and can assume different values:
        \begin{itemize}
            \item $\BBot$, indicating either that the CPU is not handling an interrupt or it is handling one originated in unprotected mode interrupt
            \item $\langle \R, \vopc, t_{\mi{pad}}\rangle$, refers to the case in which an interrupt handler whose interrupt originated in protected mode is being executed.
            The triple includes the register file and the old program counter at the time the interrupt originated and the value $t_{\mi{pad}}$, which indicates the remaining padding time that must be applied before returning into protected mode.
        \end{itemize}
\end{itemize}
%

The initial states of the CPU are represented by the initial configurations from which the computation starts.
The initial configuration for a whole program $C[\M_M] = \langle \M, \D \rangle$ is:
\[
    \initconf{C}{\M_M} \triangleq \cfg{\BBot}{\deltainit}{0}{\bot}{\M}{\R^{\mi{init}}_{\M_C}}{\mtt{0xFFFE}} \text{ where}
\]
\begin{itemize}
    \item the state of the I/O device $\D$ is $\deltainit$;
    \item the initial value of the clock is $0$ and no interrupt has arrived yet;
    \item the memory is initialized to the whole program memory $\M_C \uplus \M_M$;
    \item registers are initialized to their initial values, i.e., all the registers are set to $0$ except that $\rpc$ is set to $\mtt{0xFFFE}$ (the address from which the CPU gets the initial program counter), i.e.,  $\rpc = \M[\mtt{0xFFFE}]$ (as in Sancus), and that $\rsr$ is set to $\mtt{0x8}$ (the register is clear except for the $\mtt{GIE}$ flag);
    \item the previous program counter is also initialized to $\mtt{0xFFFE}$;
    \item the backup is set to $\BBot$ to indicate absence of any backup.
\end{itemize}

Dually, $\haltconf$ is the only configuration denoting termination, more specifically it is an opaque and distinguished configuration that indicates graceful termination.

Also, we define \emph{exception handling} configurations, through which the processor goes whenever a halt happens in protected mode or a violation happens in any mode.

Intuitively these configurations serve as a starting point for the exception handling routine provided by the attacker, whose entry point address resides at address $\mtt{0xFFFE}$:
\[
    \exconf{\cfg{\B}{\delta}{t}{t_a}{\M}{\R}{\vopc}} \triangleq \cfg{\BBot}{\delta}{t}{\bot}{\M}{\R_0 [\rpc \mapsto \M[\mtt{0xFFFE}]]}{\mtt{0xFFFE}}.
\]

\subsubsection{I/O device wrapper}
The main transition system relies on an auxiliary transition system that synchronizes the evolution of the I/O device with that of the CPU.
For that, we define a ``wrapper'' around the device $\D$:
\[
    \dwrap{k}{\delta, t, t_a}{\delta', t', t'_a}
\]
Intuitively, assume that the device is in state $\delta$, the clock time is $t$ and the last interrupt was raised at time $t_a$.
Then, after $k$ cycles the new clock time will be $t' = t+k$, the last interrupt was raised at time $t'_a$ and the new state will be $\delta'$.
Note that when no interrupt has to be handled, $t_a$ and $t'_a$ have the value $\bot$.

Formally:
\begin{mathpar}
    \inferrule
    {
        a \in \{ \epsilon, \mi{int?} \}\\
        \bigwedge\limits_{i=0}^{k-1} \delta_i \xleadsto{a}_D \delta_{i+1}\\
        t'_a = {\begin{cases}
            t + j & \text{if }
                \exists 0 \leq j < k .\, \delta_j \xleadsto{\mi{int?}}_D \delta_{j+1} \land
                \\ &
                \forall j' < j.\, \delta_{j'} \leadsto_D \delta_{j'+1} \\
            t_a & \text{o.w.}
        \end{cases}}
    }
    {
        \dwrap{k}{\delta_0, t, t_a}{\delta_k, (t+k), t'_a}
    }
\end{mathpar}

\begin{property}\label{prop:dwrap-deterministic}
    If $\dwrap{k}{\delta, t, t_a}{\delta', t', t'_a}$ and $\dwrap{k}{\delta, t, t_a}{\delta'', t'', t''_a}$, then $\delta' = \delta''$, $t' = t''$ and $t'_a = t''_a$.
\end{property}
\begin{proof}
    Trivial.
\end{proof}

\subsection{CPU mode}\label{subsec:pm-um}

There are two further relations used by the main transition systems, specifying the \emph{CPU mode} and the \emph{memory access control}, MAC.

The first tells when the given address, is an address in the protected code memory ($\mtt{PM}$) or in the unprotected one ($\mtt{UM}$):
\[
    \vpc \text{, with } \mtt{m} \in \{ \mtt{PM}, \mtt{UM} \}
\]
Formally:
\begin{mathpar}
    \inferrule
    {
        \vpc \in [\L.\mi{ts}, \L.\mi{te})
    }
    {
        \moderelPM{\vpc}
    }

    \inferrule
    {
        \vpc \not\in [\L.\mi{ts}, \L.\mi{te}) \cup [\L.\mi{ds}, \L.\mi{de})
    }
    {
        \moderelUM{\vpc}
    }
\end{mathpar}

Also, we lift the definition to configurations as follows:
\begin{mathpar}
    \inferrule
    {
        \moderel{\R[\rpc]}
    }
    {
        \moderel{\cfg{\B}{\delta}{t}{t_a}{\M}{\R}{\vopc}}
    }

    \inferrule
    { }
    {
        \moderelUM{\haltconf}
    }
\end{mathpar}

\subsection{Memory access control}\label{subsec:macrel}

The second relation holds whenever the instruction $i$ can be executed in a CPU configuration in which the previous program counter is $\vopc$, the registers are $\R$ and the backup is $\B$, and takes the following form:
\[
    \macrel{i}{\vopc}{\R}{\B}
\]
More precisely, it uses the predicate $\mac{f}{rght}{t}$ (defined in \tablename~\ref{tab:mac-rep}) that holds whenever from the location $f$ we have the rights $\mtt{rght}$ on location $t$.
The predicate checks that
$(1)$ the code we came from (i.e., that in location $\vopc$) can actually execute instructions located at $\R[\rpc]$;
$(2)$ $i$ can be executed in current CPU mode, and if $i$ is a memory operation;
$(3)$ from $\R[\rpc]$ we have the rights to perform the requested operation in memory.
\begin{table}[tb]
    \begin{center}
        \begin{tabular}{@{}llcccc@{}}
            & & \multicolumn{4}{c}{$t$}\\
            \cmidrule(l){3-6}
            & \multicolumn{1}{l}{} & Entry Point & Prot. code & Prot. Data & Other \\
            \midrule
            \multicolumn{1}{l|}{\multirow{2}{*}{$f$}} & \multicolumn{1}{l|}{Entry Point/Prot. code} & r-x & r-x & rw- & --x \\
            \multicolumn{1}{l|}{} & \multicolumn{1}{l|}{Other} & --x & --- & --- & rwx \\
            \bottomrule
        \end{tabular}
    \end{center}

\caption{Definition of $\mac{f}{rght}{t}$ function, where $f$ and $t$ are locations.}\label{tab:mac-rep}
\end{table}
Formally, the definition of the relation is the following:
\newcommand{\neqll}[1]{{#1} \neq 2^{16}-1}
\begin{mathpar}
    \inferrule*
    {
       \neqll{\R[\rsp]}\\
       \neqll{\R[\rsp]+2}\\
       \mac{\vopc}{x}{\R[\rpc]}\\
       \mac{\vopc}{x}{\R[\rpc]+1}\\
       \mac{\R[\rpc]}{r}{\R[\rsp]}\\
       \mac{\R[\rpc]}{r}{\R[\rsp]+1}\\
       \mac{\R[\rpc]}{r}{\R[\rsp]+2}\\
       \mac{\R[\rpc]}{r}{\R[\rsp]+3}
    }
    {
        \macrel{\RETI}{\vopc}{\R}{\bot}
    }

    \inferrule*
    {
       i \in \{ \NOP, \AND{r_1}{r_2}, \ADD{r_1}{r_2}, \SUB{r_1}{r_2}, \CMP{r_1}{r_2}, \MOV{r_1}{r_2}, \JMP{r}, \JZ{r} \}\\
       \mac{\vopc}{x}{\R[\rpc]}\\
       \mac{\vopc}{x}{\R[\rpc]+1}
    }
    {
        \macrel{i}{\vopc}{\R}{\bot}
    }

    \inferrule*
    {
       i \in \{ \NOT{r}, \MOVI{w}{r} \}\\
       \mac{\vopc}{x}{\R[\rpc]}\\
       \mac{\vopc}{x}{\R[\rpc]+1}\\
       \mac{\vopc}{x}{\R[\rpc]+2}\\
       \mac{\vopc}{x}{\R[\rpc]+3}\\
    }
    {
        \macrel{i}{\vopc}{\R}{\bot}
    }

    \inferrule*
    {
       i \in \{ \IN{r}, \OUT{r} \}\\
       \moderelUM{\R[\rpc]}\\
       \mac{\vopc}{x}{\R[\rpc]}\\
       \mac{\vopc}{x}{\R[\rpc]+1}
    }
    {
        \macrel{i}{\vopc}{\R}{\bot}
    }

    \inferrule*
    {
        \neqll{\R[\rn{r_1}]}\\
        \neqll{\R[\rn{r_1}]+1}\\
        \mac{\R[\rpc]}{r}{\R[\rn{r_1}]}\\
        \mac{\R[\rpc]}{r}{\R[\rn{r_1}]+1}\\
        \mac{\vopc}{x}{\R[\rpc]}\\
        \mac{\vopc}{x}{\R[\rpc]+1}
    }
    {
        \macrel{\MOVL{r_1}{r_2}}{\vopc}{\R}{\bot}
    }

    \inferrule*
    {
       \neqll{\R[\rn{r_2}]}\\
       \neqll{\R[\rn{r_2}]+1}\\
       \mac{\R[\rpc]}{w}{\R[\rn{r_2}]}\\
       \mac{\R[\rpc]}{w}{\R[\rn{r_2}]+1}\\
       \mac{\vopc}{x}{\R[\rpc]}\\
       \mac{\vopc}{x}{\R[\rpc]+1}\\
       \mac{\vopc}{x}{\R[\rpc]+2}\\
       \mac{\vopc}{x}{\R[\rpc]+3}
    }
    {
        \macrel{\MOVS{r_1}{r_2}}{\vopc}{\R}{\bot}
    }

    \inferrule*
    {
       i \neq \RETI\\
       \B \neq \bot\\
       \macrel{i}{\vopc}{\R}{\bot}\\
       \R[\rsr].\mtt{GIE} = \mtt{0}\\
       \R[\rpc] \neq \mi{ts}
    }
    {
        \macrel{i}{\vopc}{\R}{\B}
    }

    \inferrule*
    {
        \B \neq \BBot
    }
    {
        \macrel{\RETI}{\vopc}{\R}{\B}
    }
\end{mathpar}
Note that
$(i)$ for each word that is accessed in memory we also check that the first location is not the last byte of the memory (except for the program counter, for which the decode function would fail since it would try to access undefined memory);
$(ii)$ word accesses must be checked once for each byte of the word;
and $(iii)$ checks on $\rpc$ guarantee that a memory violation does not happen while decoding.

\FloatBarrier

\section{The main transition system and interrupt logic for \SH and \SL}
  The main transition systems for our versions of Sancus share a large part of inference rules, and heavily differ on the way interrupts are handled, as in \SH there are none.
  Hereafter we assume as given a context $C = \langle \M_C, \D \rangle$.

  \subsection{\SH}
We now present the operational semantics of \SH that relies a very simple auxiliary transition system for interrupts.

\subsubsection{Main transition system}

We represent how the \SH configuration $c$ becomes with a computation step $c'$ by the main transition system, with transition of the following form:
\[
    \shmain{c}{c'}
\]

Figures~\ref{fig:sh-maints},~\ref{fig:sh-maints2} and \ref{fig:sh-maints3} report the full set of rules that define the main transition system of \SH.
%

\ifdefined\LANG
\renewcommand{\LANG}{sh}
\else
\newcommand{\LANG}{sh}
\fi


\IfStrEq{\LANG}{sh}{
    \ifdefined\sint
    \renewcommand{\sint}[3]{\shint{#1}{#2}{#3}}
    \else
    \newcommand{\sint}[3]{\shint{#1}{#2}{#3}}
    \fi

    \ifdefined\smain
    \renewcommand{\smain}[2]{\shmain{#1}{#2}}
    \else
    \newcommand{\smain}[2]{\shmain{#1}{#2}}
    \fi

    \ifdefined\srulename
    \renewcommand{\srulename}[1]{\shrulename{#1}}
    \else
    \newcommand{\srulename}[1]{\shrulename{#1}}
    \fi

    \ifdefined\slang
    \renewcommand{\slang}[0]{\SH}
    \else
    \newcommand{\slang}[0]{\SH}
    \fi

    \ifdefined\slblone
    \renewcommand{\slblone}[0]{\label{fig:sh-maints}}
    \else
    \newcommand{\slblone}[0]{\label{fig:sh-maints}}
    \fi

    \ifdefined\slbltwo
    \renewcommand{\slbltwo}[0]{\label{fig:sh-maints2}}
    \else
    \newcommand{\slbltwo}[0]{\label{fig:sh-maints2}}
    \fi

    \ifdefined\slblthree
    \renewcommand{\slblthree}[0]{\label{fig:sh-maints3}}
    \else
    \newcommand{\slblthree}[0]{\label{fig:sh-maints3}}
    \fi

    \ifdefined\smainarrow
    \renewcommand{\smainarrow}[0]{\shmainarrow}
    \else
    \newcommand{\smainarrow}[0]{\shmainarrow}
    \fi
}{
    \ifdefined\sint
    \renewcommand{\sint}[3]{\slint{#1}{#2}{#3}}
    \else
    \newcommand{\sint}[3]{\slint{#1}{#2}{#3}}
    \fi

    \ifdefined\smain
    \renewcommand{\smain}[2]{\slmain{#1}{#2}}
    \else
    \newcommand{\smain}[2]{\slmain{#1}{#2}}
    \fi

    \ifdefined\srulename
    \renewcommand{\srulename}[1]{\slrulename{#1}}
    \else
    \newcommand{\srulename}[1]{\slrulename{#1}}
    \fi

    \ifdefined\slang
    \renewcommand{\slang}[0]{\SL}
    \else
    \newcommand{\slang}[0]{\SL}
    \fi

    \ifdefined\slblone
    \renewcommand{\slblone}[0]{\label{fig:sl-maints}}
    \else
    \newcommand{\slblone}[0]{\label{fig:sl-maints}}
    \fi

    \ifdefined\slbltwo
    \renewcommand{\slbltwo}[0]{\label{fig:sl-maints2}}
    \else
    \newcommand{\slbltwo}[0]{\label{fig:sl-maints2}}
    \fi

    \ifdefined\slblthree
    \renewcommand{\slblthree}[0]{\label{fig:sl-maints3}}
    \else
    \newcommand{\slblthree}[0]{\label{fig:sl-maints3}}
    \fi

    \ifdefined\smainarrow
    \renewcommand{\smainarrow}[0]{\slmainarrow}
    \else
    \newcommand{\smainarrow}[0]{\slmainarrow}
    \fi
}

\ifdefined\bnotspecial
\renewcommand{\bnotspecial}[0]{\B \neq \langle \bot, \bot, t_\mi{pad} \rangle}
\else
\newcommand{\bnotspecial}[0]{\B \neq \langle \bot, \bot, t_\mi{pad} \rangle}
\fi

\begin{figure}[bt]
    \begin{scriptsize}
    \begin{mathpar}
        \inferrule*[
            lab={\srulename{(CPU-Decode-Fail)}}
        ]
        {
            \bnotspecial\\
            \decode{\M}{\R[\rpc]} = \bot
        }
        {
            \smain{\cfg{\B}{\delta}{t}{t_a}{\M}{\R}{\vopc}}{\exconf{\cfg{\B}{\delta}{t}{t_a}{\M}{\R}{\vopc}}}
        }

        \inferrule*[
            lab={\srulename{(CPU-HLT-UM)}},
            right={$\decode{\M}{\R[\rpc])} = \HLT$}]
        {
            \bnotspecial\\
            \moderelUM{\cfg{\B}{\delta}{t}{t_a}{\M}{\R}{\vopc}} \\
        }
        {
            \smain{\cfg{\B}{\delta}{t}{t_a}{\M}{\R}{\vopc}}{\haltconf}
        }

        \inferrule*[
            lab={\srulename{(CPU-HLT-PM)}},
            right={$i = \decode{\M}{\R[\rpc])} = \HLT$}]
        {
            \bnotspecial\\
            \moderelPM{\cfg{\B}{\delta}{t}{t_a}{\M}{\R}{\vopc}} \\
        }
        {
            \smain{\cfg{\B}{\delta}{t}{t_a}{\M}{\R}{\vopc}}{\exconf{\cfg{\B}{\delta}{t + \cycles{i}}{t_a}{\M}{\R}{\vopc}}}
        }

        \inferrule*[
            lab={\srulename{(CPU-Violation-PM)}},
            right={$i = \decode{\M}{\R[\rpc])} \neq \bot$}
        ]
        {
            \bnotspecial\\
            \nmacrel{i}{\vopc}{\R}{\B}
        }
        {
            \smain{\cfg{\B}{\delta}{t}{t_a}{\M}{\R}{\vopc}}{\exconf{\cfg{\B}{\delta}{t + \cycles{i}}{t_a}{\M}{\R}{\vopc}}}
        }

        \inferrule*[
        lab={\srulename{(CPU-MovL)}},
        right={$i = \decode{\M}{\R[\rpc])} = \MOVL{r_1}{r_2}$}]
        {
            \bnotspecial\\
            \macrel{i}{\vopc}{\R}{\B}\\
            \R' = \R[\rpc \mapsto \R[\rpc] + 2][\rn{r_2} \mapsto \M[\R[\rn{r_1}]]]\\
            \dwrap{\cycles{i}}{\delta, t, t_a}{\delta', t', t'_a}\\
            \sint{\D}{\cfg{\B}{\delta'}{t'}{t'_a}{\M}{\R'}{\R[\rpc]}}{\cfg{\B'}{\delta''}{t''}{t''_a}{\M'}{\R''}{\R[\rpc]}}
        }
        {
            \smain{\cfg{\B}{\delta}{t}{t_a}{\M}{\R}{\vopc}}{\cfg{\B'}{\delta''}{t''}{t''_a}{\M'}{\R''}{\R[\rpc]}}
        }

        \inferrule*[
        lab={\srulename{(CPU-MovS)}},
        right={$i = \decode{\M}{\R[\rpc])} = \MOVS{r_1}{r_2}$}]
        {
            \bnotspecial\\
            \macrel{i}{\vopc}{\R}{\B}\\
            \R' = \R[\rpc \mapsto \R[\rpc] + 4]\\
            \M' = \M[\R[\rn{r_2}] \mapsto \R[\rn{r_1}]] \\
            \dwrap{\cycles{i}}{\delta, t, t_a}{\delta', t', t'_a}\\
            \sint{\D}{\cfg{\B}{\delta'}{t'}{t'_a}{\M'}{\R'}{\R[\rpc]}}{\cfg{\B'}{\delta''}{t''}{t''_a}{\M''}{\R''}{\R[\rpc]}}
        }
        {
            \smain{\cfg{\B}{\delta}{t}{t_a}{\M}{\R}{\vopc}}{\cfg{\B'}{\delta''}{t''}{t''_a}{\M''}{\R''}{\R[\rpc]}}
        }

        \inferrule*[
        lab={\srulename{(CPU-Mov)}},
        right={$i = \decode{\M}{\R[\rpc])} = \MOV{r_1}{r_2}$}]
        {
            \bnotspecial\\
            \macrel{i}{\vopc}{\R}{\B}\\
            \R' = \R[\rpc \mapsto \R[\rpc] + 2][\rn{r_2} \mapsto \R[\rn{r_1}]]\\
            \dwrap{\cycles{i}}{\delta, t, t_a}{\delta', t', t'_a}\\
            \sint{\D}{\cfg{\B}{\delta'}{t'}{t'_a}{\M}{\R'}{\R[\rpc]}}{\cfg{\B'}{\delta''}{t''}{t''_a}{\M'}{\R''}{\R[\rpc]}}
        }
        {
            \smain{\cfg{\B}{\delta}{t}{t_a}{\M}{\R}{\vopc}}{\cfg{\B'}{\delta''}{t''}{t''_a}{\M'}{\R''}{\R[\rpc]}}
        }

        \inferrule*[
        lab={\srulename{(CPU-MovI)}},
        right={$i = \decode{\M}{\R[\rpc])} = \MOVI{w}{r}$}]
        {
            \bnotspecial\\
            \macrel{i}{\vopc}{\R}{\B}\\
            \R' = \R[\rpc \mapsto \R[\rpc] + 4][\rn{r} \mapsto w]\\
            \dwrap{\cycles{i}}{\delta, t, t_a}{\delta', t', t'_a}\\
            \sint{\D}{\cfg{\B}{\delta'}{t'}{t'_a}{\M}{\R'}{\R[\rpc]}}{\cfg{\B'}{\delta''}{t''}{t''_a}{\M'}{\R''}{\R[\rpc]}}
        }
        {
            \smain{\cfg{\B}{\delta}{t}{t_a}{\M}{\R}{\vopc}}{\cfg{\B'}{\delta''}{t''}{t''_a}{\M'}{\R''}{\R[\rpc]}}
        }

        \inferrule*[
        lab={\srulename{(CPU-Nop)}},
        right={$i = \decode{\M}{\R[\rpc]} = \NOP$}]
        {
            \bnotspecial\\
            \macrel{i}{\vopc}{\R}{\B}\\
            \R' = \R[\rpc \mapsto \R[\rpc] + 2]\\
            \dwrap{\cycles{i}}{\delta, t, t_a}{\delta', t', t'_a}\\
            \sint{\D}{\cfg{\B}{\delta'}{t'}{t'_a}{\M}{\R'}{\R[\rpc]}}{\cfg{\B'}{\delta''}{t''}{t''_a}{\M'}{\R''}{\R[\rpc]}}
        }
        {
            \smain{\cfg{\B}{\delta}{t}{t_a}{\M}{\R}{\vopc}}{\cfg{\B'}{\delta''}{t''}{t''_a}{\M'}{\R''}{\R[\rpc]}}
        }
    \end{mathpar}
    \end{scriptsize}

    \caption{Rules of the main transition system for \slang. (part I)}
    \slblone
\end{figure}
~
\begin{figure}[bt]
    \begin{scriptsize}
    \begin{mathpar}
        \inferrule*[
        lab={\srulename{(CPU-Jz0)}},
        right={$i = \decode{\M}{\R[\rpc]} = \JZ{r} \land \R[\rsr].Z = 0$}]
        {
            \bnotspecial\\
            \macrel{i}{\vopc}{\R}{\B}\\
            \R' = \R[\rpc \mapsto \R[\rpc] + 2]\\
            \dwrap{\cycles{i}}{\delta, t, t_a}{\delta', t', t'_a}\\
            \sint{\D}{\cfg{\B}{\delta'}{t'}{t'_a}{\M}{\R'}{\R[\rpc]}}{\cfg{\B'}{\delta''}{t''}{t''_a}{\M'}{\R''}{\R[\rpc]}}
        }
        {
            \smain{\cfg{\B}{\delta}{t}{t_a}{\M}{\R}{\vopc}}{\cfg{\B'}{\delta''}{t''}{t''_a}{\M'}{\R''}{\R[\rpc]}}
        }

        \inferrule*[
        lab={\srulename{(CPU-Jz1)}},
        right={$i = \decode{\M}{\R[\rpc]} = \JZ{r} \land \R[\rsr].Z = 1$}]
        {
            \bnotspecial\\
            \macrel{i}{\vopc}{\R}{\B}\\
            \R' = \R[\rpc \mapsto \R[\rn{r}]]\\
            \dwrap{\cycles{i}}{\delta, t, t_a}{\delta', t', t'_a}\\
            \sint{\D}{\cfg{\B}{\delta'}{t'}{t'_a}{\M}{\R'}{\R[\rpc]}}{\cfg{\B'}{\delta''}{t''}{t''_a}{\M'}{\R''}{\R[\rpc]}}
        }
        {
            \smain{\cfg{\B}{\delta}{t}{t_a}{\M}{\R}{\vopc}}{\cfg{\B'}{\delta''}{t''}{t''_a}{\M'}{\R''}{\R[\rpc]}}
        }

        \inferrule*[
        lab={\srulename{(CPU-Jmp)}},
        right={$i = \decode{\M}{\R[\rpc]} = \JMP{r}$}]
        {
            \bnotspecial\\
            \macrel{i}{\vopc}{\R}{\B}\\
            \R' = \R[\rpc \mapsto \R[\rn{r}]]\\
            \dwrap{\cycles{i}}{\delta, t, t_a}{\delta', t', t'_a}\\
            \sint{\D}{\cfg{\B}{\delta'}{t'}{t'_a}{\M}{\R'}{\R[\rpc]}}{\cfg{\B'}{\delta''}{t''}{t''_a}{\M'}{\R''}{\R[\rpc]}}
        }
        {
            \smain{\cfg{\B}{\delta}{t}{t_a}{\M}{\R}{\vopc}}{\cfg{\B'}{\delta''}{t''}{t''_a}{\M'}{\R''}{\R[\rpc]}}
        }

        \inferrule*[
        lab={\srulename{(CPU-Reti-Chain)}},
        right={$i = \decode{\M}{\R[\rpc]} = \RETI$}]
        {
            \bnotspecial\\
            \B \neq \BBot \\
            \macrel{i}{\vopc}{\R}{\B}\\
            \dwrap{\cycles{i}}{\delta, t, t_a}{\delta', t', t'_a}\\
            \R[\rsr.\mtt{GIE}] = \mtt{1}\\
            t'_a \neq \bot\\
            \sint{\D}{\cfg{\B}{\delta'}{t'}{t'_a}{\M}{\R}{\R[\rpc]}}{\cfg{\B}{\delta''}{t''}{t''_a}{\M'}{\R'}{\R[\rpc]}}
        }
        {
            \smain{\cfg{\B}{\delta}{t}{t_a}{\M}{\R}{\vopc}}{\cfg{\B}{\delta''}{t''}{t''_a}{\M'}{\R'}{\R[\rpc]}}
        }

        \inferrule*[
        lab={\srulename{(CPU-Reti-PrePad)}},
        right={$i = \decode{\M}{\R[\rpc]} = \RETI$}]
        {
            \bnotspecial\\
            \B \neq \BBot \\
            \macrel{i}{\vopc}{\R}{\B}\\
            \dwrap{\cycles{i}}{\delta, t, t_a}{\delta', t', t'_a}\\
            (\R[\rsr.\mtt{GIE}] = \mtt{0} \ \lor\ t'_a = \bot)\\
        }
        {
            \smain{\cfg{\B}{\delta}{t}{t_a}{\M}{\R}{\vopc}}
            {
                \cfg{\langle \bot, \bot, \B.t_\mi{pad} \rangle}{\delta'}{t'}{t'_a}{\M}{\B.\R}{\B.\vopc}
            }
        }

        \inferrule*[lab={\srulename{(CPU-Reti-Pad)}}]
        {
            \B = \langle \bot, \bot, t_\mi{pad} \rangle \\
            \dwrap{t_\mi{pad}}{\delta, t, t_a}{\delta', t', t'_a}\\
            \sint{\D}{\cfg{\BBot}{\delta'}{t'}{t'_a}{\M}{\R}{\vopc}}{\cfg{\B'}{\delta''}{t''}{t''_a}{\M}{\R'}{\vopc}}\\
        }
        {
            \smain{\cfg{\B}{\delta}{t}{t_a}{\M}{\R}{\vopc}}{\cfg{\B'}{\delta''}{t''}{t''_a}{\M}{\R'}{\vopc}}
        }

        \inferrule*[
        lab={\srulename{(CPU-Reti)}},
        right={$i = \decode{\M}{\R[\rpc]} = \RETI$}]
        {
            \bnotspecial\\
            \macrel{i}{\vopc}{\R}{\bot}\\
            \R' = \R[\rpc \mapsto \M[\R[\rsp] + 2], \rsr \mapsto \M[\R[\rsp]], \rsp \mapsto \R[\rsp] + 4]\\
            \dwrap{\cycles{i}}{\delta, t, t_a}{\delta', t', t'_a}
        }
        {
            \smain{\cfg{\BBot}{\delta}{t}{t_a}{\M}{\R}{\vopc}}{\cfg{\BBot}{\delta'}{t'}{t'_a}{\M}{\R'}{\R[\rpc]}}
        }

        \inferrule*[
        lab={\srulename{(CPU-In)}},
        right={$i = \decode{\M}{\R[\rpc]} = \IN{r}$}]
        {
            \bnotspecial\\
            \macrel{i}{\vopc}{\R}{\B}\\
            \delta \xleadsto{\mi{rd(w)}}_D \delta'\\
            \R' = \R[\rpc \mapsto \R[\rpc] + 2][\rn{r} \mapsto w]\\
            \dwrap{\cycles{i}}{\delta', t, t_a}{\delta'', t', t'_a}\\
            \sint{\D}{\cfg{\B}{\delta''}{t'}{t'_a}{\M}{\R'}{\R[\rpc]}}{\cfg{\B'}{\delta'''}{t''}{t''_a}{\M'}{\R''}{\R[\rpc]}}
        }
        {
            \smain{\cfg{\B}{\delta}{t}{t_a}{\M}{\R}{\vopc}}{\cfg{\B'}{\delta'''}{t''}{t''_a}{\M'}{\R''}{\R[\rpc]}}
        }

        \inferrule*[
        lab={\srulename{(CPU-Out)}},
        right={$i = \decode{\M}{\R[\rpc]} = \OUT{r}$}]
        {
            \bnotspecial\\
            \macrel{i}{\vopc}{\R}{\B}\\
            \R' = \R[\rpc \mapsto \R[\rpc] + 2]\\
            \delta \xleadsto{\mi{wr(\R[\rn{r}])}}_D \delta'\\
            \dwrap{\cycles{i}}{\delta',t, t_a}{\delta'', t', t'_a}\\
            \sint{\D}{\cfg{\B}{\delta''}{t'}{t'_a}{\M}{\R'}{\R[\rpc]}}{\cfg{\B'}{\delta'''}{t''}{t''_a}{\M'}{\R''}{\R[\rpc]}}
        }
        {
            \smain{\cfg{\B}{\delta}{t}{t_a}{\M}{\R}{\vopc}}{\cfg{\B'}{\delta'''}{t''}{t''_a}{\M'}{\R''}{\R[\rpc]}}
        }
\end{mathpar}
\end{scriptsize}

\caption{Rules of the main transition system for \slang. (part II)}
\slbltwo
\end{figure}
~
\begin{figure}[bt]
    \begin{scriptsize}
    \begin{mathpar}
        \inferrule*[
        lab={\srulename{(CPU-Not)}},
        right={$i = \decode{\M}{\R[\rpc]} = \NOT{r}$}]
        {
            \bnotspecial\\
            \macrel{i}{\vopc}{\R}{\B}\\
            \R' = \R[\rpc \mapsto \R[\rpc] + 2][\rn{r} \mapsto \lnot\R[\rn{r}]]\\
            \dwrap{\cycles{i}}{\delta, t, t_a}{\delta', t', t'_a}\\
            \sint{\D}{\cfg{\B}{\delta'}{t'}{t'_a}{\M}{\R'}{\R[\rpc]}}{\cfg{\B'}{\delta''}{t''}{t''_a}{\M'}{\R''}{\R[\rpc]}}
        }
        {
            \smain{\cfg{\B}{\delta}{t}{t_a}{\M}{\R}{\vopc}}{\cfg{\B'}{\delta''}{t''}{t''_a}{\M'}{\R''}{\R[\rpc]}}
        }

        \inferrule*[
        lab={\srulename{(CPU-And)}},
        right={$i = \decode{\M}{\R[\rpc]} = \AND{r_1}{r_2}$}]
        {
            \bnotspecial\\
            \macrel{i}{\vopc}{\R}{\B}\\
            \R' = \R[\rpc \mapsto \R[\rpc] + 2][\rn{r_2} \mapsto \R[\rn{r_1}] \& \R[\rn{r_2}]]\\
            \R'' = \R'[\rsr.\mtt{N} \mapsto \R'[\rn{r_2}] \& \mtt{0x8000}, \rsr.\mtt{Z} \mapsto (\R'[\rn{r_2}] == 0), \rsr.\mtt{C} \mapsto (\R'[\rn{r_2}] \neq 0), \rsr.\mtt{V} \mapsto 0]\\
            \dwrap{\cycles{i}}{\delta, t, t_a}{\delta', t', t'_a}\\
            \sint{\D}{\cfg{\B}{\delta'}{t'}{t'_a}{\M}{\R''}{\R[\rpc]}}{\cfg{\B'}{\delta''}{t''}{t''_a}{\M'}{\R'''}{\R[\rpc]}}
        }
        {
            \smain{\cfg{\B}{\delta}{t}{t_a}{\M}{\R}{\vopc}}{\cfg{\B'}{\delta''}{t''}{t''_a}{\M'}{\R'''}{\R[\rpc]}}
        }

        \inferrule*[
        lab={\srulename{(CPU-Cmp)}},
        right={$i = \decode{\M}{\R[\rpc]} = \CMP{r_1}{r_2}$}]
        {
            \bnotspecial\\
            \macrel{i}{\vopc}{\R}{\B}\\
            \R' = \R[\rpc \mapsto \R[\rpc] + 2][\rn{r_2} \mapsto \R[\rn{r_1}] - \R[\rn{r_2}]]\\
            \R'' = \R'[\rsr.\mtt{N} \mapsto (\R'[\rn{r_2}] < 0), \rsr.\mtt{Z} \mapsto (\R'[\rn{r_2}] == 0), \rsr.\mtt{C} \mapsto (\R'[\rn{r_2}] \neq 0), \rsr.\mtt{V} \mapsto \mi{overflow}(\R[\rn{r_1}] - \R[\rn{r_2}])]\\
            \dwrap{\cycles{i}}{\delta, t, t_a}{\delta', t', t'_a}\\
            \sint{\D}{\cfg{\B}{\delta'}{t'}{t'_a}{\M}{\R''}{\R[\rpc]}}{\cfg{\B'}{\delta''}{t''}{t''_a}{\M'}{\R'''}{\R[\rpc]}}        }
        {
            \smain{\cfg{\B}{\delta}{t}{t_a}{\M}{\R}{\vopc}}{\cfg{\B'}{\delta''}{t''}{t''_a}{\M'}{\R'''}{\R[\rpc]}}
        }

        \inferrule*[
        lab={\srulename{(CPU-Add)}},
        right={$i = \decode{\M}{\R[\rpc]} = \ADD{r_1}{r_2}$}]
        {
            \bnotspecial\\
            \macrel{i}{\vopc}{\R}{\B}\\
            \R' = \R[\rpc \mapsto \R[\rpc] + 2][\rn{r_2} \mapsto \R[\rn{r_1}] + \R[\rn{r_2}]]\\
            \R'' = \R'[\rsr.\mtt{N} \mapsto (\R'[\rn{r_2}] < 0), \rsr.\mtt{Z} \mapsto (\R'[\rn{r_2}] == 0), \rsr.\mtt{C} \mapsto (\R'[\rn{r_2}] \neq 0), \rsr.\mtt{V} \mapsto \mi{overflow}(\R[\rn{r_1}] + \R[\rn{r_2}])]\\
            \dwrap{\cycles{i}}{\delta, t, t_a}{\delta', t', t'_a}\\
            \sint{\D}{\cfg{\B}{\delta'}{t'}{t'_a}{\M}{\R''}{\R[\rpc]}}{\cfg{\B'}{\delta''}{t''}{t''_a}{\M'}{\R'''}{\R[\rpc]}}        }
        {
            \smain{\cfg{\B}{\delta}{t}{t_a}{\M}{\R}{\vopc}}{\cfg{\B'}{\delta''}{t''}{t''_a}{\M'}{\R'''}{\R[\rpc]}}
        }

        \inferrule*[
        lab={\srulename{(CPU-Sub)}},
        right={$i = \decode{\M}{\R[\rpc]} = \SUB{r_1}{r_2}$}]
        {
            \bnotspecial\\
            \macrel{i}{\vopc}{\R}{\B}\\
            \R' = \R[\rpc \mapsto \R[\rpc] + 2][\rn{r_2} \mapsto \R[\rn{r_1}] - \R[\rn{r_2}]]\\
            \R'' = \R'[\rsr.\mtt{N} \mapsto (\R'[\rn{r_2}] < 0), \rsr.\mtt{Z} \mapsto (\R'[\rn{r_2}] == 0), \rsr.\mtt{C} \mapsto (\R'[\rn{r_2}] \neq 0), \rsr.\mtt{V} \mapsto \mi{overflow}(\R[\rn{r_1}] - \R[\rn{r_2}])]\\
            \dwrap{\cycles{i}}{\delta, t, t_a}{\delta', t', t'_a}\\
            \sint{\D}{\cfg{\B}{\delta'}{t'}{t'_a}{\M}{\R''}{\R[\rpc]}}{\cfg{\B'}{\delta''}{t''}{t''_a}{\M'}{\R'''}{\R[\rpc]}}
        }
        {
            \smain{\cfg{\B}{\delta}{t}{t_a}{\M}{\R}{\vopc}}{\cfg{\B'}{\delta''}{t''}{t''_a}{\M'}{\R'''}{\R[\rpc]}}
        }
    \end{mathpar}
\end{scriptsize}

\caption{Rules of the main transition system for \slang. (part III)}
\slblthree
\end{figure}

\subsubsection{Interrupts in \SH}

Intuitively the transition system implements the logic that decides what happens when an interrupt arrives, and its transitions have the following form:
\[
    \shint{\D}{\cfg{\B}{\delta}{t}{t_a}{\M}{\R}{\vopc}}{\cfg{\B'}{\delta'}{t'}{t'_a}{\M'}{\R'}{\vopc}}
\]

\noindent
Interrupts in \SH are \emph{always} ignored, thus the configuration is left unchanged.
\begin{mathpar}
\inferrule*[lab={\shrulename{INT}}]
{ }
{
    \shint{\D}{\cfg{\B}{\delta}{t}{t_a}{\M}{\R}{\vopc}}{\cfg{\B}{\delta}{t}{t_a}{\M}{\R}{\vopc}}
}
\end{mathpar}

\FloatBarrier

  \subsection{\SL}

The operational semantics of \SL\ is given by a main transition system and one for interrupts that are handled securely both in protected and unprotected mode.

\subsubsection{Main transition system}

As above, the main transition system describes how the \SL configurations $c, c'$ evolve during the execution given an I/O device $\D$.
Its transitions have the following form:
\[
    \slmain{c}{c'}
\]

Figures~\ref{fig:sl-maints},~\ref{fig:sl-maints2} and \ref{fig:sl-maints3} report the full set of rules that define the main transition system of \SL.
%

\ifdefined\LANG
\renewcommand{\LANG}{sl}
\else
\newcommand{\LANG}{sl}
\fi

\subsubsection{Interrupts in \SL}

What happens in \SL\ when an interrupt arrives is specified by the transition system with transitions of the form (note that they differ from those of \SH\ only in the arrow, that here is $\slint{\cdot}{\cdot}{\cdot}$)
\[
    \slint{\D}{\cfg{\B}{\delta}{t}{t_a}{\M}{\R}{\vopc}}{\cfg{\B'}{\delta'}{t'}{t'_a}{\M'}{\R'}{\vopc}}
\]
\noindent
The following inference rules incorporate the mitigation described in depth in the paper to handle interrupts also in protected mode (see rule \slrulename{(INT-PM-P)} below).
%
%
\begin{mathpar}
    \inferrule*[lab={\slrulename{(INT-UM-P)}}]
    {
        \moderelUM{\vopc}\\
        \R[\rsr].\mtt{GIE} = 1\\
        t_a \neq \bot\\
        \R' = \R[\rpc \mapsto \mi{isr}, \rsr \mapsto 0, \rsp \mapsto \R[\rsp]-4]\\
        \M' = \M[\R[\rsp]-2 \mapsto \R[\rpc], \R[\rsp]-4 \mapsto \R[\rsr]]\\
        \dwrap{6}{\delta, t, \bot}{\delta', t', t'_a}
    }
    {
        \slint{\D}{\cfg{\B}{\delta}{t}{t_a}{\M}{\R}{\vopc}}{\cfg{\B}{\delta'}{t'}{t'_a}{\M'}{\R'}{\vopc}}
    }

    \inferrule*[lab={\slrulename{(INT-UM-NP)}}]
    {
        \moderelUM{\vopc}\\
        (\R[\rsr].\mtt{GIE} = \mtt{0} \lor t_a = \bot)
    }
    {
        \slint{\D}{\cfg{\B}{\delta}{t}{t_a}{\M}{\R}{\vopc}}{\cfg{\B}{\delta}{t}{t_a}{\M}{\R}{\vopc}}
    }

    \inferrule*[lab={\slrulename{(INT-PM-P)}}]
    {
        k = \MT - (t - t_a)\\
        \moderelPM{\vopc}\\
        \R[\rsr].\mtt{GIE} = \mtt{1}\\
        t_a \neq \bot\\
        \R' = \R_0[\rpc \mapsto \mi{isr}]\\
        \dwrap{6+k}{\delta, t, \bot}{\delta', t', t'_a}\\
        \B' = \langle \R, \vopc, t - t_a \rangle
    }
    {
        \slint{\D}{\cfg{\B}{\delta}{t}{t_a}{\M}{\R}{\vopc}}{\cfg{\B'}{\delta'}{t'}{\bot}{\M}{\R'}{\vopc}}
    }

    \inferrule*[lab={\slrulename{(INT-PM-NP)}}]
    {
        \moderelPM{\vopc}\\
        (\R[\rsr].\mtt{GIE} = \mtt{0} \lor t_a = \bot)
    }
    {
        \slint{\D}{\cfg{\B}{\delta}{t}{t_a}{\M}{\R}{\vopc}}{\cfg{\B}{\delta}{t}{t_a}{\M}{\R}{\vopc}}
    }
\end{mathpar}
%
It might be worthy to briefly describe what happens upon ``corner cases'':
\begin{itemize}
    \item Whenever an interrupt has to be handled in protected mode, but the current instruction lead the CPU in unprotected mode, the padding mechanism is applied as in the standard case \emph{including} the padding after the $\RETI$. Indeed, if partial padding (resp.~no padding at all) was applied then the duration of the padding (resp.~of the last instruction) would be leaked to the attacker (cf.~definition below).
    \item Interrupts arising during the padding \emph{before} the interrupt service routine is invoked need to be ignored, since the padding duration and the instruction duration would be leaked otherwise (cf. definition below, rule~\slrulename{(INT-PM-P)} ignores any interrupt happening during the cycles needed for the interrupt logic and for the padding).
    \item Interrupts happening \emph{during} the execution of the interrupt service routine are simply ``chained'' and handled as soon the current routine is completed (see rule~\slrulename{(CPU-Reti-Chain)}).
    \item Finally, interrupts happening during the padding \emph{after} the interrupt service routine are handled as any other interrupt happening in protected mode (see rule~\slrulename{(CPU-Reti-Pad)}).
\end{itemize}

\FloatBarrier

\section{Security theorems}\label{app:security}

Security of \SL is obtained by proving it fully abstract w.r.t.\ \SH.
We define full abstraction here relying on the convergence of whole programs.

\begin{definition}
    Let $C = \langle \M_C, \D \rangle$ be a context, and $\M_M$ be a software module.
    A whole program $C[\M_M]$ converges in \SH (written $\sconv{C[\M_M]}$) \emph{iff}
    \[
        \shmainstar{\D}{\initconf{C}{\M_M}}{\haltconf}. 
    \]

    \noindent
    Similarly, the same whole program converges in \SL (written $\tconv{C[\M_M]}$) \emph{iff}
    \[
        \slmainstar{\D}{\initconf{C}{\M_M}}{\haltconf}. 
    \]
\end{definition}

The following definition formalizes the notion of contextual equivalence of two software modules.
Recall from the paper that contextually equivalent software modules behave in the same way under any attacker (i.e., context).
\begin{definition}
    Two software modules $\M_M$ and $\M_{M'}$ are contextually equivalent in \SH, written $\M_M \seq \M_{M'}$, \emph{iff}
    \[
        \forall C .\ \left(\sconv{C[\M_M]} \iff \sconv{C[\M_{M'}]} \right).
    \]

    Similarly, two software modules $\M_M$ and $\M_{M'}$ are contextually equivalent in \SL, written $\M_M \teq \M_{M'}$, \emph{iff}
    \[
        \forall C .\ \left(\tconv{C[\M_M]} \iff \tconv{C[\M_{M'}]} \right).
    \]
\end{definition}

\begin{restatable}[Full abstraction]{thm}{fullabstractionrep}\label{thm:fa-rep}
    $\forall \M_M, \M_{M'}.\ (\M_M \seq \M_{M'} \iff \M_M \teq \M_{M'})$.
\end{restatable}

For proving the full abstraction theorem we first easily establish that
$(\M_M \seq \M_{M'} \Leftarrow \M_M \teq \M_{M'})$ (Lemma~\ref{lemma:reflection}), i.e.\ reflection of behavious.
Then, the other implication, i.e.\ preservation of behaviours is proved by Lemma~\ref{lemma:preservation} following the strategy summarized in~\figurename~\ref{fig:strategy-rep}.
There we use the trace equivalence $\ttreq$ of Definition~\ref{def:trace-eq}.
Intuitively, we say that a module $M$ plugged in a context performs a trace made of those actions performed by $M$ that can be observed by an attacker, i.e.\ when a call to $M$ occurs and when instead $M$ returns;
also information about the contents of the registers will be recorded in both cases, and also on the flow of time in the second case.
Two modules are then equivalent if they exhibit the same traces.
Proving preservation is then done in two steps, the composition of which gives $(iii)$ in~\figurename~\ref{fig:strategy-rep}.
First Lemma~\ref{lemma:pres-tr} establishies $(ii)$ in~\figurename~\ref{fig:strategy-rep}: two modules equivalent in \SH\ are trace equivalent.
Then Lemma~\ref{lemma:refl-tr} establishes $(i)$ in~\figurename~\ref{fig:strategy-rep}: two modules that are trace equivalent are also equivalent in \SL.

\begin{figure}
    \centering
    \begin{tikzpicture}
        \matrix (m) [matrix of math nodes,row sep=3em,column sep=4em,minimum width=2em] {
        \M_M \seq \M_{M'} & \\
        \M_M \teq \M_{M'} & \M_M \ttreq \M_{M'} \\};
        \path[-stealth]
          (m-2-2) edge [double] node [below] {$(i)$} (m-2-1)
          (m-1-1) edge [double] node [above] {$(ii)$} (m-2-2)
          (m-1-1) edge [double] node [left] {$(iii)$} (m-2-1);
    \end{tikzpicture}

    \caption{An illustration of the proof strategy of preservation of behaviours.}\label{fig:strategy-rep}
\end{figure}
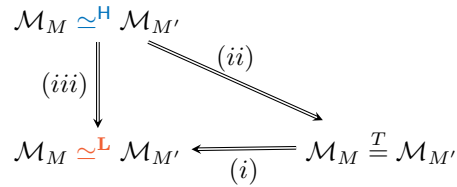

\subsection{Reflection of behaviors}

To prove the reflection of behaviors, i.e., that for all $\M_M, \M_{M'}.\ \M_M \teq \M_{M'}$ implies $\M_M \seq \M_{M'}$ we first need to introduce the notion of \emph{interrupt-less context} $\nI{C}$ for a context $C$.
Intuitively, $\nI{C}$ behaves as $C$ but never raises any interrupt.
In practice, we obtain it from $C$ by removing in the device the transitions that may raise an interrupt.
Formally:
\begin{definition}\label{def:nointctx}
    Let $\D = \langle \Delta, \deltainit, \xleadsto{a}_D \rangle$ be an I/O device.
    Given a context $C = \langle \M_C, \D \rangle$, we define its corresponding \emph{interrupt-less context} as $\nI{C} = \langle \M_C, \nI{\xleadsto{a}_D} \rangle$ where:
    \begin{itemize}
        \item $\nI{\D} = \langle \Delta, \deltainit, \nI{\xleadsto{a}_D} \rangle$, and
        \item $\nI{\xleadsto{a}_D} \triangleq\ \xleadsto{a}_D \, \cup\, \{ (\delta, \epsilon, \delta') \mid (\delta, \mi{int?}, \delta') \in\ \xleadsto{a}_D\} \setminus \{ (\delta, \mi{int?}, \delta') \mid (\delta, \mi{int?}, \delta') \in\ \xleadsto{a}_D\}.$
    \end{itemize}
\end{definition}
\noindent
Note that $\nI{\D}$ is actually a device, due to the contraints on its transition function.

The behavior of interrupt-less contexts in \SL has a direct correspondence to the behavior of their standard counterparts in \SH (recall that \SH ignores all the interrupts).
In fact:
\begin{lemma}\label{lemma:nointeq}
    For any module $\M_M$, context $C$, and corresponding interrupt-less context $\nI{C}$:
    \[
        \tconv{\nI{C}[\M_M]} \iff \sconv{C[\M_M]}
    \]
\end{lemma}
\begin{proof}
    By definition of $\dwrap{k}{\cdot}{\cdot}$, the value $t_a$ in the CPU configuration (that signals the presence of an unhandled interrupt) is changed only when an interrupt has been raised since the last time it was checked.

    Since any $\mi{int?}$ action has been substituted with an $\epsilon$, $t_a$ is never changed from its initial $\bot$ value.

    Since the only difference in behavior between the two levels is in the interrupt logic, and since the ISR in $\nI{C}$ is never invoked (thus, it does not affect the program behavior), $\slint{\D}{\cdot}{\cdot}$ behaves exactly as $\shint{\D}{\cdot}{\cdot}$.
    So, $\tconv{\nI{C}[\M_M]}$ implies $\sconv{C[\M_M]}$ and vice versa.
\end{proof}

Given Definition~\ref{def:nointctx} and Lemma~\ref{lemma:nointeq} it is relatively easy to prove reflection, since whole programs in \SH behave just like a subset of whole programs in \SL:
\begin{lemma}[Reflection]\label{lemma:reflection}
\[
    \forall \M_M, \M_{M'}.\ (\M_M \teq \M_{M'} \implies \M_M \seq \M_{M'}).
\]
\end{lemma}
\begin{proof}
We can expand the hypothesis using the definition of $\teq$ and $\seq$ as follows:
\[
    (\forall C.\, \tconv{C[\M_M]} \iff \tconv{C[\M_{M'}]}) \implies (\forall C'.\, \sconv{C'[\M_M]} \iff \sconv{C'[\M_{M'}]}).
\]

For any $C'$ we can build the corresponding interrupt-less context $C'_{\not\,I}$.

Since interrupt-less contexts are a (strict) subset of all the contexts, by hypothesis:
\[
    \tconv{C'_{\not\,I}[\M_M]} \iff \tconv{C'_{\not\,I}[\M_{M'}]}.
\]

But from Lemma~\ref{lemma:nointeq} it follows that
\[
    \sconv{C'[\M_M]} \iff \sconv{C'[\M_{M'}]}.
\]
\end{proof}

\subsection{Preservation of behaviors}
The preservation of behaviors is stated as follows:
\begin{lemma}\label{lemma:preservation}
\[
    \forall \M_M, \M_{M'}.\ (\M_M \seq \M_{M'} \Rightarrow  \M_M \teq \M_{M'}).
\]
\end{lemma}
Its proof is harder than the one of reflection and requires the definition of a trace semantics whose traces, intuitively, correspond to the behaviors that an attacker can observe in \SL.

\subsubsection{Fine-grained and coarse-grained trace semantics}\label{sssec:tracesem}

To simplify the extraction of the traces we first define a very fine-grained trace semantics and then we transform it to a more coarse-grained one to match what attackers can observe.

The fine-grained trace semantics has the following observables ($k \in \mb{N}$):
\begin{align*}
    \ta \Coloneqq&\
    \usilent \mid \psilent{k} \mid \cnv\\
     &\ \jmpin{\R} \mid \jmpout{k}{\R} \\
     &\ \reti{k} \mid \handle{k}.
\end{align*}
Traces are defined as strings of observables $\ta$, and we denote the empty trace as $\emptystr$.

Intuitively, $\usilent$ denotes actions performed by the context that are not observed, $\psilent{k}$ indicates an internal action taking $k$ cycles.
The observable $\cnv$ indicates that termination occurred.
A $\jmpin{\R}$ happens when the CPU enters protected mode, $\jmpout{k}{\R}$ happens when it exits.
Finally, $\handle{k}$ and $\reti{k}$ denote when the processor starts executing the interrupt service routine from protected mode and when it returns from it, respectively.
\begin{figure}
\begin{mathpar}
    \inferrule*
    [
        lab={(Obs-Internal-PM)}
    ]
    {
        \moderelPM{\R[\rpc]}\\
        \slmain{\cfg{\B}{\delta}{t}{t_a}{\M}{\R}{\vopc}}{\cfg{\BBot}{\delta'}{t+k}{t'_a}{\M'}{\R'}{\vopcp}}\\
        \moderelPM{\R'[\rpc]}\\
    }
    {
        \atrsem{\D}{\cfg{\B}{\delta}{t}{t_a}{\M}{\R}{\vopc}}{\psilent{k}}{\cfg{\BBot}{\delta'}{t+k}{t'_a}{\M'}{\R'}{\vopcp}}
    }

    \inferrule*
    [
        lab={(Obs-JmpIn)}
    ]
    {
        \moderelUM{\R[\rpc]}\\
        \slmain{\cfg{\BBot}{\delta}{t}{t_a}{\M}{\R}{\vopc}}{\cfg{\BBot}{\delta'}{t'}{t'_a}{\M'}{\R'}{\vopcp}}\\
        \moderelPM{\R'[\rpc]}\\
    }
    {
        \atrsem{\D}{\cfg{\BBot}{\delta}{t}{t_a}{\M}{\R}{\vopc}}{\jmpin{\R'}}{\cfg{\BBot}{\delta'}{t'}{t'_a}{\M'}{\R'}{\vopcp}}
    }

    \inferrule*
    [
        lab={(Obs-Reti)}
    ]
    {
        \moderelUM{\R[\rpc]}\\
        \B \neq \BBot\\
        \slmain{\cfg{\B}{\delta}{t}{t_a}{\M}{\R}{\vopc}}{\cfg{\langle \bot, \bot, t_\mi{pad} \rangle}{\delta'}{t+k}{t'_a}{\M'}{\R'}{\vopcp}}\\
    }
    {
        \atrsem{\D}{\cfg{\B}{\delta}{t}{t_a}{\M}{\R}{\vopc}}{\reti{k}}{\cfg{\langle \bot, \bot, t_\mi{pad} \rangle}{\delta'}{t'}{t'_a}{\M'}{\R'}{\vopcp}}
    }

    \inferrule*
    [
        lab={(Obs-JmpOut)}
    ]
    {
        \moderelPM{\R[\rpc]}\\
        \slmain{\cfg{\BBot}{\delta}{t}{t_a}{\M}{\R}{\vopc}}{\cfg{\BBot}{\delta'}{t+k}{t'_a}{\M'}{\R'}{\vopcp}}\\
        \moderelUM{\R'[\rpc]}\\
    }
    {
        \atrsem{\D}{\cfg{\BBot}{\delta}{t}{t_a}{\M}{\R}{\vopc}}{\jmpout{k}{\R'}}{\cfg{\B'}{\delta'}{t+k}{t'_a}{\M'}{\R'}{\vopcp}}
    }

    \inferrule*
    [
        lab={(Obs-JmpOut-PostPoned)}
    ]
    {
        \moderelUM{\R[\rpc]}\\
        \slmain{\cfg{\langle \bot, \bot, t_\mi{pad} \rangle}{\delta}{t}{t_a}{\M}{\R}{\vopc}}{\cfg{\BBot}{\delta'}{t+k}{t'_a}{\M'}{\R'}{\vopcp}}\\
        \moderelUM{\R'[\rpc]}\\
    }
    {
        \atrsem{\D}{\cfg{\langle \bot, \bot, t_\mi{pad} \rangle}{\delta}{t}{t_a}{\M}{\R}{\vopc}}{\jmpout{k}{\R'}}{\cfg{\B'}{\delta'}{t+k}{t'_a}{\M'}{\R'}{\vopcp}}
    }

    \inferrule*
    [
        lab={(Obs-Handle)}
    ]
    {
        \slmain{\cfg{\B}{\delta}{t}{t_a}{\M}{\R}{\vopc}}{\cfg{\B'}{\delta'}{t+k}{t'_a}{\M'}{\R'}{\vopcp}}\\
        \moderelUM{\R'[\rpc]}\\
        \B' \neq \BBot
    }
    {
        \atrsem{\D}{\cfg{\B}{\delta}{t}{t_a}{\M}{\R}{\vopc}}{\handle{k}}{\cfg{\B'}{\delta'}{t+k}{t'_a}{\M'}{\R'}{\vopcp}}
    }

    \inferrule*
    [
        lab={(Obs-Internal-UM)}
    ]
    {
        \moderelUM{\R[\rpc]}\\
        \slmain{\cfg{\B}{\delta}{t}{t_a}{\M}{\R}{\vopc}}{\cfg{\B}{\delta'}{t'}{t'_a}{\M'}{\R'}{\vopcp}}\\
        \moderelUM{\R'[\rpc]}\\
    }
    {
        \atrsem{\D}{\cfg{\B}{\delta}{t}{t_a}{\M}{\R}{\vopc}}{\usilent}{\cfg{\B}{\delta'}{t'}{t'_a}{\M'}{\R'}{\vopcp}}
    }

    \inferrule*
    [
        lab={(Obs-Final)}
    ]
    {
        \moderelUM{\R[\rpc]}\\
        \slmain{\cfg{\B}{\delta}{t}{t_a}{\M}{\R}{\vopc}}{\haltconf}
    }
    {
        \atrsem{\D}{\cfg{\B}{\delta}{t}{t_a}{\M}{\R}{\vopc}}{\cnv}{\haltconf}
    }
\end{mathpar}

\caption{Formal definition of relation $\atrarrow{\ta}$ for fine-grained observables.}\label{fig:finegrained-obs}
\end{figure}

The relation $\atrarrow{\ta}$ in~\figurename~\ref{fig:finegrained-obs} formally defines how observables can be extracted from the execution of a whole program.
It is worth noting that the relation $\atrarrow{\ta}$ is defined in such a way that each transition $\slmain{c}{c'}$ has a corresponding transition $\atrsem{\D}{c}{\ta}{c'}$ for some $\ta$,
possibly the non observable one, $\xi$.

Fine-grained traces $\bar \alpha$ are obtained by transitively and reflexively closing $\atrarrow{\ta}$, written $\atrarrowstar{\bar\ta}$.
Note that in any trace $\bar\ta$, only the observables $\psilent{k}, \reti{k}$ or $\handle{k}$ can occur between a $\jmpin{\R}$ and a $\jmpout{\dt}{\R}$.

When an interrupt has to be handled, the trace that is observed starts with an $\handle{\cdot}$, followed by a sequence of $\usilent$ and, if a $\RETI$ is executed, a $\reti{k}$ ($k$ always has value $\cycles{\RETI}$) is observed.

If the interrupted instruction was a jump from protected mode to unprotected mode, the $\reti{\cdot}$ is followed by a $\jmpout{\cdot}{\cdot}$ (cf. rules~\rulename{(Obs-Handle)},~\rulename{(Obs-Internal-UM)},~\rulename{(Obs-Reti)} and~\rulename{(Obs-JmpOut-PostPoned)}), otherwise a $\psilent{\cdot}$ -- or a $\handle{\cdot}$ if an interrupt has to be handled -- is observed.

Actually, these traces contain more information than what an attacker (i.e., the context) can observe.
To match what the context can observe we introduce more coarse-grained traces with the following observables, where $\jmpin{\R}$ and $\jmpout{\dt}{\R}$ represent invoking a module and returning from it:
\begin{align*}
    \tb \Coloneqq&\ \cnv \mid \jmpin{\R} \mid \jmpout{\dt}{\R}.
\end{align*}
Traces $\btr$ are defined as strings of $\tb$ actions with $\emptystr$ as the empty trace.

Note that observables for interrupts and silent actions are not visible anymore.
In addition, $\jmpout{\dt}{\R}$ has a $\dt$ parameter that models that an attacker can just measure the end-to-end time of a piece of code running in protected mode.
\begin{definition}[Traces of a module]
    The set of (observable) traces of the module $M$ is
    \[
        \ttr{\M_M} \triangleq \{ \btr \mid \exists C = \langle \M_C, \D \rangle.\, \btrsemstar{\D}{\initconf{C}{\M_M}}{\btr}{c'} \}.
    \]
    where $\btrarrowstar{\cdot}$ is the reflexive and transitive closure of the $\btrarrow{\cdot}$ relation defined in~\figurename~\ref{fig:coarsegrained-obs}.
\end{definition}
\begin{figure}


    \begin{mathpar}
        \inferrule
        {
            \atrsemstar{\D}{\initconf{C}{\M_M}}{\usilent \cdots \usilent \cdot \jmpin{\R}}{c}
        }
        {
            \btrsem{\D}{\initconf{C}{\M_M}}{\jmpin{\R}}{c}
        }

        \inferrule
        {
            \atrsemstar{\D}{\initconf{C}{\M_M}}{\usilent \cdots \usilent \cdot \cnv}{\haltconf}
        }
        {
            \btrsem{\D}{\initconf{C}{\M_M}}{\cnv}{\haltconf}
        }

        \inferrule
        {
            \exists c.\, \btrsem{\D}{c}{\jmpout{\dt}{\R'}}{c'}\\
            \atrsemstar{\D}{c'}{\usilent \cdots \usilent \cdot \jmpin{\R''}}{c''}
        }
        {
            \btrsem{\D}{c'}{\jmpin{\R''}}{c''}
        }

        \inferrule
        {
            \exists c.\, \btrsem{\D}{c}{\jmpout{\dt}{\R'}}{c'}\\
            \atrsemstar{\D}{c'}{\usilent \cdots \usilent \cdot \cnv}{\haltconf}
        }
        {
            \btrsem{\D}{c'}{\cnv}{\haltconf}
        }

        \inferrule
        {
            \exists c.\, \btrsem{\D}{c}{\jmpin{\R'}}{c'}\\
            \atrsemstar{\D}{c'}{\ta^{(0)} \cdots \ta^{(n-1)} \cdot \jmpout{k''}{\R''}}{c''}\\
            \forall 0 \leq i < n.\, \ta_i \notin \{ \jmpout{\_}{\_}, \cnv \}\\
            \dt = k'' + \sum_{i=0}^{n-1}\mi{time}(\ta^{(i)})
        }
        {
            \btrsem{\D}{c'}{\jmpout{\dt}{\R''}}{c''}
        }

        \inferrule
        {
            \exists c.\, \btrsem{\D}{c}{\jmpin{\R'}}{c'}\\
            \atrsemstar{\D}{c'}{\ta_0 \cdots \ta_{n-1}  \cdot \cnv}{\haltconf}\\
            \forall 0 \leq i < n.\, \ta_i \notin \{ \jmpout{\_}{\_}, \cnv \}\\
        }
        {
            \btrsem{\D}{c'}{\cnv}{\haltconf}
        }
    \end{mathpar}

    where
    \[
        \mi{time}(\ta) =
        \begin{cases}
            k & \text{ if } \ta \in \{ \reti{k}, \handle{k}, \psilent{k}, \jmpout{k}{\R} \}\\
            0 & \text{o.w.}
        \end{cases}
    \]
    \caption{Formal definition of relation $\btrarrow{\btr}$ for coarse-grained observables.}\label{fig:coarsegrained-obs}
\end{figure}

We eventually define when two modules are trace equivalent:
\begin{definition}\label{def:trace-eq}
Two modules are \emph{(coarse-grained) trace equivalent}, written $\M_M \ttreq \M_{M'}$, \emph{iff}
\[
    \ttr{\M_M} = \ttr{\M_{M'}}.
\]
\end{definition}

\medskip
\paragraph{Notation.}
If not specified, let $x \in \{1, 2\}$, in the rest of the report.
Moreover, beside using $c, c_1, c_2, \dots$, possibly dashed, to denote configurations, we will write
$c^{(n)}_x = \cfg{\B^{(n)}_x}{\delta^{(n)}_x}{t^{(n)}_x}{t^{(n)}_{a_x}}{\M^{(n)}_x}{\R^{(n)}_x}{\vopc^{(n)}_x}$ for the configuration reached after $n$ execution steps from the initial configuration $c^{(0)}_x$.
Similarly, the components of a context $C_x$ will be accordingly indexed.
Also, we will denote with $c^{(i)}_x$ the configuration \emph{right before} the action of index $i$ in a given fine or coarse-grained trace.

Finally, we define some notions and prove a property that will be of use in the rest of the report.
The first definition defines a partitioning of fine-grained traces in sub-traces that correspond to handling interrupts and those that are not.
We call \emph{(complete) interrupt segments} those starting with an $\handle{\cdot}$ action (in the $i^{th}$ position in the given trace) and ending with a $\reti{\cdot}$ action (in the $j^{th}$ position).
In this way the set of interrupt segments is a set of pairs $(i, j)$, as defined below.
\begin{definition}[Complete interrupt segments]\label{def:completely-handled-ints}
    Let $\atr = \ta_0 \ \cdots\ \ta_n$ be a fine-grained trace.
    The set  $\mb{I}_\atr$ of \emph{complete interrupt segments} of $\atr$ is defined as follows:
    \[
        \mb{I}_\atr \triangleq \{ (i, j) \mid \ta_i = \handle{k}\ \land\ \ta_j = \reti{k'}\ \land\ i < j \ \land\ \forall i < l < j .\ \ta_l = \usilent \}.
    \]
\end{definition}

The second definition expresses the time taken by the current protected-mode instruction in the given configuration to be executed.
\begin{definition}
    We define the \emph{length of the current protected-mode instruction} in configuration $c$ as
    \[
        \ilen{c} \triangleq \begin{cases}
            \cycles{\decode{\M}{\R[\rpc]}} & \text{if } \moderelPM{c} \,\land\, \B = \bot\\
            0 & \text{o.w.}
        \end{cases}
    \]
\end{definition}
\begin{property}\label{prop:btr-timings}
    If $\moderelPM{c^{(0)}}$ and $\atrsemstar{\D}{c^{(0)}}{\atr}{c^{(n+1)}}$,
    with $\atr = \ta^{(0)} \cdots \ta^{(n-1)} \cdot \jmpout{k^{(n)}}{\R'}$,
    then
    $k + \sum_{i=0}^{n-1} \mi{time}(\ta^{(i)}) = \sum_{i=0}^{n} \ilen{c^{(i)}} + (11 + \MT)\cdot{|\mb{I}_{\atr}|}$.
\end{property}
\begin{proof}
    By definition of the interrupt logic and the operational semantics of \SL, for each interrupt handled in protected mode we perform a $0 \leq k \leq \MT$ padding \emph{before} invoking the interrupt service routine and an additional padding of $(\MT - k)$ cycles \emph{after} its execution, i.e., the padding time introduced for each complete interrupt segment amounts to $\MT$.
    Also, since the interrupt logic always requires $6$ cycles to jump to the interrupt service routine and $5$ cycles are required upon $\RETI$ it easily follows that:
    \[
        k + \sum_{i=0}^{n-1} \mi{time}(\ta^{(i)}) = \sum_{i=0}^{n} \ilen{c^{(i)}} + (11 + \MT)\cdot{|\mb{I}_{\atr}|}.
    \]
\end{proof}

Before we move to the actual proof of preservation of behaviours, it is convenient introducing two relations (actually, two equivalences) between configurations and to establish a number of useful properties.
Roughly, the equivalences holds two configurations cannot be kept apart by looking at those parts that can be inspected when the CPU is operating in either protected mode or unprotected mode, respectively.
\begin{definition}
    We say that two configurations are \emph{$P$-equivalent} (written $c \peq c'$) \emph{iff}
    \begin{align*}
        &(c = c' = \haltconf)\ \lor\\
        &(c = \cfg{\B}{\delta}{t}{t_a}{\M}{\R}{\vopc}\ \land\
        c' = \cfg{\B'}{\delta'}{t'}{t'_a}{\M'}{\R'}{\vopcp} \ \land\
        \M \pmem \M' \ \land\\
        &\qquad\qquad\qquad
        \moderel{\vopc}\ \land \ \moderel{\vopcp} \ \land\
        \R \preg{m}  \R' \ \land\
        \B \bowtie \B')\\
    \end{align*}
    where
    \begin{itemize}
        \item $\M \pmem \M'$ iff $\forall l \in [\mi{ts}, \mi{te}) \cup [\mi{ds}, \mi{de}).\ M[l] = \M'[l]$.
        \item $\R \preg{m} \R'$ iff $(\mtt{m} = \mtt{PM} \implies \R = \R')$
        \item $\B \bowtie \B'$ iff $(\B \neq \BBot \ \land\  \B' \neq \BBot) \lor (\B = \B' = \BBot)$.
    \end{itemize}
\end{definition}
\begin{definition}
    We say that two configurations are \emph{$U$-equivalent} (written $c \ueq c'$) \emph{iff}
    \begin{align*}
        &(c = c' = \haltconf)\ \lor\\
        &(c = \cfg{\B}{\delta}{t}{t_a}{\M}{\R}{\vopc}\ \land\
        c' = \cfg{\B'}{\delta'}{t'}{t'_a}{\M'}{\R'}{\vopcp} \ \land\
        \M \umem \M' \ \land\\
        &\qquad\qquad\qquad
        \moderel{c}\ \land \ \moderel{c'} \ \land\ \delta = \delta' \ \land\ t = t' \ \land\ t_a = t'_a \ \land\
        \R \ureg{m}  \R' \ \land\
        \B \bowtie \B')\\
    \end{align*}
    where
    \begin{itemize}
        \item $\M \umem \M'$ iff $\forall l \not\in [\mi{ts}, \mi{te}) \cup [\mi{ds}, \mi{de}).\ M[l] = \M'[l]$
        \item $\R \ureg{m} \R'$ iff $(\mtt{m} = \mtt{UM} \implies \R = \R') \ \land\ \R[\rsr.\mtt{GIE}] = \R'[\rsr.\mtt{GIE}]$
        \item $\B \bowtie \B'$ iff $(\B \neq \BBot \ \land\  \B' \neq \BBot) \lor (\B = \B' = \BBot)$.
    \end{itemize}
\end{definition}
\begin{property}
    Both $\peq$ and $\ueq$ are equivalence relations.
\end{property}
\begin{proof}
    Trivial.
\end{proof}

\subsubsection{Properties of \texorpdfstring{$P$}{P}-equivalence}

The first property says that if a configuration can take a step, also another P-equivalent configuration can.
\begin{property}\label{prop:peq-samedecode}
    If $c_1 \peq c_2$, $\moderelPM{c_1}$, $\slmaind{\D'}{c_1}{c'_1}$ then $\decode{\M_1}{\R_1[\rpc]} = \decode{\M_2}{\R_2[\rpc]}$ and $\slmaind{\D'}{c_2}{c'_2}$.
\end{property}
\begin{proof}
    Since $c_1 \peq c_2$ and $\moderelPM{c_1}$, it also holds that $\moderelPM{c_2}$.
    Also, the instruction $\decode{\M_1}{\R_1[\rpc]}$ is decoded in both $\M_1$ and $\M_2$ at the same protected address, hence $\decode{\M_1}{\R_1[\rpc]} = \decode{\M_2}{\R_2[\rpc]}$, and $\slmaind{\D'}{c_2}{c'_2}$.
\end{proof}
\begin{property}\label{prop:peq-main-peq}
    If $c_1 \peq c_2$, $\moderelPM{c_1}$, $\slmaind{\D}{c_1}{c'_1}$, $\slmaind{\D'}{c_2}{c'_2}$ and $\B'_1 \bowtie \B'_2$ then $c'_1 \peq c'_2$.
\end{property}
\begin{proof}
    Since $c_1 \peq c_2$, $\moderelPM{c_1}$ and $\slmaind{\D}{c_1}{c'_1}$, by Property~\ref{prop:peq-samedecode}, $i = \decode{\M_1}{\R_1[\rpc]} = \decode{\M_2}{\R_2[\rpc]}$ and $\slmaind{\D'}{c_2}{c'_2}$.

    Sinc $\B'_1 \bowtie \B'_2$, we have two cases:
    \begin{enumerate}
        \item \emph{Case $\B'_1 = \B'_2 = \BBot$.}
            In this case we know that no interrupt handling started during the step, and by exhaustive cases on $i$ we can show $c'_1 \peq c'_2$:
            \begin{itemize}
                \item \emph{Case $i \in \{ \HLT,  \IN{r}, \OUT{r} \}$}.
                    In both cases we have $c'_1 = \exconf{c_1} \peq \exconf{c_2} = c'_2$.
                \item \emph{Otherwise}.
                    The relevant values in $c'_1$ and $c'_2$ just depend on values that coincide also in $c_1$ and $c_2$.
                    Hence, by determinism of the rules, we get $c'_1 \peq c'_2$.
            \end{itemize}
        \item \emph{Case $\B'_1 \neq \BBot$ and $\B'_2 \neq \BBot$}.
            In this case an interrupt was handled, but the same instruction was indeed executed in protected mode, hence $\M'_1 \pmem \M'_2$.
            Also, $\R'_1 \preg{UM} \R'_2$ holds trivially, $\B'_1 \bowtie \B'_2$ by hypothesis and $\moderelUM{\vopcp_1}$ and $\moderelUM{\vopcp_2}$.
            Thus, $c'_1 \peq c'_2$.
    \end{enumerate}
\end{proof}

Some sequences of fine-grained traces preserve $P$-equivalence.
\begin{property}\label{prop:peq-epsilon-jmpin-peq}
    If $c_1 \peq c_2$,
    $\atrsemstar{\D}{c_1}{\soverbrace{\usilent\ \cdots\ \usilent}^{\ell_1}}{c'_1 \atrarrow{\jmpin{\R}} c''_1}$,
    $\atrsemstar{\D'}{c_2}{\soverbrace{\usilent\ \cdots\ \usilent}^{\ell_2}}{c'_2 \atrarrow{\jmpin{\R}} c''_2}$,
    then $c''_1 \peq c''_2$.
\end{property}
\begin{proof}
    We show by Noetherian induction over $(\ell_1, \ell_2)$ that $\M'_1 \pmem \M'_2$.
    For that, we use well-founded relation $(\ell_1, \ell_2) \prec (\ell'_1, \ell'_2)$ \emph{iff} $\ell_1 < \ell'_1 \land \ell_2 < \ell'_2$.

    \begin{itemize}
        \item \emph{Case $(0, 0)$.}
            Trivial.
        \item \emph{Case $(0, \ell_2)$, with $\ell_2 > 0$. (and symmetrically $(\ell_1, 0)$, with $\ell_1 > 0$)}
            We have to show that
            \[
                \atrsemstar{\D}{c_1}{\emptystr}{c'_1}
                \land
                \atrsemstar{\D'}{c_2}{\soverbrace{\usilent\ \cdots\ \usilent}^{\ell_2}}{c'_2}
                \Rightarrow
                \M'_1 \pmem \M'_2
            \]
            Since from $c_1$ there is no step, $c_1 = c'_1$.
            Moreover a sequence of $\usilent$ was observed starting from $c_2$, and since both configurations are in unprotected mode and no violation occurred (see Table~\ref{tab:mac-rep}) the protected memory is unchanged.
            Thus, by transitivity of $\pmem$, we have $\M'_1 = \M_1 \pmem \M_2 \pmem \M'_2$.
        \item \emph{Case $(\ell_1, \ell_2) = (\ell'_1 + 1, \ell'_2 + 1)$.}
            If
            \[
                \atrsemstar{\D}{c_1}{\soverbrace{\usilent\ \cdots\ \usilent}^{\ell'_1}}{c'''_1}
                \land
                \atrsemstar{\D'}{c_2}{\soverbrace{\usilent\ \cdots\ \usilent}^{\ell'_2}}{c'''_2}\Rightarrow
                \M'''_1 \pmem \M'''_2 \text{ (IHP)}
            \]
            then
            \[
                \atrsemstar{\D}{c_1}{\soverbrace{\usilent\ \cdots\ \usilent}^{\ell'_1}}{c'''_1 \atrarrow{\usilent} c'_1}
                \land
                \atrsemstar{\D'}{c_2}{\soverbrace{\usilent\ \cdots\ \usilent}^{\ell'_2}}{c'''_2 \atrarrow{\usilent} c'_2}\Rightarrow
                \M'_1 \pmem \M'_2.
            \]

            By (IHP) we know that $\M'''_1 \pmem \M'''_2$.
            Indeed, since we observed $\usilent$ it means that $\moderel{\vopc'_1} \ \land\ \moderel{\vopc'_2}$.
            Moreover (see Figure~\ref{fig:finegrained-obs}) since $\usilent$ was observed starting from $c'''_1$ and from $c'''_2$ and since both configurations are in unprotected mode, protected memory is unchanged.
            Thus, $\M'_1 \pmem \M'''_1 \pmem \M'''_2 \pmem \M'_2$.
    \end{itemize}

    Since the instruction generating $\ta = \jmpin{\R}$ was executed in unprotected mode, we have that $\M''_1 \pmem \M''_2$.
    Also $\R''_1 = \R \preg{PM} \R = \R''_2$, $\moderelUM{\vopcp''_1}$, $\moderelUM{\vopcp''_2}$ and $\B''_1 \bowtie \B''_2$.
\end{proof}

\begin{property}\label{prop:peq-handle-reti-peq}
    If $c_1 \peq c_2$,
    $\atrsemstar{\D}{c_1}{\handle{k_1}}{c'_1 \atrarrowstar{\soverbrace{\usilent\ \cdots\ \usilent}^{\ell_1}}{c''_1 \atrarrow{\reti{k'_1}} c'''_1}}$,\\
    $\atrsemstar{\D'}{c_2}{\handle{k_2}}{c'_2 \atrarrowstar{\soverbrace{\usilent\ \cdots\ \usilent}^{\ell_2}}{c''_2 \atrarrow{\reti{k'_2}} c'''_2}}$,
    then $c'''_1 \peq c'''_2$.
\end{property}
\begin{proof}
    Since upon observation of $\handle{k_x}$ the protected memory cannot be modified, we know that $\M'_1 \pmem \M'_2$.

    We show by Noetherian induction over $(\ell_1, \ell_2)$ that $\M''_1 \pmem \M''_2$.
    For that, we use well-founded relation $(\ell_1, \ell_2) \prec (\ell'_1, \ell'_2)$ \emph{iff} $\ell_1 < \ell'_1 \land \ell_2 < \ell'_2$.

    \begin{itemize}
        \item \emph{Case $(0, 0)$.}
            Trivial.
        \item \emph{Case $(0, \ell_2)$, with $\ell_2 > 0$ (and symmetrically $(\ell_1, 0)$, with $\ell_1 > 0$).}
            We have to show that
            \[
                \atrsemstar{\D}{c'_1}{\emptystr}{c''_1}
                \land
                \atrsemstar{\D'}{c'_2}{\soverbrace{\usilent\ \cdots\ \usilent}^{\ell_2}}{c''_2}
                \Rightarrow
                \M''_1 \pmem \M''_2
            \]
            Since from $c'_1$ there is no step, $c''_1 = c'_1$.
            Moreover a sequence of $\usilent$ was observed starting from $c'_2$, and since both configurations are in unprotected mode and no violation occurred (see Table~\ref{tab:mac-rep}) the protected memory is unchanged.
            Thus, by transitivity of $\pmem$, we have $\M''_1 = \M'_1 \pmem \M'_2 \pmem \M''_2$.
        \item \emph{Case $(\ell_1, \ell_2) = (\ell'_1 + 1, \ell'_2 + 1)$.}
                If
                \[
                    \atrsemstar{\D}{c'_1}{\soverbrace{\usilent\ \cdots\ \usilent}^{\ell'_1}}{c^{iv}_1}
                    \land
                    \atrsemstar{\D'}{c'_2}{\soverbrace{\usilent\ \cdots\ \usilent}^{\ell'_2}}{c^{iv}_2}\Rightarrow
                    \M^{iv}_1 \pmem \M^{iv}_2 \text{ (IHP)}
                \]
                then
                \[
                    \atrsemstar{\D}{c'_1}{\soverbrace{\usilent\ \cdots\ \usilent}^{\ell'_1}}{c^{iv}_1 \atrarrow{\usilent} c''_1}
                    \land
                    \atrsemstar{\D'}{c'_2}{\soverbrace{\usilent\ \cdots\ \usilent}^{\ell'_2}}{c^{iv}_2 \atrarrow{\usilent} c''_2}\Rightarrow
                    \M''_1 \pmem \M''_2.
                \]

                By (IHP) we know that $\M^{iv}_1 \pmem \M^{iv}_2$.
                Indeed, since we observed $\usilent$ it means that $\moderelUM{\vopc''_1} \ \land\ \moderelUM{}{\vopc''_2}$.
                Moreover (see Figure~\ref{fig:finegrained-obs}) since $\usilent$ was observed starting from $c^{iv}_1$ and from $c^{iv}_2$ and since both configurations are in unprotected mode, no violation occurred and by Table~\ref{tab:mac-rep} protected memory is unchanged.
                Thus, by transitivity of $\pmem$, we have $\M''_1 \pmem \M^{iv}_1 \pmem \M^{iv}_2 \pmem \M''_2$.
    \end{itemize}

    Thus, we have that $\M'''_1 \pmem \M'''_2$, since $\ta = \reti{\cdot}$ does not modify protected memory.
    Also $\R'''_1 \preg{UM} \R'''_2$, $\B'''_1 \bowtie \B'''_2$, $\moderelUM{\vopcp_1}$ and $\moderelUM{\vopcp_2}$, by definition of $\ta = \reti{\cdot}$.
\end{proof}

\begin{property}\label{prop:peq-alpha-alphapeq}
    If $c_1 \peq c_2$,
    $\moderelPM{c_1}$,
    $\atrsem{\D}{c_1}{\ta_1}{c'_1}$,
    $\atrsem{\D'}{c_2}{\ta_2}{c'_2}$,
    $\ta_1, \ta_2 \neq \handle{\cdot}$
    then
    $\ta_1 = \ta_2$ and $c'_1 \peq c'_2$.
\end{property}
\begin{proof}
    By definition of fine-grained traces we know that the transition leading to the observation of $\ta_1$ happens upon the execution of an instruction that must also be executed starting from $c_2$ (by Property~\ref{prop:peq-samedecode}) and that $c'_1 \peq c'_2$ (by Property~\ref{prop:peq-main-peq}).
    Also, since $\moderelPM{c_1}$, we know that $\ta_1 \in \{ \psilent{k_1}, \jmpout{k_1}{\R_1} \}$.
    Thus, in both cases and since by hypothesis $\ta_2 \neq \handle{\cdot}$, it must be that $\ta_2 = \ta_1$.
\end{proof}

\begin{property}\label{prop:peq-swap-ctx}
    If
    $c_1 \peq c_2$,
    $\atrsemstar{\D}{c_1}{\psilent{k^{(0)}_1} \,\cdots\, \psilent{k^{(n_1-1)}_1} \cdot \ta_1}{c'_1}$,
    $\atrsemstar{\D'}{c_2}{\psilent{k^{(0)}_2} \,\cdots\, \psilent{k^{(n_2-1)}_2} \cdot \ta_2}{c'_2}$,
    and $\ta_1, \ta_2 \neq \handle{\cdot}$
    then
    $\psilent{k^{(0)}_1} \cdots \psilent{k^{(n_1-1)}_1} \cdot \ta_1 = \psilent{k^{(0)}_2} \cdots \psilent{k^{(n_2-1)}_2} \cdot \ta_2$ and $c'_1 \peq c'_2$.
\end{property}
\begin{proof}
    Corollary of Property~\ref{prop:peq-alpha-alphapeq}.
\end{proof}

\noindent
$P$-equivalence is preserved by complete interrupt segments (recall Definition~\ref{def:completely-handled-ints}).
Indeed, from now onwards denote
\begin{align*}
\atr_x \,\in\,
    &\{\emptystr\}\, \cup \\
    &\{ \ta^{(0)}_x \cdots \ta^{(n_x - 1)}_x \mid n_x \geq 1 \, \land\, \ta^{(n_x - 1)}_x = \reti{k^{(n_x - 1)}_x} \, \land\, \\
    &\qquad\qquad\qquad\qquad\forall i.\, 0 \leq i \leq n_x - 1.\, \ta^{(i)}_x \notin \{ \cnv, \jmpin{\R^{(i)}_x}, \jmpout{k^{(i)}_x}{\R^{(i)}_x} \} \}.
\end{align*}
%
\begin{property}\label{prop:peq-pres-handlereti}
    Let $\D$ and $\D'$ be two devices.

    \noindent
    If
    $c^{(0)}_1 \peq c^{(0)}_2$,
    $\atrsem{\D}{c_1}{\jmpin{\R}}
    {
        c^{(0)}_1
        \atrarrowstar{\atr_1}
        c^{(n_1)}_1
    }$ and
    $\atrsem
    {\D'}
    {c_2}
    {\jmpin{\R}}
    {
        c^{(0)}_2
        \atrarrowstar{\atr_2}
        c^{(n_2)}_2
    }$
    then
    $c^{(n_1)}_1 \peq c^{(n_2)}_2$.
\end{property}
\begin{proof}

    We first show by induction on $|\mb{I}_{\atr_1}|$ (see Definition~\ref{def:completely-handled-ints}) that
    \begin{align*}
        \atrsemstar{\D}{c^{(0)}_1}{\atr_1}{c^{(n_1)}_1}
        \ \land\
        \atrsemstar{\D'}{c^{(0)}_2}{\atr_2}{c^{(n_2)}_1}
        \Rightarrow
        c^{(n_1)}_1 \peq c^{(n_2)}_2
    \end{align*}
    assuming wlog that ${|\mb{I}_{\atr_2}|} \leq {|\mb{I}_{\atr_1}|}$.

    \begin{itemize}
        \item \emph{Case ${|\mb{I}_{\atr_1}|} = 0$.}
            Trivial.
        \item \emph{Case ${|\mb{I}_{\atr_1}|} = {|\mb{I}_{\atr'_1}|} + 1$.}
            If
            \begin{align*}
                \atrsemstar{\D}{c^{(0)}_1}{\atr'_1}{c^{(n'_1)}_1} \ \land\ \atrsemstar{\D'}{c^{(0)}_2}{\atr'_2}{c^{(n'_2)}_2}
                \Rightarrow
                c^{(n'_1)}_1 \peq c^{(n'_2)}_2 \text{ (IHP)}
            \end{align*}
            then
            \begin{align*}
                \atrsemstar{\D}{c^{(0)}_1}{\atr_1}{c^{(n_1)}_1} \ \land\ \atrsemstar{\D'}{c^{(0)}_2}{\atr_2}{c^{(n_2)}_2}
                \Rightarrow
                c^{(n_1)}_1 \peq c^{(n_2)}_2
            \end{align*}
            Now let $(i_1, j_1)$ be the new interrupt segment of $\atr_1$ that we split it as follows:
            \begin{align*}
                \atr_1 =\ \atr'_1 \cdot \psilent{k^{(n'_1)}_1} \cdots \psilent{k^{(i_1-1)}_1} \cdot \handle{k^{(i_1)}_1} \cdots \reti{k^{(j_1)}_1}
            \end{align*}
            The following two exhaustive cases may arise.
                \begin{enumerate}
                    \item \emph{Case ${|\mb{I}_{\atr_1}|} = {|\mb{I}_{\atr_2}|}$.}
                        For some $(i_2, j_2)$ we then have:
                        \begin{align*}
                            \atr_2 =\ \atr'_2 \cdot
                            \psilent{k^{(n'_2)}_2} \cdots \psilent{k^{(i_2-1)}_2} \cdot
                            \handle{k^{(i_2)}_2} \cdots \reti{k^{(j_2)}_2}
                        \end{align*}
                        By Properties~\ref{prop:peq-swap-ctx} and~\ref{prop:peq-handle-reti-peq} we know that $c^{(n_1)}_1 \peq c^{(n_2)}_2$, being reached through $\ta^{(j_1)}_1$ and $\ta^{(j_2)}_2$.
                    \item \emph{Case ${|\mb{I}_{\atr_2}|} < {|\mb{I}_{\atr_1}|}$.}
                        In this case we have
                        \begin{align*}
                            \atr_2 =\ \atr'_2 \cdot
                            \psilent{k^{(n'_2)}_2} \cdots \psilent{k^{(n_2 - 2)}_2} \cdot \psilent{k^{(n_2 - 1)}_2}
                        \end{align*}
                        with $c^{\ell}_1 \peq c^{\ell}_2$ for $n'_2 \leq \ell \leq n_2 - 2 = i_1-1$,
                        where the last equality holds because the module is executing from configurations that are $P$-equivalent.
                        As soon as the interrupt arrives, the same instruction is executed (Property~\ref{prop:peq-samedecode}) that causes the same changes in the registers, the old program counter and the protected memory.
                        In turn the first two are stored in the backup before handling the interrupt.
                        They are then restored by the $\RETI$, observed as $\ta^{(j_1)}_1$, while the protected memory is left untouched.
                        Consequently, we have that $c^{(n_1)}_1 \peq c^{(n_2)}_2$, that are the configurations reached through $\ta^{(j_1)}_1$ and $\psilent{k^{(n_2)-1}_2}$.
                \end{enumerate}
    \end{itemize}
\end{proof}

Finally, we can show that $P$-equivalence is preserved by coarse-grained traces:
\begin{property}\label{prop:init-jmpin}
    If
    $\btrsem{\D}{\initconf{C}{\M_M}}{\jmpin{\R}}{c_1}$ and
    $\btrsem{\D'}{\initconf{C'}{\M_M}}{\jmpin{\R}}{c_2}$
    then $c_1 \peq c_2$.
\end{property}
\begin{proof}
    By definition of coarse-grained traces, we have that in both premises $\jmpin{\R}$ is preceded by a sequence of $\usilent$ actions (possibly in different numbers).
    Since neither $\usilent$ actions nor $\jmpin{\R}$ ever change the protected memory (by definition of memory access control) and since the $\jmpin{\R}$ sets the registers to the values in $\R$, it follows that $c_1 \peq c_2$.
\end{proof}

The following definition gives an equality up to timings among coarse-grained traces:
\begin{definition}\label{def:btr-approx-btrp}
    Let $\btr = \tb_0 \ldots \tb_n$ and $\btr' = \tb'_0 \ldots \tb'_{n'}$ be two coarse-grained traces.
    We say that $\btr$ is \emph{equal up to timings} to $\btr'$ (written $\btr \tapprox \btr'$) \emph{iff}
    \[
        n = n' \land (\forall i \in \{ 0, \ldots, n \} .\, \tb_i = \tb'_i \lor (\tb_i = \jmpout{\dt}{\R} \land \tb'_i = \jmpout{\dt'}{\R})).
    \]
\end{definition}

and the following property shows that if traces that are equal up to timings preserve $P$-equivalence:
\begin{property}\label{prop:peq-pres-btr}
    If
    $c_1 \peq c_2$,
    $\btrsemstar{\D}{c_1}{\btr}{c'_1}$,
    $\btrsemstar{\D'}{c_2}{\btr'}{c'_2}$ and
    $\btr \tapprox \btr'$
    then $c'_1 \peq c'_2$.
\end{property}
\begin{proof}
    The thesis easily follows from Property~\ref{prop:peq-epsilon-jmpin-peq} and Property~\ref{prop:peq-pres-handlereti}.
\end{proof}

\subsubsection{Properties of \texorpdfstring{$U$}{U}-equivalence}
Also for U-equivalent configurations it holds that when one takes a step, also the other does.
\begin{property}\label{prop:ueq-samedecode}
    If $c_1 \ueq c_2$, $\moderelUM{c_1}$ then $\decode{\M_1}{\R_1[\rpc]} = \decode{\M_2}{\R_2[\rpc]}$.
\end{property}
\begin{proof}
    Since $c_1 \ueq c_2$ and $\moderelUM{c_1}$, it also holds that $\moderelUM{c_2}$.
    Also, the instruction $\decode{\M_1}{\R_1[\rpc]}$ is decoded in both $\M_1$ and $\M_2$ at the same unprotected address, hence $\decode{\M_1}{\R_1[\rpc]} = \decode{\M_2}{\R_2[\rpc]}$.
\end{proof}

Next we prove that $\ueq$ is preserved by unprotected-mode steps of the \SL operational semantics:
\begin{property}\label{prop:ueq-main-ueq}
    If $c_1 \ueq c_2$, $\moderelUM{c_1}$ and $\slmain{c_1}{c'_1}$, then $\slmain{c_2}{c'_2} \ \land\ c'_1 \ueq c'_2$.
\end{property}
\begin{proof}
    Since $c_1 \ueq c_2$, $\moderelUM{c_1}$ and $\slmain{c_1}{c'_1}$, by Property~\ref{prop:ueq-samedecode}, $i = \decode{\M_1}{\R_1[\rpc]} = \decode{\M_2}{\R_2[\rpc]}$.

    To show that $c'_1 \ueq c'_2$, we consider the following exhaustive cases:
    \begin{itemize}
        \item \emph{Case $i = \bot$.}
            Since $c_1 \ueq c_2$ we get $\moderelUM{c_2}$ and by definition of $\slmaind{\cdot}{\cdot}{\cdot}$ we get $c'_1 = \exconf{c_1}$ and $c'_2 = \exconf{c_2}$.
            However, by definition of $\exconf{\cdot}$, we have that $\M'_1 \umem \M'_2$, $\moderelUM{c'_1}$, $\moderelUM{c'_2}$, $\delta'_1 = \delta_1 = \delta_2 = \delta'_2$, $t'_1 = t_1 = t_2 = t'_2$, $t'_{a_1} = t_{a_1} = t_{a_2} = t'_{a_2}$, $\R'_1 \ureg{m} \R'_2$, and $\BBot = \B'_1 \bowtie \B'_2 = \BBot$, i.e., $c'_1 \ueq c'_2$.
        \item \emph{Case $i = \HLT$.}
            Trivial, since $c'_1 = \haltconf = c'_2$.
        \item \emph{Case $i \neq \bot$.}
            We have the following exhaustive sub-cases, depending on $c_1'$:
            \begin{itemize}
                \item \emph{Case $c'_1 = \exconf{c_1}$. }
                    In this case a violation occurred, i.e., $\nmacrel{i}{\vopc_1}{\R_1}{\B_1}$.
                    However, the same violation also occurs for $c_2$, since the only parts that may keep $c_1$ apart from $c_2$ are $\vopc$ and $\B$, and thus $c'_1 \ueq c'_2$ because:
                    \begin{itemize}
                        \item $\vopc_2 \neq \vopc_1$, cannot cause a failure since unprotected code is executable from anywhere,
                        \item $\B_1 = \langle \R_1, \vopc_1, t_{\mi{pad}_1} \rangle \neq \langle \R_2, \vopc_2, t_{\mi{pad}_2} \rangle = \B_2$, cannot cause a failure since the additional conditions on the configuration imposed by the memory access control only concern values that are the same in both configurations.
                    \end{itemize}

                \item \emph{Case $c'_1 \neq \exconf{c_1}$ and $i = \RETI$.}
                    If $\B_1 = \BBot$, then $\B_1 = \B_2 = \B'_1 = \B'_2 = \BBot$, hence rule~\slrulename{(CPU-Reti)} of Figure~\ref{fig:sl-maints2} applies and we get $c'_1 \ueq c'_2$ since $\R'_1 = \R'_2$ and $\dwrap{\cdot}{\cdot}{\cdot}$ is a deterministic relation (Property~\ref{prop:dwrap-deterministic}).
                    If $\B_1 \neq \BBot$ it must also be that $\B_2 \neq \BBot$ by $U$-equivalence, so either rule~\slrulename{(CPU-Reti-Chain)} or rule~\slrulename{(CPU-Reti-PrePad)} applies.
                    In the first case we get $c'_1 \ueq c'_2$ because $c_1 \ueq c_2$ and by determinism of $\dwrap{\cdot}{\cdot}{\cdot}$ and $\slint{\D}{\cdot}{\cdot}$.
                    In the second case we get $c'_1 \ueq c'_2$ since $\langle \bot, \bot, t'_{\mi{pad}_1} \rangle = \B'_1 \bowtie \B'_2 = \langle \bot, \bot, t'_{\mi{pad}_2} \rangle$ and $\R'_1 \ureg{PM} \R'_2$ holds since we restored the register files from backups in which the interrupts were enabled (otherwise the CPU would not have handled the interrupt it is returning from).
                \item \emph{Case $c'_1 \neq \exconf{c_1}$ and $i \not\in \{\bot, \HLT, \RETI \}$.}
                    All the other rules depend on both
                    $(i)$ parts of the configurations that are equal due to $c_1 \ueq c_2$, and on
                    $(ii)$ $\dwrap{5}{\cdot}{\cdot}$ and $\slint{\D}{\cdot}{\cdot}$ which are deterministic and have the same inputs (since $c_1 \ueq c_2$).
                    Hence, $c'_1 \ueq c'_2$ as requested.
            \end{itemize}
        \end{itemize}
\end{proof}

The above property carries on fine-grained traces, provided that the computation is carried on in unprotected mode:
\begin{property}\label{prop:ueq-finesingle-ueq}
    If $c_1 \ueq c_2$,
    $\moderelUM{c_1}$,
    $\atrsem{\D}{c_1}{\ta}{c'_1}$
    then
    $\atrsem{\D}{c_2}{\ta}{c'_2}$
    and
    $c'_1 \ueq c'_2$.
\end{property}
\begin{proof}
    Properties~\ref{prop:ueq-samedecode} and~\ref{prop:ueq-main-ueq} guarantee that $c'_1 \ueq c'_2$ and $i = \decode{\M_1}{\R_1[\rpc]} = \decode{\M_2}{\R_2[\rpc]}$.
    Thus, since the same $i$ is executed under $U$-equivalent configurations and since $c'_1 \ueq c'_2$, we have that $\atrsem{\D}{c_2}{\ta}{c'_2}$.
\end{proof}
\begin{property}\label{prop:ueq-fine-ueq}
    If $c_1 \ueq c_2$,
    $\moderelUM{c_1}$,
    $\atrsemstar{\D}{c_1}{\usilent\cdots\usilent\cdot\ta}{c'_1}$ and
    $\ta \in \{ \usilent, \cnv, \jmpin{\R}, \reti{k} \}$
    then
    $\atrsemstar{\D}{c_2}{\usilent\cdots\usilent\cdot\ta}{c'_2}$
    and
    $c'_1 \ueq c'_2$.
\end{property}
\begin{proof}
    The proof goes by induction on the length $n$ of $\usilent\cdots\usilent$.
    \begin{itemize}
        \item \emph{Case $n = 0$.} Property~\ref{prop:ueq-finesingle-ueq} applies.
        \item \emph{Case $n' = n + 1$.}
            By induction hypothesis for some $c'''_1$, $c'''_2$, $c''_1$ and $c''_2$ we have $\atrsem{\D}{c_1}{\soverbrace{\usilent\cdots\usilent}^{n'}}{ c'''_1 \atrarrow{\ta} c''_1}$, $\atrsem{\D}{c_2}{\soverbrace{\usilent\cdots\usilent}^{n'}}{c'''_2 \atrarrow{\ta} c''_2}$ and $c''_1 \ueq c''_2$.
            Thus, if $\atrsem{\D}{c'''_1}{\usilent}{c^{iv}_1}$ (i.e., we observe a further $\usilent$ starting from $c_1$), by Property~\ref{prop:ueq-finesingle-ueq} we get $\atrsem{\D}{c'''_2}{\usilent}{c^{iv}_2}$ and $c^{iv}_1 \ueq c^{iv}_2$.
            Finally, by Property~\ref{prop:ueq-finesingle-ueq} applies on $c^{iv}_1$ and $c^{iv}_2$ we get the thesis.
    \end{itemize}
\end{proof}

Now we move our attention to $\handle{\cdot}$.
\begin{property}\label{prop:um-pres-tauhandle}
    If $c_1^{(0)} \ueq c_2^{(0)}$,
    $\atrsemstar{\D}{c_1^{(0)}}{\psilent{k_1^{(0)}} \ \cdots\  \psilent{k_1^{(n_1 - 1)}} \cdot \handle{k_1^{(n_1)}}}{c_1^{(n_1+1)}}$ and\\
    $\atrsemstar{\D}{c_2^{(0)}}{\psilent{k_2^{(0)}} \ \cdots\  \psilent{k_2^{(n_2 - 1)}} \cdot \handle{k_2^{(n_2)}}}{c_2^{(n_2+1)}}$
    then
    $c_1^{(n_1+1)} \ueq c_2^{(n_2+1)}$.
\end{property}
\begin{proof}
    \begin{itemize}
        \item By definition of fine-grained semantics, $\handle{k_x^{(n_x)}}$ only happens when an interrupt is handled with $c_x^{(n_x)}$ in protected mode.
        \item By definition of $\slint{\D}{\cdot}{\cdot}$, $\R^{(n_1+1)}_1 = \R^{(n_2+1)}_2 = \R_0 [\rpc \mapsto \mi{isr}]$.
        \item Since unprotected memory cannot be changed by protected mode actions without causing a violation (that would cause the observation of a $\jmpout{\cdot}{\cdot}$) and is not changed upon $\RETI$ when it happens in a configuration with backup different from $\BBot$ (cf. rules~\slrulename{(CPU-Reti-*)}), $\M^{(n_1+1)}_1 \umem \M^{(n_2+1)}_2$.
        \item Since we observe $\handle{k_x^{(n_x)}}$ it must be that $\mtt{GIE} = \mtt{1}$ and it had to be such also in $c^{(0)}_x$ (because by definition the operations on registers cannot modified this flag in protected mode).
        Hence, $t^{i}_{a_x} = \bot$ for $0 \leq i \leq n_x$.
        Let $t^\mi{int}_{a_1}$ and $t^\mi{int}_{a_2}$ be the arrival times of the interrupt that originated the observations $\handle{k_1^{(n_1)}}$ and $\handle{k_2^{(n_2)}}$, resp.
        By definition of $\dwrap{\cdot}{\cdot}{\cdot}$, $t^\mi{int}_{a_1}$ and $t^\mi{int}_{a_2}$ are the first absolute times after $t^{(n_1)}_1$ and $t^{(n_2)}_2$ in which an interrupt was raised and, since $\D$ is deterministic and $t^{(i)}_{a_x} = \bot$ for $0 \leq i \leq n_x$, it must be that $t^\mi{int}_{a_1} = t^\mi{int}_{a_2} = t^\mi{int}$ (recall that $c_1^{(0)} \ueq c_2^{(0)}$ and that $\IN{}\hspace{-0.5em}$ or $\OUT{}\hspace{-0.5em}$ instructions are forbidden in protected mode).

        Assume now that the instruction during which the interrupt occurred ended at time $t^f_x$.
        Then we can write $t^{(n_x+1)}$ as:
                \begin{align*}
                    t^{(n_x+1)} = t^{(n_x)} + k_x^{(n_x)} &=
                    t^{(n_x)} + \underbrace{t^\mi{int} - t^{(n_x)} + t^f_x - t^\mi{int}}_{\text{Duration of the instruction}} + \underbrace{\MT - t^f_x + t^\mi{int}}_{\text{Mitigation from~\slrulename{(INT-PM-P)}}} + 6\\
                    &= \cancel{t^{(n_x)}} + t^\mi{int} - \cancel{t^{(n_x)}} + \cancel{t^f_x} - \cancel{t^\mi{int}} + \MT - \cancel{t^f_x} + \cancel{t^\mi{int}} + 6\\
                    &= t^\mi{int} + \MT + 6
                \end{align*}
        and therefore $t^{(n_1+1)} = t^{(n_2+1)}$.

        \item Since $t^{(n_1+1)} = t^{(n_2+1)}$, $c_1^{(0)} \ueq c_2^{(0)}$ and no interaction with $\D$ via $\IN{}\hspace{-0.5em}$ or $\OUT{}\hspace{-0.5em}$ can occur in protected mode, the deterministic device $\D$ performed the same number of steps in both computations, and then $t^{(n_1 + 1)}_{a_1} = t^{(n_2 + 1)}_{a_2}$ and $\delta^{(n_1 + 1)}_{1} = \delta^{(n_2 + 1)}_{2}$.
    \end{itemize}

    Hence, $c_1^{(n_1+1)} \ueq c_2^{(n_2+1)}$ as requested.
\end{proof}

The following properties show that the combination of $U$-equivalence and trace equivalence induces some useful properties of modules and sequences of complete interrupt segments.
Before doing that we define the $(\overline{a}, n)$-interrupt-limited version of a context $C$ as the context that behaves as $C$ but such that
$(i)$ the transition relation of its device results from unrolling at most $n$ steps of its transition relation and
$(ii)$ its device never raises interrupts \emph{after} observing the sequence of actions $\overline{a}$:




\begin{definition}\label{def:limintctx}
    Let $\D = \langle \Delta, \deltainit, \xleadsto{a}_D \rangle$ be an I/O device.
    Let $\overline{a}$ be a string over the signature $A$ of I/O devices and denote $\ell$ as the function that associates to each string over $A$ a unique natural number (e.g., its position in a suitable lexicographic order).
    Given a context $C = \langle \M_C, \D \rangle$, we define its corresponding \emph{$(\overline{a}, n)$-interrupt-limited context} as $\nIl{C}{\overline{a}, n} = \langle \M_C, \nIl{\D}{\overline{a}, n} \rangle$ where $\nIl{\D}{\overline{a}, n} = \langle \mi{img}(\nIl{\xleadsto{a}_D}{\overline{a}, n}) \cup \mi{dom}(\nIl{\xleadsto{a}_D}{\overline{a}, n}), 0, \nIl{\xleadsto{a}_D}{\overline{a}, n} \rangle$ and
    \begin{align*}
        \nIl{\xleadsto{a}_D}{\overline{a}, n} \triangleq\
            &\{ (p, a, p') \mid \forall \overline{a}'.\, p = \ell(\overline{a}') \land p' = \ell(\overline{a}'\cdot a) \land \delta_\mi{init} \xleadstostar{\overline{a}'}_D \delta \xleadsto{a}_D \delta' \land |\overline{a}'\cdot a| \leq n \}\ \setminus \\
            &\{ (p, \mi{int?}, p') \mid \forall \overline{a}'.\, p = \ell(\overline{a}\cdot\overline{a}') \land p' = \ell(\overline{a}\cdot\overline{a}'\cdot\mi{int?}) \}\ \cup\\
            &\{ (p, \epsilon, p') \mid  \forall \overline{a}'.\, p = \ell(\overline{a}\cdot\overline{a}') \land p' = \ell(\overline{a}\cdot\overline{a}'\cdot\mi{int?}) \land \delta_\mi{init} \xleadstostar{\overline{a}\cdot\overline{a}'}_D \delta \xleadsto{\mi{int?}}_D \delta' \land |\overline{a}\cdot\overline{a}'\cdot\mi{int?}| \leq n \}.
    \end{align*}
\end{definition}
\noindent
(Note that any $(\overline{a}, n)$-interrupt-limited context is actually a device, due to the constraint on its transition function).

Now, let
\begin{align*}
\atr_x \,\in\,
    &\{\emptystr\}\, \cup \\
    &\{ \ta^{(0)}_x \cdots \ta^{(n_x - 1)}_x \mid n_x \geq 1 \, \land\, \ta^{(n_x - 1)}_x = \reti{k^{(n_x - 1)}_x} \, \land\, \\
    &\qquad\qquad\qquad\qquad\forall i.\, 0 \leq i \leq n_x - 1.\, \ta^{(i)}_x \notin \{ \cnv, \jmpin{\R^{(i)}_x}, \jmpout{k^{(i)}_x}{\R^{(i)}_x} \} \}.
\end{align*}
\begin{property}\label{prop:ttreq-ilen}
    If
    \begin{itemize}
     \item $\M_M \ttreq \M_{M'}$
     \item $\btrsemstar{\D}{\initconf{C}{M_M}}{\btr\cdot\jmpin{\R}}{c^{(0)}_1}$
     \item $\btrsemstar{\D}{\initconf{C}{M_{M'}}}{\btr\cdot\jmpin{\R}}{c^{(0)}_2}$
     \item $c^{(0)}_1 \ueq c^{(0)}_2$
     \item for some $m_1 \geq 0$, $\atrsemstar{\D}{c^{(0)}_1}{\atr_1 \cdot \psilent{k^{(n_1)}_1} \cdots \psilent{k^{(n_1 + m_1 - 1)}_1} \cdot \jmpout{k^{(n_1 + m_1)}_1}{\R'}}{c^{(n_1+m_1+1)}_1}$
     \item for some $m_2 \geq 0$, $\atrsemstar{\D}{c^{(0)}_2}{\atr_2 \cdot \psilent{k^{(n_2)}_2} \cdots \psilent{k^{(n_2 + m_2 - 1)}_2} \cdot \jmpout{k^{(n_2 + m_2)}_2}{\R'}}{c^{(n_2+m_2+1)}_2}$
    \end{itemize}
     then
     $\sum_{i=0}^{n_1+m_1} \ilen{c^{(i)}_1} = \sum_{i=0}^{n_2+m_2} \ilen{c^{(i)}_2}$.
 \end{property}
 \begin{proof}
     We show this property by contraposition.
     Indeed, we show that if $\sum_{i=0}^{n_1+m_1} \ilen{c^{(i)}_1} \neq \sum_{i=0}^{n_2+m_2} \ilen{c^{(i)}_2}$ then $\M_M \nttreq \M_{M'}$.
     For that it suffices to show that
     \[
        \exists C'. \btrsemstar{\D'}{\initconf{C'}{\M_M}}{\btr\cdot\jmpin{\R}}{c^{(0)}_3 \btrarrow{\jmpout{\dt_3}{\R^{(n_3+m_3)}_3}} c^{(n_3+m_3+1)}_3}
     \]
     (i.e., $\atrsemstar{\D}{c^{(0)}_3}{\atr_3 \cdot \psilent{k^{(n_3)}_3} \cdots \psilent{k^{(n_3 + m_3 - 1)}_3} \cdot \jmpout{k^{(n_3 + m_3)}_3}{\R^{(n_3+m_3)}_3}}{c^{(n_3+m_3+1)}_3}$)

     \noindent
     such that
     \[
        \forall C''.\, \btrsemstar{\D''}{\initconf{C''}{\M_{M'}}}{\btr\cdot\jmpin{\R}}{c^{(0)}_4 \btrarrow{\jmpout{\dt_4}{\R^{(n_4+m_4+1)}_4}} c^{(n_4+m_4+1)}_4} \quad \text{ with } \dt_3 \neq \dt_4
     \]
     (i.e., $\atrsemstar{\D}{c^{(0)}_4}{\atr_4 \cdot \psilent{k^{(n_4)}_4} \cdots \psilent{k^{(n_4 + m_4 - 1)}_4} \cdot \jmpout{k^{(n_4 + m_4)}_4}{\R^{(n_4+m_4)}_4}}{c^{(n_4+m_4+1)}_4}$).

     Assume wlog that $\sum_{i=0}^{n_1+m_1} \ilen{c^{(i)}_1} < \sum_{i=0}^{n_2+m_2} \ilen{c^{(i)}_2}$.
     Noting that the first observable of $\btr\cdot\jmpin{\R}$ must be a $\jmpin{\cdot}$, by Properties~\ref{prop:init-jmpin} and~\ref{prop:peq-pres-btr}, we have that $c^{(0)}_1 \peq c^{(0)}_3$ and, similarly, $c^{(0)}_2 \peq c^{(0)}_4$.
     Thus, as a consequence of Properties~\ref{prop:peq-samedecode},~\ref{prop:peq-pres-handlereti} and~\ref{prop:peq-swap-ctx}, $\sum_{i=0}^{n_1+m_1} \ilen{c^{(i)}_1} = \sum_{i=0}^{n_3+m_3} \ilen{c^{(i)}_3}$ and $\sum_{i=0}^{n_2+m_2} \ilen{c^{(i)}_2} = \sum_{i=0}^{n_4+m_4} \ilen{c^{(i)}_4}$.

     Let $n \in \mb{N}$ be greater than the number of steps over the relation $\xleadsto{\cdot}_D$ in the computation $\slmainstar{\D}{\initconf{C}{\M_M}}{c^{(n_1+m_1+1)}_1}$ and let $\overline{a}$ be the sequence of actions over $\xleadsto{\cdot}_D$ in the computation $\slmainstar{\D}{\initconf{C}{\M_M}}{c^{(0)}_1}$.
     Choosing $C' = \nIl{C}{\overline{a}, n}$ we get $\dt_3 = \sum_{i=0}^{n_1+m_1} \ilen{c^{(i)}_1} = \sum_{i=0}^{n_3+m_3} \ilen{c^{(i)}_3}$.
     Any other context $C''$ that allows to observe the same $\btr\cdot\jmpin{\R}$ from $\initconf{C''}{\M_{M'}}$ raises $0$ or more interrupts ``after'' $c^{0}_4$, hence taking additional $S \geq 0$ cycles on top of those required for the instructions to be executed.
     Thus $\M_M \nttreq \M_{M'}$, since $\sum_{i=0}^{n_1+m_1} \ilen{c^{(i)}_1} < \sum_{i=0}^{n_2+m_2} \ilen{c^{(i)}_2}$ and $\sum_{i=0}^{n_1+m_1} \ilen{c^{(i)}_1} = \dt_3 < \dt_4 = \sum_{i=0}^{n_2+m_2} \ilen{c^{(i)}_2} + S$.
 \end{proof}

 \begin{property}\label{prop:ueq-pres-handlereti}
   If
   \begin{itemize}
    \item $\btrsemstar{\D}{\initconf{C}{M_M}}{\btr\cdot\jmpin{\R}}{c^{(0)}_1}$
    \item $\btrsemstar{\D}{\initconf{C}{M_{M'}}}{\btr'\cdot\jmpin{\R}}{c^{(0)}_2}$
    \item $c^{(0)}_1 \ueq c^{(0)}_2$
    \item $\atrsemstar{\D}{c^{(0)}_1}{\atr_1 \cdot \psilent{k^{(n_1)}_1} \cdots \psilent{k^{(n_1 + m_1 - 1)}_1} \cdot \ta_1}{c^{(n_1+m_1+1)}_1}$ for some $m_1 \geq 0$ and $\ta_1 \in \{ \jmpout{k^{(n_1 + m_1)}_1}{\R'},$ $ \handle{k^{(n_1 + m_1)}_1} \}$
    \item $\atrsemstar{\D}{c^{(0)}_2}{\atr_2 \cdot \psilent{k^{(n_2)}_2} \cdots \psilent{k^{(n_2 + m_2 - 1)}_2} \cdot \ta_2}{c^{(n_2+m_2+1)}_2}$ for some $m_2 \geq 0$ and $\ta_2 \in \{ \jmpout{k^{(n_2 + m_2)}_2}{\R'},$ $ \handle{k^{(n_2 + m_2)}_2} \}$
   \end{itemize}
    then
   \begin{enumerate}
    \item ${|\mb{I}_{\atr_1}|} = {|\mb{I}_{\atr_2}|}$
    \item $c^{(n_1)}_1 \ueq c^{(n_2)}_2$.
   \end{enumerate}
\end{property}
\begin{proof}
    Assume wlog that $\sum_{i=0}^{n_1+m_1} \ilen{c^{(i)}_1} \leq \sum_{i=0}^{n_2+m_2} \ilen{c^{(i)}_2}$, and we prove by induction on ${|\mb{I}_{\atr_1}|}$ that
    \begin{align*}
        \atrsemstar{\D}{c^{(0)}_1}{\atr_1}{c^{(n_1)}_1}
        \ \land\
        \atrsemstar{\D}{c^{(0)}_2}{\atr_2}{c^{(n_2)}_1}
        \text{\ \ \ imply \ \ \ }
        c^{(n_1)}_1 \ueq c^{(n_2)}_2
        \ \land\
        {|\mb{I}_{\atr_1}|} = {|\mb{I}_{\atr_2}|}
    \end{align*}

    \begin{itemize}
        \item \emph{Case ${|\mb{I}_{\atr_1}|} = 0$.}
            Since no complete interrupt segment was observed it means that $\atr_1$ cannot end with a $\reti{\cdot}$, so it must be $\atr_1 = \emptystr$.
            Moreover, since $c^{(0)}_1 \ueq c^{(0)}_2$ and the value of the $\mtt{GIE}$ bit cannot be changed in protected mode, we know that:
            \begin{itemize}
                \item \emph{Case $\R^{(0)}_1 [\rsr.\mtt{GIE}] = \R^{(0)}_2 [\rsr.\mtt{GIE}] = \mtt{0}$.} Then no $\handle{\cdot}$ can be observed in $\atr_2$, hence it must be that $\atr_2 = \emptystr$ and the two thesis easily follow.
                \item \emph{Case $\R^{(0)}_1 [\rsr.\mtt{GIE}] = \R^{(0)}_2 [\rsr.\mtt{GIE}] = \mtt{1}$.} Then it means that no interrupt was raised by the device in the computation starting with $c^{(0)}_1$ and the same must happen in $c^{(0)}_2$ because of $U$-equivalence and $\sum_{i=0}^{n_1+m_1} \ilen{c^{(i)}_1} \leq \sum_{i=0}^{n_2+m_2} \ilen{c^{(i)}_2}$.
                Hence it must be that $\atr_2 = \emptystr$ and the two thesis easily follow.
            \end{itemize}

        \item \emph{Case ${|\mb{I}_{\atr_1}|} = {|\mb{I}_{\atr'_1}|} + 1$.}
            If
            \begin{align*}
                \atrsemstar{\D}{c^{(0)}_1}{\atr'_1}{c^{(n'_1)}_1}
                \ \land\
                \atrsemstar{\D}{c^{(0)}_2}{\atr'_2}{c^{(n'_2)}_2}
                \text{\ \ \ imply \ \ \ }
                c^{(n'_1)}_1 \ueq c^{(n'_2)}_2
                \ \land\
                {|\mb{I}_{\atr'_1}|} = {|\mb{I}_{\atr'_2}|}
                \text{ (IHP)}
            \end{align*}
            then
            \begin{align*}
                \atrsemstar{\D}{c^{(0)}_1}{\atr_1}{c^{(n_1)}_1}
                \ \land\
                \atrsemstar{\D}{c^{(0)}_2}{\atr_2}{c^{(n_2)}_2}
                \text{\ \ \ imply \ \ \ }
                c^{(n_1)}_1 \ueq c^{(n_2)}_2
                \ \land\
                 {|\mb{I}_{\atr_1}|} = {|\mb{I}_{\atr_2}|}
            \end{align*}

            Now let $(i_1, j_1)$ be the new interrupt segment of $\atr_1$, that we split as follows:
            \begin{align*}
                \atr_1 =\ \atr'_1 \cdot
                \psilent{k^{(n'_1)}_1} \cdots \psilent{k^{(i_1-1)}_1} \cdot
                \handle{k^{(i_1)}_1} \cdots \reti{k^{(j_1)}_1}.
            \end{align*}

            Since by (IHP) $c^{(n'_1)}_1 \ueq c^{(n'_2)}_2$ and $\D$ is deterministic and no successfully I/O ever happens in protected mode, the first new interrupt (i.e. the one leading to the observation of $\handle{k^{(i_1)}_1}$) is raised at the same cycle in both computations.
            Call $c^{(i_2)}_2$ the configuration at the beginning of the step of computation in which such interrupt was raised (the choice of indexes will be clear below).
            From this configuration only three cases for the fine-grained action might be observed:
            \begin{itemize}
                \item \emph{Case $\psilent{\cdot}$ and $\jmpout{\cdot}{\cdot}$.} Never happens, since $\B^{(i_2+1)}_2 \neq \bot$.
                \item \emph{Case $\handle{k^{(i_2)}_2}$.}
                    Property~\ref{prop:um-pres-tauhandle} ensures that $c^{(i_2+1)}_2 \ueq c^{(i_1+1)}_1$, and
                    Property~\ref{prop:ueq-fine-ueq} that at some index $j_2$ a $\reti{k^{(j_2)}_2}$ is observed in $\atr_2$, i.e., a new interrupt segment $(i_2, j_2)$ is observed.
                    Thus, ${|\mb{I}_{\atr_2}|} = {|\mb{I}_{\atr'_2}|} + 1 = {|\mb{I}_{\atr'_1}|} + 1 = {|\mb{I}_{\atr_1}|}$ (where the second equality holds by (IHP)).
                    Finally, by definition of $\atr_2$, we have that $n_1 = j_1+1$ and $n_2 = j_2+2$, hence $c^{(n_1)}_1 \ueq c^{(n_2)}_2$.
            \end{itemize}
    \end{itemize}
\end{proof}
The following property states that $U$-equivalent unprotected-mode configurations perform the same single coarse-grained action:
\begin{property}\label{prop:ueq-um-samebeta}
    If $c_1 \ueq c_2$, $\moderelUM{c_1}$ and $\btrsem{\D}{c_1}{\tb}{c'_1}$, then $\btrsem{\D}{c_2}{\tb}{c'_2}$ and $c'_1 \ueq c'_2$.
\end{property}
\begin{proof}
    Since $\moderelUM{c_1}$, the segment of fine-grained trace that originated $\tb$ (see Figure~\ref{fig:coarsegrained-obs}) is in the form:
    \[
        \atrsemstar{\D}{c_1}{\usilent \cdots \usilent \cdot \ta}{c'_1}
    \]
    with either $\ta = \cnv$ or $\ta = \jmpin{\R}$.

    \noindent
    Property~\ref{prop:ueq-fine-ueq} guarantees that:
    \[
        \atrsemstar{\D}{c_2}{\usilent \cdots \usilent \cdot \ta}{c'_2} \land c'_1 \ueq c'_2.
    \]
    Thus, $\btrsem{\D}{c_2}{\tb}{c'_2}$ and $c'_1 \ueq c'_2$.
\end{proof}
Finally, we can show that $U$-equivalence is preserved by coarse-grained traces:
\begin{property}\label{prop:ueq-pres-btr}
    If $c_1 \ueq c_2$, $\moderelUM{c_1}$, $\btrsemstar{\D}{c_1}{\btr}{c'_1}$, $\btrsemstar{\D}{c_2}{\btr}{c'_2}$, $\moderelUM{c'_1}$ and $\moderelUM{c'_2}$ then $c'_1 \ueq c'_2$.
\end{property}
\begin{proof}
    We show the property by induction on $n$, the length of $\btr$:
    \begin{itemize}
        \item \emph{Case $n=0$.}
            By definition of $\btrarrowstar{\emptystr}$ we know that it must be $c'_1 = c_1$ and $c'_2 = c'_2$ and the thesis easily follows.
        \item \emph{Case $n=n'+1$.}
            The only case in which a coarse-grained trace can be extended by just one action, while remaining in unprotected mode, is when the action is $\cnv$.
            In this case the hypothesis easily follows from the definition of $\cnv$ and $U$-equivalence.
        \item \emph{Case $n=n'+2$.}
            If
            \[
                \btrsemstar{\D}{c_1}{\btr}{c''_1} \ \land\
                \btrsemstar{\D}{c_2}{\btr}{c''_2} \ \land\
                \moderelUM{\R''_1[\rpc]} \ \land\
                \moderelUM{\R''_2[\rpc]}
                \text{\ \ imply \ }
                c''_1 \ueq c''_2
            \]
            then
            \[
                \btrsemstar{\D}{c_1}{\btr}{c''_1 \btrarrow{\tb\tb'} c'_1} \ \land\
                \btrsemstar{\D}{c_2}{\btr}{c''_2 \btrarrow{\tb\tb'} c'_2} \ \land\
                \moderelUM{\R'_1[\rpc]} \ \land\
                \moderelUM{\R'_2[\rpc]}
                \text{\ \ imply \ }
                c'_1 \ueq c'_2.
            \]

            By cases on $\tb\tb'$:
            \begin{itemize}
                \item  \emph{Case $\tb\tb' = \jmpin{\R}\, \cnv$.}
                    Directly follows from definition of $\cnv$ and $\ueq$.
                \item \emph{Case $\tb\tb' = \jmpin{\R}\, \jmpout{\dt}{\R'}$.}
                    By definition they are originated by
                    \begin{align*}
                        &\atrsemstar
                        {\D}
                        {c''_1}
                        {\usilent \cdots \usilent \cdot \jmpin{\R}}
                        {
                            c^{(0)}_1
                            \atrarrowstar{\ta^{(0)}_1\ \cdots\ \ta^{(n_1 - 1)}_1}
                            c^{(n_1)}_1
                            \atrarrow{\jmpout{k^{(n_1)}_1}{\R'}}
                            c'_1
                        }\\
                        &\atrsemstar
                        {\D}
                        {c''_2}
                        {\usilent \cdots \usilent \cdot \jmpin{\R}}
                        {
                            c^{(0)}_2
                            \atrarrowstar{\ta^{(0)}_2\ \cdots\ \ta^{(n_2 - 1)}_2}
                            c^{(n_2)}_2
                            \atrarrow{\jmpout{k^{(n_2)}_2}{\R'}}
                            c'_2
                        }.
                    \end{align*}

                    By \emph{(IHP)} and by Property~\ref{prop:ueq-fine-ueq} we can conclude that $c^{(0)}_1 \ueq c^{(0)}_2$.

                    Let $c^{(M_x)}_x$ be the configuration generated by the last $\reti{\cdot}$ in $\ta^{(0)}_x \ \cdots\ \ta^{(n_x - 1)}_x$.
                    By Property~\ref{prop:ueq-pres-handlereti} the number of completely handled interrupts is the same in the two traces and $c^{(M_1)}_1 \ueq c^{(M_2)}_2$.
                    Also:
                    \begin{itemize}
                        \item By definition of $\jmpout{k^{(n_1)}_1}{\R'}$ and $\jmpout{k^{(n_2)}_2}{\R'}$ we trivially get $\R'_1 = \R'_2 = \R'$.
                        \item Since unprotected memory cannot be changed in protected mode (see Table~\ref{tab:mac-rep}) and $c^{(M_1)}_1 \ueq c^{(M_2)}_2$, $\M'_1 \umem \M'_2$.
                        \item
                            Let $\atr_x = \ta^{(0)}_x\ \cdots\ \ta^{(n_x - 1)}_x\cdot\jmpout{k^{(n_x)}_x}{\R'}$.
                            By definition of $\tb = \jmpout{\dt}{\R'}$:
                            \begin{align*}
                                t'_1 &= t^{(0)}_1 + \dt + \sum_{(i_1, j_1) \in |\mb{I}_{\atr_1}|} (t^{(j_1)}_1 - t^{(i_1+1)}_1)\\
                                t'_2 &= t^{(0)}_2 + \dt + \sum_{(i_2, j_2) \in |\mb{I}_{\atr_2}|} (t^{(j_2)}_2 - t^{(i_2+1)}_2)
                            \end{align*}
                            But $t^{(0)}_1 = t^{(0)}_2$ since $c^{(0)}_1 \ueq c^{(0)}_2$.
                            Also, each operand in $(t^{(j_1)}_1 - t^{(i_1+1)}_1)$ equals the corresponding $(t^{(j_2)}_2 - t^{(i_2+1)}_2)$ because for each ($p^{th}$ element) $(i_1, j_1) \in \mb{I}_{\atr_1}$ and corresponding $(i_2, j_2) \in \mb{I}_{\atr_2}$, Property~\ref{prop:um-pres-tauhandle} guarantees that $t^{(i_1+1)}_1 = t^{(i_2+1)}_2$ and Property~\ref{prop:ueq-fine-ueq} guarantees that $t^{(j_1)}_1 = t^{(j_2)}_2$.
                        \item Finally, since no interaction with $\D$ via $\IN{}\hspace{-0.5em}$ or $\OUT{}\hspace{-0.5em}$ occurs in protected mode and since the same deterministic device performed the same number of steps (starting from $c^{(0)}_1 \ueq c^{(0)}_2$), it follows that $t'_{a_1} = t'_{a_2}$ and $\delta'_{1} = \delta'_{2}$.
                    \end{itemize}
                \end{itemize}
    \end{itemize}
\end{proof}

\subsubsection{Proof of preservation}\label{sssec:preservation}

Before proving the preservation and reflection of contextual equivalence, we prove the following facts about the trace semantics:
\begin{proposition}\label{prop:cnvimplconv}
    $\tconv{C[\M_M]}$ \emph{iff} $\exists \btr.\ \btrsemstar{\D}{\initconf{C}{\M_M}}{\btr \cdot \cnv}{\haltconf}$.
\end{proposition}
\begin{proof}
    We split the proof in the two directions:
    \begin{itemize}
        \item \emph{Case $\Rightarrow$.}
            By definition of $\tconv{C[\M_M]}$, we know that $\slmainstar{\D}{\initconf{C}{\M_M}}{\haltconf}$.
            Thus, definition of fine-grained and coarse-grained traces (Figures~\ref{fig:finegrained-obs} and~\ref{fig:coarsegrained-obs}) guarantee that the last observed action is $\cnv$ as requested.
        \item \emph{Case $\Leftarrow$.} Trivial.
    \end{itemize}
\end{proof}
\begin{proposition}\label{prop:eqmode}
    Let $C = \langle \M_C, \D \rangle$.
    If $\btrsemstar{\D}{\initconf{C}{\M_M}}{\btr}{c_1}$ and
    $\btrsemstar{\D}{\initconf{C}{\M_{M'}}}{\btr}{c_2}$, then
    $\moderel{c_1}$ and
    $\moderel{c_2}$.
\end{proposition}
\begin{proof}
    Let $\tb$ the last observable of $\btr$.
    By definition $c_1$ and $c_2$ are such that, for some $c'_1$ and $c'_2$:
    \[
        \atrsem{\D}{c'_1}{\ta}{c_1} \qquad  \atrsem{\D}{c'_2}{\ta}{c_2}
    \]
    with $\ta$ equal to $\cnv$, $\jmpin{\cdot}$ or $\jmpout{\cdot}{\cdot}$ (depending on the value of $\tb$).
    In either case, since $c'_1$ and $c'_1$ are the configuration \emph{right after} $\ta$ and by definition of fine-grained traces, we have $\moderel{c_1}$ and $\moderel{c_2}$.
\end{proof}
\begin{proposition}~\label{prop:tb-samesort}
For any context $C = \langle \M_C, \D \rangle$ and module $\M_M$, if $\btrsemstar{\D}{\initconf{C}{\M_M}}{\tbf{0}\ \cdots\ \tbf{n}}{c}$ with $n \geq 0$, then:
    \begin{enumerate}[(i)]
        \item Observables in even positions ($\tbf{0}$, $\tbf{2}$, \ldots) in traces are either $\cnv$ or $\jmpin{\R}$ (for some $\R$)
        \item Observables in odd positions ($\tbf{1}$, $\tbf{3}$, \ldots) in traces are either $\cnv$ or $\jmpout{\dt}{\R}$ (for some $\dt$ and $\R$)
    \end{enumerate}
\end{proposition}
\begin{proof}
    Both easily follow from Figures~\ref{fig:finegrained-obs} and~\ref{fig:coarsegrained-obs}.
\end{proof}

\paragraph{Reflection of $\teq$ at \SL.}\label{par:tr-reflection}
In this section we prove the implication $(i)$ of Figure~\ref{fig:strategy-rep}, i.e., that $\M_M \ttreq \M_{M'} \implies \M_M \teq \M_{M'}$.

First, we show that, due to the mitigation, the behavior of the context does not depend on the behavior of the module:
\begin{lemma}\label{lemma:umeqtraces}
    Let $C = \langle \M_C, \D \rangle$.
    If
    $\btrsemstar{\D}{\initconf{C}{\M_M}}{\btr}{c_1 \btrarrow{\tb} c'_1}$,
    $\btrsemstar{\D}{\initconf{C}{\M_{M'}}}{\btr}{c_2}$,
    $\moderelUM{c_1}$ and
    $\moderelUM{c_2}$,
    then $\btrsem{\D}{c_2}{\tb}{c'_2}$.
\end{lemma}
\begin{proof}
    First, observe that $\initconf{C}{\M_M} \ueq \initconf{C}{\M_{M'}}$, because
    \begin{align*}
        \initconf{C}{\M_M} &= \cfg{\BBot}{\deltainit}{0}{\bot}{\M_C \uplus \M_M}{\R^{\mi{init}}_{\M_C}}{\mtt{0xFFFE}}\\
        \initconf{C}{\M_{M'}} &= \cfg{\BBot}{\deltainit}{0}{\bot}{\M_C \uplus \M_{M'}}{\R^{\mi{init}}_{\M_C}}{\mtt{0xFFFE}}.
    \end{align*}
    Since $\moderelUM{\initconf{C}{\M_M}}$, $\initconf{C}{\M_M} \ueq \initconf{C}{\M_{M'}}$, $\btrsemstar{\D}{\initconf{C}{\M_M}}{\btr}{c_1}$, $\btrsemstar{\D}{\initconf{C}{\M_{M'}}}{\btr}{c_2}$, $\moderelUM{c_1}$ and $\moderelUM{c_2}$, by Property~\ref{prop:ueq-pres-btr} we have $c_1 \ueq c_2$.
    Finally, since $\btrsem{\D}{c_1}{\tb}{c'_1}$ and by Property~\ref{prop:ueq-um-samebeta} we get $\btrsem{\D}{c_2}{\tb}{c'_2}$.
\end{proof}

Then the following lemma shows that the isolation mechanism offered by the enclave guarantees that the behavior of the module is not influenced by the one of the context:
\begin{lemma}\label{lemma:pmeqtraces}
    Let $C = \langle \M_C, \D \rangle$. \\
    If
    $\M_M \ttreq \M_{M'}$,
    $\btrsemstar{\D}{\initconf{C}{\M_M}}{\btr}{c''_1 \btrarrow{\jmpin{\R_1}} c_1 \btrarrow{\tb} c'_1}$,
    $\btrsemstar{\D}{\initconf{C}{\M_{M'}}}{\btr}{c''_2 \btrarrow{\jmpin{\R_2}} c_2}$,
    then $\btrsem{\D}{c_2}{\tb}{c'_2}$.
\end{lemma}
\begin{proof}
    Noting that $\moderelPM{c_1}$ and that the last observable of $\btr$ is a $\jmpin{\cdot}$, by definition of coarse-grained traces (see Figure~\ref{fig:coarsegrained-obs}) we have the following fine-grained traces starting from $c''_1$:
    \begin{align*}
        \atrsemstar
        {\D}
        {c''_1}
        {\usilent\ \cdots\ \usilent \cdot \jmpin{\R_1}}
        {
            c_1
            \atrarrowstar{\atr_1}
            c^{(n_1)}_1
            \atrarrowstar{\psilent{k^{(n_1)}_1} \ \cdots\ \psilent{k^{(n_1 + m_1 - 1)}_1} \cdot \atr'_1}
            c'_1
        }
    \end{align*}
    with $\atr'_1 \in \{ \jmpout{k_1}{\R'_1}, \handle{k_1}\cdot \usilent \cdots \usilent \cdot \cnv \}$.

    Similarly for $c_2$ it must be:
    \begin{align*}
        \atrsemstar
        {\D}
        {c''_2}
        {\usilent\ \cdots\ \usilent \cdot \jmpin{\R_2}}
        {
            c_2
            \atrarrowstar{\atr_2}
            c^{(n_2)}_2
            \atrarrowstar{\psilent{k^{(n_2)}_2} \ \cdots\ \psilent{k^{(n_2 + m_2 - 1)}_1} \cdot \atr'_2}
            c'_2
        }.
    \end{align*}
    with $\atr'_2 \in \{ \jmpout{k_2}{\R'_2}, \handle{k_2} \cdot \usilent \cdots \usilent \cdot \cnv \}$.

    We have now two cases:
    \begin{itemize}
        \item \emph{Case $\tb = \jmpout{\dt}{\R}$.}
            $\M_M \ttreq \M_{M'}$ implies the existence of a context $C' = \langle \M_{C'}, \D' \rangle$ that allow us to observe $\btrsem{\D'}{\initconf{C'}{\M_{M'}}}{\btr}{c_3 \btrarrow{\tb}{c'_3}}$, i.e.\
            \[
                \atrsemstar{\D'}{c_3}{\atr_3}{
                    c^{(n_3)}_3
                    \atrarrow{\psilent{k^{(n_3)}_3} \ \cdots\ \psilent{k^{(n_3 + m_3 - 1)}_3} \cdot \atr'_3}
                    c'_3}
            \]
            with $\atr'_3 \in \{ \jmpout{k_3}{\R'_3}, \handle{k_3}\cdot \usilent \cdots \usilent \cdot \cnv \}$.

            By Properties~\ref{prop:init-jmpin} and~\ref{prop:peq-pres-btr} we have that $c_2 \peq c_3$, and by Property~\ref{prop:peq-pres-handlereti} we can conclude that $c^{(n_3)}_3 \peq c^{(n_2)}_2$.

            Property~\ref{prop:peq-swap-ctx} guarantees that
            \[
                \psilent{k^{(n_2)}_2} \ \cdots\ \psilent{k^{(n_2 + m_2 - 1)}_2} \cdot \atr'_2 = \psilent{k^{(n_3)}_3} \ \cdots\ \psilent{k^{(n_3 + m_3 - 1)}_3} \cdot \atr'_3.
            \]

            Since $\atr'_2 = \atr'_3 = \jmpout{k_3}{\R_1}$, we know that $\btrsem{\D}{c^{(n_2)}_2}{\jmpout{\dt'}{\R_1}}{c'_2}$.

            By Property~\ref{prop:btr-timings}, we have
            \begin{align*}
                \dt &= \sum_{i=0}^{n_1 + m_1} \ilen{c^{(i)}_1} + (11 + \MT)\cdot {|\mb{I}_{\atr_1}|}\\
                \dt' &= \sum_{i=0}^{n_2 + m_2} \ilen{c^{(i)}_2} + (11 + \MT)\cdot {|\mb{I}_{\atr_2}|}.
            \end{align*}

            Since by Properties~\ref{prop:ttreq-ilen} and~\ref{prop:ueq-pres-handlereti} we have $\sum_{i=0}^{n_1 + m_1} \ilen{c^{(i)}_1} = \sum_{i=0}^{n_2 + m_2} \ilen{c^{(i)}_2}$ and ${|\mb{I}_{\atr_1}|} = {|\mb{I}_{\atr_2}|}$, we get $\dt = \dt'$ as requested.

        \item \emph{Case $\tb = \cnv$.}
            Then it must be that $\atr'_1 = \handle{k_1}\cdot \usilent \cdots \usilent \cdot \cnv$ and $\atr'_2 = \handle{k_2}\cdot \usilent \cdots \usilent \cdot \cnv$.
            If this was not the case (i.e., if $\atr'_2 = \jmpout{k_2}{\R'_2}$), then $c_2$ could be swapped with $c_1$ (and $c_1$ with $c_2$) in the the statement of this Lemma and the previous case would apply.
            Thus, the thesis follows.
    \end{itemize}
\end{proof}

From the previous two lemmata we can then show the following:
\begin{lemma}\label{lemma:eqtraces}
    Given a context $C = \langle \M_C, \D \rangle$ and two modules $\M_M$ and $\M_{M'}$.
    If $\M_M \ttreq \M_{M'}$ and $\btrsemstar{\D}{\initconf{C}{\M_M}}{\btr}{c_1}$, then $\btrsemstar{\D}{\initconf{C}{\M_{M'}}}{\btr}{c_2}$.
\end{lemma}
\begin{proof}
    We can show this by induction on the length $n$ of $\btr$.
    \begin{itemize}
        \item $n = 0$. Since $\btr = \emptystr$, by definition of $\btrarrowstar{\cdot}$, we have $c_1 = \initconf{C}{\M_M} = c_1$.
            Again, by definition of $\btrarrowstar{\cdot}$, we can choose $c_2 = \initconf{C}{\M_{M'}}$ and get the thesis.

        \item $n = n' + 1$.
            The induction hypothesis (IHP) is then:
            \[
                \btrsemstar{\D}{\initconf{C}{\M_M}}{\btrp}{c'_1} \Rightarrow
                \btrsemstar{\D}{\initconf{C}{\M_{M'}}}{\btrp}{c'_2}
            \]
            and we must show that
            \[
                \btrsemstar{\D}{\initconf{C}{\M_M}}{\btrp}{c'_1 \btrarrow{\tb} c_1}
                \Rightarrow
                \btrsemstar{\D}{\initconf{C}{\M_{M'}}}{\btrp}{c'_2 \btrarrow{\tb} c_2}
            \]

            By cases on the CPU mode in $c'_1$ and $c'_2$:
            \begin{itemize}
                \item \emph{Case $\moderelUM{\R'_1[\rpc]}$ and $\moderelUM{\R'_2[\rpc]}$:}
                Follows by (IHP) and Lemma~\ref{lemma:umeqtraces}.
                \item \emph{Case $\moderelPM{\R'_1[\rpc]}$ and $\moderelPM{\R'_2[\rpc]}$:}
                Follows by (IHP) and Lemma~\ref{lemma:pmeqtraces}.
                \item \emph{Case $\moderel{\R'_1[\rpc]}$ and $\moderelp{\R'_2[\rpc]}$ and $\mtt{m} \neq \mtt{m'}$:}
                It never happens, as observed in Proposition~\ref{prop:eqmode}.
            \end{itemize}
    \end{itemize}
\end{proof}

Finally we can prove that $(i)$ from Figure~\ref{fig:strategy-rep} holds, i.e., that if two modules are trace equivalent then they are contextually equivalent in \SL:
\begin{lemma}\label{lemma:refl-tr}
    If $\M_M \ttreq \M_{M'}$ then $\M_M \teq \M_{M'}$.
\end{lemma}
\begin{proof}
    Expanding the definition of $\teq$, the statement becomes:
    \[
        \M_M \ttreq \M_{M'} \Rightarrow (\forall C = \langle \M_C, \D \rangle. \tconv{C[\M_M]} \iff \tconv{C[\M_{M'}]})
    \]

    We split the double implication and we show the two cases independently.
    \begin{itemize}
        \item \emph{Case $\Rightarrow$}, i.e., $\M_M \ttreq \M_{M'} \Rightarrow (\forall C. \tconv{C[\M_M]} \Rightarrow \tconv{C[\M_{M'}]})$.
        By Proposition~\ref{prop:cnvimplconv} there exists $\btr$ such that $\btrsemstar{\D}{\initconf{C}{\M_M}}{\btr\cdot\cnv}{\haltconf}$.

        Since $\M_M \ttreq \M_{M'}$, we know by Lemma~\ref{lemma:eqtraces} that $\btrsemstar{\D}{\initconf{C}{\M_{M'}}}{\btr\cdot\cnv}{\haltconf}$.
        Thus, again by Proposition~\ref{prop:cnvimplconv}, we have $\tconv{C[\M_{M'}]}$.

        \item \emph{Case $\Leftarrow$}, i.e., $\M_M \ttreq \M_{M'} \implies (\forall C. \tconv{C[\M_M]} \Leftarrow \tconv{C[\M_{M'}]})$, symmetric to the previous one.
    \end{itemize}
\end{proof}

\paragraph{Preservation of $\seq$ at \SH.}\label{par:tr-preservation}

In this section we prove the implications $(ii)$ - and consequently $(iii)$ - of Figure~\ref{fig:strategy-rep}, i.e., that $\M_M \seq \M_{M'} \implies \M_M \ttreq \M_{M'}$ and $\M_M \seq \M_{M'} \implies \M_M \teq \M_{M'}$.

For that, we first give a formal definition of \emph{distinguishing traces} for a pair of modules.
Then we give two algorithms that start from two distinguishing traces, their corresponding modules and the distinguishing context in \SL build a memory and a device that, put together as a context, differentiate the two modules in \SH.

\begin{definition}[Distinguishing traces]\label{def:distinguishing-traces}
    Let $\M_M$ and $\M_{M'}$ be two modules.
    We call $\btr = \btr_s \cdot \tb \cdot \btr_e \in \ttr{\M_M}$ and $\btr' = \btr_s \cdot \tb' \cdot \btr'_e \in \ttr{\M_{M'}}$ \emph{distinguishing traces for $\M_M$ and $\M_{M'}$} if
    $\tb \neq \tb'$,
    $\btr \notin \ttr{\M_{M'}}$, $\btr' \notin \ttr{\M_M}$ and
    they are observed under the same context $C^L$, i.e, $\btrsemstar{\D^L}{\initconf{C^L}{\M_M}}{\btr}{c}$ and $\btrsemstar{\D^L}{\initconf{C^L}{\M_{M'}}}{\btr'}{c'}$, for some $c, c'$.
\end{definition}
From now onwards, for simplicity, we write $\tb = \emptystr$ (resp. $\tb'  = \emptystr$) if $\btr$ (resp. $\btr'$) is shorter than $\btr'$ (resp. $\btr$).

\begin{property}\label{prop:distinguishing-exist}
    If $\M_M$ and $\M_{M'}$ are two modules such that $\M_M \nteq \M_{M'}$, then there always exist $\btr$ and $\btr'$ that are distinguishing traces for $\M_M$ and $\M_{M'}$.
\end{property}
\begin{proof}
    From the contrapositive of Lemma~\ref{lemma:refl-tr} we know that $\M_M \nttreq \M_{M'}$, i.e., there exist $\btr \in \ttr{\M_M}$ and $\btr' \in \ttr{\M_{M'}}$ such that $\btr \notin \ttr{\M_{M'}}$ and $\btr \in \ttr{\M_M}$.
    Also, since  $\M_M \nteq \M_{M'}$, we have that there exists a context $C^L$ such that $\tconv{C^L[\M_M]}$ and $\ntconv{C^L[\M_{M'}]}$ (or vice versa)  --- assume wlog $\tconv{C^L[\M_M]}$ and $\ntconv{C^L[\M_{M'}]}$.

    Thus, by Proposition~\ref{prop:cnvimplconv}:
    \begin{align*}
        &\btrsemstar{\D^L}{\initconf{C^L}{\M_M}}{\btr''}{\haltconf}\\
        &\btrsemstar{\D^L}{\initconf{C^L}{\M_{M'}}}{\btr'''}{c \neq \haltconf}
    \end{align*}
    for some $\btr''$ (ending in $\cnv$), $c$ and for all $\btr'''$ that can be observed.

    \noindent
    Indeed, we can always write that $\btr'' = \btr_s \cdot \tb \cdot \btr_e$ and $\btr''' = \btr_s \cdot \tb' \cdot \btr'_e$ where:
    \begin{itemize}
        \item $\btr_s$ is the longest (possibly empty) common prefix of the two traces
        \item $\tb$ and $\tb' \neq \cnv$ are the first different observables -- one of the two may be $\emptystr$ or, by Proposition~\ref{prop:cnvimplconv}, it may be $\tb = \cnv$
        \item $\btr_e$ and $\btr'_e$ are the (possibly empty) remainders of the two traces
    \end{itemize}
    Thus, since $\btr''$ and $\btr'''$ are also observed under the same context $C^L$, they are distinguishing traces.
\end{proof}

\paragraph{First algorithm: memory initialization.}

The pseudo-code in Algorithm~\ref{algo:memoryctx} describes how to build the memory of the distinguishing context starting from two distinguishing traces for the modules, $\btr = \btr_s \cdot \tb \cdot \btr_e$ and $\btr' = \btr_s \cdot \tb' \cdot \btr'_e$ (cf. Definition~\ref{def:distinguishing-traces}).
Throughout the algorithm we assume as given an \emph{assembler} function $\mi{encode}$ that takes an assembly instruction as input and returns its encoding as one or two words -- according to the size specified by Table~\ref{tab:op-summary-rep}.
Also, we assume that there is enough space in the unprotected memory to contain the context code: we do not lack generality since the required space for the code is bounded by a constant ($\leq 25$ words) plus the number of different addresses which the protected code jumps to (that must be part of the unprotected memory anyway).
Moreover, the algorithm uses five constants: each of them represents an unprotected memory address assumed different from $(i)$ each other, $(ii)$ $\mtt{0xFFFE}$ and $(iii)$ any address $\R[\rpc]$ such that $\jmpout{\dt}{\R}$ belongs to one of the input distinguishing traces.
For simplicity, assume that no jumps to $\mtt{0xFFFE}$ are performed by the modules.
Note that this limitation is easily lifted by changing Algorithm~\ref{algo:memoryctx} a bit: upon the jump into protected mode right before the said jump to $\mtt{0xFFFE}$ the context has to write the right code to deal with it in $\mtt{0xFFFE}$ and, afterwards, restore the old content of such an address.

Intuitively, the algorithm first initializes the memory of the context $\M_C$ by filling it with the code in~\figurename~\ref{fig:memctx}.
Then, if $\tb$ and $\tb'$ differ because they are both $\jmpout{\cdot}{\cdot}$ but with different registers, two cases arise:
\begin{itemize}
    \item If the register differentiating $\tb$ and $\tb'$ is $\mtt{r} \neq \rpc$, then, starting at address $\mtt{A\_RDIFF}$, add the code to request a new program counter (that will depend on the value of $\rn{r}$) to the device;
    \item Otherwise, add the code to request the new program counter at the addresses to which each of the modules jumps (call those addresses $\mi{joutd}$ and $\mi{joutd'}$).
\end{itemize}

The algorithm then adds the code to deal with jumps out from the protected module to unprotected code for any $\jmpout{\dt}{\R}$ in $\btr_s$ such that $\R[\rpc] \neq \mtt{joutd}$ and $\R[\rpc] \neq \mtt{joutd'}$.
Finally, the algorithm returns the memory built and the values of $\mi{joutd}$ and $\mi{joutd'}$ (to be used afterwards).
\begin{algorithm}
\begin{algorithmic}[1]
    \Procedure{BuildMem}{$\btr = \btr_s \cdot \tb \cdot \btr_e, \btr' = \btr_s \cdot \tb' \cdot \btr'_e$}
        \State $\quad\triangleright$ $\btr$ and $\btr'$ are distinguishing traces w. common prefix $\btr_s$
        \State $\mi{joutd} = \mi{joutd'} = \bot$
        \State $\M_C = \text{filled as described in Figure~\ref{fig:memctx}}$
        \If {$\tb = \jmpout{\dt}{\R} \land \tb' = \jmpout{\dt}{\R'} \land (\exists \rn{r}.\, \R[\rn{r}] \neq \R'[\rn{r}])$}
            \If {$\rn{r} \neq \rpc$}
                \State $\M_C = \M_C[\mtt{A\_RDIFF} \mapsto \mi{encode}(\mtt{OUT\ r}), \mtt{A\_RDIFF} + 1 \mapsto \mi{encode}(\mtt{IN\ \rpc})]$
            \Else
                \State $\mi{joutd} = \R[\rpc]$
                \State $\mi{joutd'} = \R'[\rpc]$
                \State $\M_C = \M_C[\mi{joutd} \mapsto \mi{encode}(\mtt{OUT\ \rpc}), \mi{joutd} + 1 \mapsto \mi{encode}(\mtt{IN\ \rpc})]$
                \State $\M_C = \M_C[\mi{joutd'} \mapsto \mi{encode}(\mtt{OUT\ \rpc}), \mi{joutd'} + 1 \mapsto \mi{encode}(\mtt{IN\ \rpc})]$
            \EndIf
        \EndIf

        \For {$\jmpout{\dt}{\R} \in \btr_s$}
            \If {$\R[\rpc] \neq \mi{joutd} \land \R[\rpc] \neq \mi{joutd'}$}
                \State $\M_C = \M_C[\R[\rpc] \mapsto \mi{encode}(\mtt{IN\ \rpc})]$
            \EndIf
        \EndFor
        \State \textbf{return } $(\M_C, \mi{joutd}, \mi{joutd'})$
    \EndProcedure
\end{algorithmic}

\caption{Builds the memory of the distinguishing context.}\label{algo:memoryctx}
\end{algorithm}
\begin{figure}
    \begin{lstlisting}
A_HALT. HLT

A_LOOP. JMP pc

A_JIN . IN sp
      . IN sr
      . IN $\reg{3}$
      . IN $\reg{4}$
      . IN $\reg{5}$
      . IN $\reg{6}$
      . IN $\reg{7}$
      . IN $\reg{8}$
      . IN $\reg{9}$
      . IN $\reg{10}$
      . IN $\reg{11}$
      . IN $\reg{12}$
      . IN $\reg{13}$
      . IN $\reg{14}$
      . IN $\reg{15}$
      . IN pc

A_EP  . OUT pc
      . IN pc

0xFFFE. A_EP
    \end{lstlisting}

    \caption{Initial content of unprotected memory as used by Algorithm~\ref{algo:memoryctx}.}\label{fig:memctx}
\end{figure}

\paragraph{Second algorithm: device construction.}

This second algorithm iteratively builds a device that cooperates with the memory of the context given by Algorithm~\ref{algo:memoryctx} to distinguish $\M_M$ from $\M_{M'}$.

The first two parameters of $\textsc{BuildDevice}$ -- $\mi{joutd}$ and $\mi{joutd'}$ -- are differentiating $\jmpout{\cdot}{\cdot}$ addresses (if any), as returned by the $\textsc{BuildMem}$ (Algorithm~\ref{algo:memoryctx}).
Parameters $\btr$ and $\btr'$ are distinguishing traces for $\M_M$ and $\M_{M'}$ generated under the context $C^L$ (cf. Definition~\ref{def:distinguishing-traces}).
Finally, $\mi{term}$ (resp.~$\mi{term'}$) denotes whether $\M_M$ (resp.~$\M_{M'}$) converges in a context with no interrupts after the last jump into protected mode.
\begin{algorithm}
    \caption{Builds the device of the distinguishing context.}\label{algo:devicectx}
    \begin{algorithmic}[1]
        \Procedure{BuildDevice}{$\mi{joutd}, \mi{joutd'}, \btr = \tb_0 \cdots \tb_{n-1} \cdot \tb \cdot \btr_e, \btr' = \tb_0 \cdots \tb_{n-1} \cdot \tb' \cdot \btr'_e, \mi{term}, \mi{term'}, C^L$}
            \State $\quad\triangleright$ $\mi{joutd}, \mi{joutd'}$ are differentiating $\jmpout{\cdot}{\cdot}$ addresses, if any
            \State $\quad\triangleright$ $\btr$ and $\btr'$ are distinguishing traces generated by the context $C^L$
            \State $\quad\triangleright$ $\mi{term}$ (resp.~$\mi{term'}$) denotes whether $\M_M$ (resp.~$\M_{M'}$) converges in a context with no interrupts after the last jump into protected mode
            \State $\Delta = \{ 0 \}$
            \State $\xleadsto{\cdot}_D\, = \emptyset$
            \State $\delta_L = 0$ \Comment{This variable keeps track of the last added device state.} \label{algo:devicectx:deltal}

            \For {$i \in 0 .. n - 1$}
                    \If {$\tb_i = \jmpin{\R}$}
                        \State $\Delta = \Delta \cup \{ \delta_L +1 , \ldots, \delta_L + 17 \}$
                        \State $\xleadsto{\cdot}_D\, =\, \xleadsto{\cdot}_D \cup\, \{ (\delta_L, wr(w), \delta_L) \mid w \in \mi{Word}\}$
                        \State $\xleadsto{\cdot}_D\, =\, \xleadsto{\cdot}_D \cup\, \{ (\delta_L, rd(\mtt{A\_JIN}), \delta_L+1) \}$ \label{algo:devicectx:ajin}
                        \State $\xleadsto{\cdot}_D\, =\, \xleadsto{\cdot}_D \cup\, \{ (\delta_L+1, rd(\R[\rsp]), \delta_L+2) \}$
                        \State $\xleadsto{\cdot}_D\, =\, \xleadsto{\cdot}_D \cup\, \{ (\delta_L+2, rd(\R[\rsr]), \delta_L+3) \}$
                        \State $\xleadsto{\cdot}_D\, =\, \xleadsto{\cdot}_D \cup\, \{ (\delta_L+i, rd(\R[\rn{i}]), \delta_L+i+1) \mid 3 \leq i \leq 15 \}$
                        \State $\xleadsto{\cdot}_D\, =\, \xleadsto{\cdot}_D \cup\, \{ (\delta_L+16, rd(\R[\rpc]), \delta_L+17) \}$
                        \State $\xleadsto{\cdot}_D\, =\, \xleadsto{\cdot}_D \cup\, \{ (\delta_L+i, \epsilon, \delta_L+i) \mid 0 \leq i \leq 16 \}$
                        \State $\delta_L = \delta_L + 17$
                    \ElsIf {$\tb_i = \jmpout{\dt}{\R}$}
                        \State $\xleadsto{\cdot}_D\, =\, \xleadsto{\cdot}_D\, \cup \{ (\delta_L, \epsilon, \delta_L) \} \cup \{ (\delta_L, wr(w), \delta_L) \mid w \in \mi{Word} \}$
                    \EndIf
            \EndFor

            \If {$\tb = \jmpout{\dt}{\R} \land \tb' = \jmpout{\dt'}{\R'} \land (\exists \rn{r}.\, \R[\rn{r}] \neq \R'[\rn{r}])$}
                \If {$\rn{r} \neq \rpc$}
                    \State $\Delta = \Delta \cup \{ \delta_L + 1 , \ldots, \delta_L + 4 \}$
                    \State $\xleadsto{\cdot}_D\, =\, \xleadsto{\cdot}_D \cup\, \{ (\delta_L, rd(\mtt{A\_RDIFF}), \delta_L+1) \}$
                    \State $\xleadsto{\cdot}_D\, =\, \xleadsto{\cdot}_D \cup\, \{ (\delta_L+1, wr(\R[\rpc]), \delta_L+2) \}$
                    \State $\xleadsto{\cdot}_D\, =\, \xleadsto{\cdot}_D \cup\, \{ (\delta_L+1, wr(\R'[\rpc]), \delta_L+3) \}$
                    \State $\xleadsto{\cdot}_D\, =\, \xleadsto{\cdot}_D \cup\, \{ (\delta_L+2, rd(\mtt{A\_HALT}), \delta_L+4) \}$ \label{algo:devicectx:ahalt-rdiff}
                    \State $\xleadsto{\cdot}_D\, =\, \xleadsto{\cdot}_D \cup\, \{ (\delta_L+3, rd(\mtt{A\_LOOP}), \delta_L+4) \}$ \label{algo:devicectx:aloop-rdiff}
                    \State $\xleadsto{\cdot}_D\, =\, \xleadsto{\cdot}_D \cup\, \{ (\delta_L+i, \epsilon, \delta_L+i) \mid 0 \leq i \leq 3 \}$
                    \State $\delta_L = \delta_L + 4$
                \Else
                    \State $\Delta = \Delta \cup \{ \delta_L + 1 , \ldots, \delta_L + 3 \}$
                    \State $\xleadsto{\cdot}_D\, =\, \xleadsto{\cdot}_D \cup\, \{ (\delta_L, wr(\mi{joutd}), \delta_L+1) \}$
                    \State $\xleadsto{\cdot}_D\, =\, \xleadsto{\cdot}_D \cup\, \{ (\delta_L, wr(\mi{joutd'}), \delta_L+2) \}$
                    \State $\xleadsto{\cdot}_D\, =\, \xleadsto{\cdot}_D \cup\, \{ (\delta_L+1, rd(\mtt{A\_HALT}), \delta_L+3) \}$ \label{algo:devicectx:ahalt-joutd}
                    \State $\xleadsto{\cdot}_D\, =\, \xleadsto{\cdot}_D \cup\, \{ (\delta_L+2, rd(\mtt{A\_LOOP}), \delta_L+3) \}$ \label{algo:devicectx:aloop-joutd}
                    \State $\xleadsto{\cdot}_D\, =\, \xleadsto{\cdot}_D \cup\, \{ (\delta_L+i, \epsilon, \delta_L+i) \mid 0 \leq i \leq 2 \}$
                    \State $\delta_L = \delta_L + 3$
                \EndIf

            \State continues ...
    \algstore{builddev}
    \end{algorithmic}
\end{algorithm}

\begin{algorithm}
    \begin{algorithmic}[1]
        \algrestore{builddev}
            \State ... continued
            \ElsIf {$\tb = \jmpout{\dt}{\R} \land \tb' = \jmpout{\dt'}{\R} \land \dt \neq \dt'$}
                \State $\quad\triangleright$ Let $\btrsemstar{\D^L}{\initconf{C}{\M_M}}{\btr_s}{c_1}$ and $\btrsem{\nI{\D^L}}{c_1}{\jmpout{\nI{\dt}}{\R}}{c'_1}$.
                \State $\quad\triangleright$ Let $\btrsemstar{\D^L}{\initconf{C}{\M_{M'}}}{\btr_s}{c_2}$ and $\btrsem{\nI{\D^L}}{c_2}{\jmpout{\nI{\dt'}}{\R}}{c'_2}$.
                \State $i = t'_1 - t_1$
                \State $i' = t'_2 - t_2$
                \State $\Delta = \Delta \cup \{ \delta_L + 1 , \ldots, \delta_L + \mi{max}(i, i') + 1 \}$
                \State $\xleadsto{\cdot}_D\, =\, \xleadsto{\cdot}_D \cup\, \{ (\delta_L+\mi{min}(i, i'), rd(\mtt{A\_HALT}), \delta_L + \mi{max}(i, i') + 1) \}$ \label{algo:devicectx:ahalt-time}
                \State $\xleadsto{\cdot}_D\, =\, \xleadsto{\cdot}_D \cup\, \{ (\delta_L+\mi{max}(i, i'), rd(\mtt{A\_LOOP}), \delta_L + \mi{max}(i, i') + 1)) \}$\label{algo:devicectx:aloop-time}
                \State $\xleadsto{\cdot}_D\, =\, \xleadsto{\cdot}_D \cup\, \{ (\delta_L+k, \epsilon, \delta_L+k+1) \mid 0 \leq k \leq \mi{max}(i, i')  \}$
                \State $\delta_L = \delta_L + \mi{max}(i, i') + 1$
            \ElsIf {$\tb = \cnv \land \tb' = \jmpout{\dt}{\R}$}
                \If {$\mi{term}$}
                    \State $\Delta = \Delta \cup \{ \delta_L + 1 , \ldots, \delta_L + 2 \}$
                    \State $\xleadsto{\cdot}_D\, =\, \xleadsto{\cdot}_D \cup\, \{ (\delta_L, wr(\mtt{A\_EP}), \delta_L+1) \}$
                    \State $\xleadsto{\cdot}_D\, =\, \xleadsto{\cdot}_D \cup\, \{ (\delta_L+1, rd(\mtt{A\_HALT}), \delta_L+2) \}$ \label{algo:devicectx:ahalt-cnv}
                    \State $\xleadsto{\cdot}_D\, =\, \xleadsto{\cdot}_D \cup\, \{ (\delta_L, rd(\mtt{A\_LOOP}), \delta_L+2) \}$ \label{algo:devicectx:aloop-cnv}
                    \State $\xleadsto{\cdot}_D\, =\, \xleadsto{\cdot}_D \cup\, \{ (\delta_L, wr(w), \delta_L) \mid w \in \mi{Word} \setminus \{ \mtt{A\_EP} \} \}$
                    \State $\xleadsto{\cdot}_D\, =\, \xleadsto{\cdot}_D \cup\, \{ (\delta_L+i, \epsilon, \delta_L+i) \mid 0 \leq i \leq 1 \}$
                    \State $\delta_L = \delta_L + 2$
                \Else
                    \State $\Delta = \Delta \cup \{ \delta_L + 1 \}$
                    \State $\xleadsto{\cdot}_D\, =\, \xleadsto{\cdot}_D \cup\, \{ (\delta_L, rd(\mtt{A\_HALT}), \delta_L+1) \}$ \label{algo:devicectx:aloop-cnv-noterm}
                    \State $\xleadsto{\cdot}_D\, =\, \xleadsto{\cdot}_D \cup\, \{ (\delta_L, wr(w), \delta_L) \mid w \in \mi{Word} \}$
                    \State $\xleadsto{\cdot}_D\, =\, \xleadsto{\cdot}_D \cup\, \{ (\delta_L, \epsilon, \delta_L) \}$
                    \State $\delta_L = \delta_L + 2$
                \EndIf
            \ElsIf {$\tb = \jmpout{\dt}{\R} \land \tb' = \emptystr$}
                \State $\quad\triangleright$ As the previous case, with $\mi{term'}$ in place of $\mi{term}$.
            \Else
                \State \textbf{return } $\bot$
            \EndIf

            \State $\D = \langle \Delta, 0, \xleadsto{\cdot}_D \rangle$ \Comment{As above, assume to have a sink state where all undefined actions lead to.}
            \State \textbf{return } $\D$
        \EndProcedure
    \end{algorithmic}
\end{algorithm}
The first two lines define the initial set of states, which will be a finite subset of $\mb{N}$ in the end, and the initial \emph{empty} transition function.

Line~\ref{algo:devicectx:deltal} defines $\delta_L$ that records the last state that was added to the I/O device.
At the beginning it is initialized to $0$.

The algorithm then proceeds by iterating over all the observables in $\btr_s$ (all the steps below also update $\Delta$ and $\delta_L$, but we omit to state it explicitly):
\begin{itemize}
    \item \emph{Case $\tb_i = \tb'_i = \jmpin{\R}$.}
        In this case we know that either this is the first observable or previous one was a $\jmpout{\cdot}{\cdot}$.
        Since the memory is obtained following Algorithm~\ref{algo:memoryctx}, we know that in both cases we reach the instruction $\IN \rpc$ (either at address $\mtt{A\_EP}$ or those of jumps out of protected mode), waiting for the next program counter (sometimes before that we perform a write, which shall be ignored).
        Thus, the device ignores any write operation and replies with $\mtt{A\_JIN}$ (line~\ref{algo:devicectx:ajin}).
        Then it starts to send the values of the registers in $\R$, so to simulate in \SH what happens in \SL and to match the requests from the code.
        To help the intuition Figure~\ref{fig:devicectx:same-jmpin} depicts how the transition function looks after the update (the solid black state denotes the new value of $\delta_L$).
    \item \emph{Case $\tb_i = \tb'_i = \jmpout{\dt}{\R}$.}
        The device is simply updated with a loop on $\delta_L$ with action $\epsilon$ and ignores any write operation (so as to deal with $\R[\rpc] = \mi{joutd}$ or $\R[\rpc] = \mi{joutd'}$).
        Figure~\ref{fig:devicectx:same-jmpout} pictorially represents this case.
\end{itemize}

Then, when $\btr_s$ ends, the algorithm analyses $\tb$ and $\tb'$ and sets up the device to differentiate the two modules:
\begin{itemize}
    \item \emph{Case $\tb = \jmpout{\dt}{\R} \land \tb' = \jmpout{\dt'}{\R'} \land (\exists \rn{r}.\, \R[\rn{r}] \neq \R'[\rn{r}])$.}
        In this case the differentiation is due to a register, and two further sub-cases may arise, depending on whether it is $\rpc$.
        If the register is $\rpc$ then the device waits for the differentiating value for the context (that is executing code at $\mi{joutd}$ and $\mi{joutd'}$ by construction) and based on that value, it replies with either $\mtt{A\_HALT}$ (line~\ref{algo:devicectx:ahalt-joutd}) or $\mtt{A\_LOOP}$ (line~\ref{algo:devicectx:aloop-joutd}).
        Instead, if the differentiation register is not $\rpc$ then the code of the context is waiting for the next program counter and the context replies with $\mtt{A\_RDIFF}$.
        From this address we find the code that sends the differentiating register and, based on that value, the device replies with either $\mtt{A\_HALT}$ (line~\ref{algo:devicectx:ahalt-rdiff}) or $\mtt{A\_LOOP}$ (line~\ref{algo:devicectx:aloop-rdiff}).
        Figures~\ref{fig:devicectx:diff-jmpout-regs} and~\ref{fig:devicectx:diff-jmpout-regspc} may help the intuition.

    \item \emph{Case $\tb = \jmpout{\dt}{\R} \land \tb' = \jmpout{\dt'}{\R} \land \dt \neq \dt'$.}
        This case is probably the most interesting since differentiation happens in \SL due to timings.
        However, different timings in \SL correspond to different timings in \SH (as observed in proof of Property~\ref{prop:build-distinguishing}), and the device is programmed to reply with either $\mtt{A\_HALT}$ (line~\ref{algo:devicectx:ahalt-time}) or $\mtt{A\_LOOP}$ (line~\ref{algo:devicectx:aloop-time}) depending on the time value.
        Figure~\ref{fig:devicectx:diff-jmpout-time} intuitively depicts this situation.

    \item \emph{Case $\tb = \cnv \land \tb' = \jmpout{\dt}{\R}$.}
        In this case $\cnv$ may occur during an interrupt service routine.
        We then have two sub-cases, depending on whether the first module terminates when executed in a context with no interrupts after the last jump into protected mode or not (i.e., encoded by the value of $\mi{term}$).
        When $\mi{term}$ holds, the first module makes the CPU go through an exception handling configuration that jumps to $\mtt{A\_EP}$ and the device instructs the code to jump to $\mtt{A\_HALT}$ (line~\ref{algo:devicectx:ahalt-cnv}), while for the second module the CPU jumps to any other location ($\mtt{A\_EP}$ is chosen to be different from any other jump out address!) and is instructed to jump to $\mtt{A\_LOOP}$ (line~\ref{algo:devicectx:aloop-cnv}).
        When $\mi{term}$ does not hold, the first module diverges, while for the second module the CPU jumps to a location in unprotected code and it is instructed to jump to $\mtt{A\_HALT}$ (line~\ref{algo:devicectx:aloop-cnv-noterm}).
        Figures~\ref{fig:devicectx:diff-cnv-jmpout-term} and ~\ref{fig:devicectx:diff-cnv-jmpout-noterm} may help the intuition.

    \item \emph{Case $\tb = \jmpout{\dt}{\R} \land \tb' = \emptystr$.}
        Analogous to the previous case.

    \item \emph{Otherwise.}
        No other cases may arise, as noted in Property~\ref{prop:builddevice-isdevice}.
\end{itemize}

Finally, the algorithm returns a device with the set of states $\Delta$, the initial state $0$ and the transition function built as just explained.
\tikzset{every state/.style={minimum size=1em}}
\begin{figure}
    \begin{subfigure}[b]{0.7\textwidth}
        \centering
        \begin{tikzpicture}[>=stealth',shorten >=1pt, node distance=2cm]
            \tikzstyle{ndS} = [state, fill]
            \node[state] (dS)      {$\delta_L$};
            \node[state] (d0) [right of=dS]  { };
            \node[state] (d1) [right of=d0]  { };
            \node[state] (d2) [right of=d1]  { };
            \node[state] (d3) [right of=d2]  { };
            \node[state] (d4) [right of=d3]  { };
            \node[state] (d5) [below of=d4]  { };
            \node[state] (d6) [left of=d5]  { };
            \node[state] (d7) [left of=d6]  { };
            \node[state] (d8) [left of=d7]  { };
            \node[state] (d9) [left of=d8]  { };
            \node[state] (d10) [left of=d9]  { };
            \node[state] (d11) [below of=d10]  { };
            \node[state] (d12) [right of=d11]  { };
            \node[state] (d13) [right of=d12]  { };
            \node[state] (d14) [right of=d13]  { };
            \node[state] (d15) [right of=d14]  { };
            \node[ndS] (d16) [right of=d15]  { };

            \path[->] (dS)  edge [loop above] node {$\epsilon$} (dS);
            \path[->] (dS)  edge [loop left] node {\footnotesize $\mi{wr}(\_)$} (dS);
            \path[->] (d0)  edge [loop above] node {$\epsilon$} (d0);
            \path[->] (d1)  edge [loop above] node {$\epsilon$} (d1);
            \path[->] (d2)  edge [loop above] node {$\epsilon$} (d2);
            \path[->] (d3)  edge [loop above] node {$\epsilon$} (d3);
            \path[->] (d4)  edge [loop above] node {$\epsilon$} (d4);
            \path[->] (d5)  edge [loop right] node {$\epsilon$} (d5);
            \path[->] (d6)  edge [loop above] node {$\epsilon$} (d6);
            \path[->] (d7)  edge [loop above] node {$\epsilon$} (d7);
            \path[->] (d8)  edge [loop above] node {$\epsilon$} (d8);
            \path[->] (d9)  edge [loop above] node {$\epsilon$} (d9);
            \path[->] (d10)  edge [loop above] node {$\epsilon$} (d10);
            \path[->] (d11)  edge [loop left] node {$\epsilon$} (d11);
            \path[->] (d12)  edge [loop above] node {$\epsilon$} (d12);
            \path[->] (d13)  edge [loop above] node {$\epsilon$} (d13);
            \path[->] (d14)  edge [loop above] node {$\epsilon$} (d14);
            \path[->] (d15)  edge [loop above] node {$\epsilon$} (d15);
            \path[->] (dS)  edge [above] node {\footnotesize $\mi{rd}(\mtt{JIN})$} (d0);
            \path[->] (d0)  edge [above] node {\footnotesize $\mi{rd}(\R[\rsp])$} (d1);
            \path[->] (d1)  edge [above] node {\footnotesize $\mi{rd}(\R[\rsr])$} (d2);
            \path[->] (d2)  edge [above] node {\footnotesize $\mi{rd}(\R[\reg{3}])$} (d3);
            \path[->] (d3)  edge [above] node {\footnotesize $\mi{rd}(\R[\reg{4}])$} (d4);
            \path[->] (d4)  edge [bend left, sloped, above] node {\footnotesize $\mi{rd}(\R[\reg{5}])$} (d5);
            \path[->] (d5)  edge [above] node {\footnotesize $\mi{rd}(\R[\reg{6}])$} (d6);
            \path[->] (d6)  edge [above] node {\footnotesize $\mi{rd}(\R[\reg{7}])$} (d7);
            \path[->] (d7)  edge [above] node {\footnotesize $\mi{rd}(\R[\reg{8}])$} (d8);
            \path[->] (d8)  edge [above] node {\footnotesize $\mi{rd}(\R[\reg{9}])$} (d9);
            \path[->] (d9)  edge [above] node {\footnotesize $\mi{rd}(\R[\reg{10}])$} (d10);
            \path[->] (d10)  edge [bend right, sloped, below] node {\footnotesize $\mi{rd}(\R[\reg{11}])$} (d11);
            \path[->] (d11)  edge [above] node {\footnotesize $\mi{rd}(\R[\reg{12}])$} (d12);
            \path[->] (d12)  edge [above] node {\footnotesize $\mi{rd}(\R[\reg{13}])$} (d13);
            \path[->] (d13)  edge [above] node {\footnotesize $\mi{rd}(\R[\reg{14}])$} (d14);
            \path[->] (d14)  edge [above] node {\footnotesize $\mi{rd}(\R[\reg{15}])$} (d15);
            \path[->] (d15)  edge [above] node {\footnotesize $\mi{rd}(\R[\rpc])$} (d16);
        \end{tikzpicture}

        \caption{The case of $\tb_i = \tb'_i = \jmpin{\R}$.} \label{fig:devicectx:same-jmpin}
    \end{subfigure}
    ~
    \begin{subfigure}[b]{0.2\textwidth}
        \centering
        \begin{tikzpicture}[>=stealth',shorten >=1pt, node distance=2.8cm]
            \tikzstyle{ndS} = [state, fill, text=white, font=\bfseries]
            \node[ndS] (dS)      {$\delta_L$};
            \path[->] (dS)  edge [loop above] node {$\epsilon$} (dS);
            \path[->] (dS)  edge [loop below] node {\footnotesize $\mi{wr}(\_)$} (dS);
        \end{tikzpicture}

        \caption{The case of $\tb_i = \tb'_i = \jmpout{\dt}{\R}$.} \label{fig:devicectx:same-jmpout}
    \end{subfigure}

    \bigskip
    \begin{subfigure}[b]{0.6\textwidth}
        \centering
        \begin{tikzpicture}[>=stealth',shorten >=1pt, node distance=2.8cm]
            \tikzstyle{ndS} = [state, fill]
            \node[state] (dS)      {$\delta_L$};
            \node[state] (d0) [right of=dS]  { };
            \node[state] (d1) [right of=d0]  { };
            \node[state] (d2) [below of=d0]  { };
            \node[ndS] (d3) [right of=d1]  { };

            \path[->] (dS)  edge [loop above] node {$\epsilon$} (dS);
            \path[->] (d0)  edge [loop above] node {$\epsilon$} (d0);
            \path[->] (d1)  edge [loop above] node {$\epsilon$} (d1);
            \path[->] (d2)  edge [loop left] node {$\epsilon$} (d2);
            \path[->] (dS)  edge [above] node {\footnotesize $\mi{rd}(\mtt{A\_RDIFF})$} (d0);
            \path[->] (d0)  edge [sloped, anchor=center,above] node {\footnotesize $\mi{wr}(\R'[\mtt{r}])$} (d1);
            \path[->] (d0)  edge [sloped, anchor=center,below] node {\footnotesize $\mi{wr}(\R[\mtt{r}])$} (d2);
            \path[->] (d1)  edge [sloped, anchor=center,above] node {\footnotesize $\mi{rd}(\mtt{A\_LOOP})$} (d3);
            \path[->] (d2)  edge [sloped, anchor=center, below, bend right] node {\footnotesize $\mi{rd}(\mtt{A\_HALT})$} (d3);
        \end{tikzpicture}

        \caption{The case of $\tb_i = \jmpout{\dt}{\R} \land \tb'_i = \jmpout{\dt'}{\R'} \land (\exists \rn{r}.\, \R[\rn{r}] \neq \R'[\rn{r}])$.} \label{fig:devicectx:diff-jmpout-regs}
    \end{subfigure}
    ~
    \begin{subfigure}[b]{0.4\textwidth}
        \centering
        \begin{tikzpicture}[>=stealth',shorten >=1pt, node distance=2.8cm]
            \tikzstyle{ndS} = [state, fill]
            \node[state] (dS)      {$\delta_L$};
            \node[state] (d0) [right of=dS]  { };
            \node[state] (d1) [below of=dS]  { };
            \node[ndS] (d2) [right of=d0]  { };

            \path[->] (dS)  edge [loop above] node {$\epsilon$} (dS);
            \path[->] (d0)  edge [loop above] node {$\epsilon$} (d0);
            \path[->] (d1)  edge [loop left] node {$\epsilon$} (d1);
            \path[->] (dS)  edge [sloped, anchor=center,above] node {\footnotesize $\mi{wr}(\mi{joutd})$} (d0);
            \path[->] (dS)  edge [sloped, anchor=center,below] node {\footnotesize $\mi{wr}(\mi{joutd'})$} (d1);
            \path[->] (d0)  edge [sloped, anchor=center,above] node {\footnotesize $\mi{rd}(\mtt{A\_LOOP})$} (d2);
            \path[->] (d1)  edge [sloped, anchor=center, below, bend right] node {\footnotesize $\mi{rd}(\mtt{A\_HALT})$} (d2);
        \end{tikzpicture}

        \caption{The case of $\tb_i = \jmpout{\dt}{\R} \land \tb'_i = \jmpout{\dt'}{\R'} \land  \R[\rpc] \neq \R'[\rpc]$.} \label{fig:devicectx:diff-jmpout-regspc}
    \end{subfigure}

    \bigskip
    \begin{subfigure}[b]{\textwidth}
        \centering
        \begin{tikzpicture}[>=stealth',shorten >=1pt, node distance=2.8cm]
            \tikzstyle{ndS} = [state, fill]
            \tikzstyle{empty} = [draw=none,fill=none]
            \node[state] (dS)      {$\delta_L$};
            \node[empty] (d0) [right of=dS]  {$\ldots$};
            \node[state,outer sep=2.3mm] (d1) [right of=d0]  { };
            \node[empty] (d2) [right of=d1]  {$\ldots$};
            \node[state,outer sep=2.3mm] (d3) [right of=d2]  { };
            \node[ndS] (d4) [below of=d3]  { };

            \path[->] (dS)  edge [sloped, anchor=center,above] node {$\epsilon$} (d0);
            \path[->] (d0)  edge [sloped, anchor=center,above] node {$\epsilon$} (d1);
            \path[->] (d1)  edge [sloped, anchor=center,above] node {$\epsilon$} (d2);
            \path[->] (d2)  edge [sloped, anchor=center,above] node {$\epsilon$} (d3);
            \path[->] (d1)  edge [sloped, below, bend right] node {\footnotesize $\mi{rd}(\mtt{A\_LOOP})$} (d4);
            \path[->] (d3)  edge [sloped, above] node {\footnotesize $\mi{rd}(\mtt{A\_HALT})$} (d4);

            \draw [decorate, decoration={brace, raise=2mm, amplitude=5pt}] (dS.north) -- node[sloped, above=3mm] {\footnotesize${\mi{min}(i, i')}$} (d1.north);
            \draw [decorate, decoration={brace, raise=2mm, amplitude=5pt}] (d1.north) -- node[sloped, above=3mm] {\footnotesize${\mi{max}(i, i') - \mi{min}(i, i')}$} (d3.north);
        \end{tikzpicture}

        \caption{The case of $\tb_i = \jmpout{\dt}{\R} \land \tb'_i = \jmpout{\dt'}{\R} \land \dt \neq \dt'$. Let $i$ and $i'$ as in Algorithm~\ref{algo:devicectx}.} \label{fig:devicectx:diff-jmpout-time}
    \end{subfigure}

    \bigskip
    \begin{subfigure}[b]{0.5\textwidth}
        \centering
        \begin{tikzpicture}[>=stealth',shorten >=1pt, node distance=2.8cm]
            \tikzstyle{ndS} = [state, fill]
            \tikzstyle{empty} = [draw=none,fill=none]
            \node[state] (dS)      {$\delta_L$};
            \node[state] (d0) [right of=dS]  { };
            \node[ndS] (d1) [right of=d0]  { };

            \path[->] (dS)  edge [loop above] node {$\epsilon$} (dS);
            \path[->] (dS)  edge [loop below] node {\footnotesize $\mi{wr}(w), w \neq \mtt{A\_EP}$} (dS);
            \path[->] (dS)  edge [sloped, anchor=center,above] node {{\footnotesize $\mi{wr}(\mtt{A\_EP})$}} (d0);
            \path[->] (d0)  edge [sloped, anchor=center,above] node {{\footnotesize $\mi{rd}(\mtt{A\_HALT})$}} (d1);
            \path[->] (dS)  edge [sloped, anchor=center, below, bend right] node {{\footnotesize $\mi{rd}(\mtt{A\_LOOP})$}} (d1);
        \end{tikzpicture}

        \caption{The case of $\tb_i = \cnv \land \tb'_i = \jmpout{\dt}{\R} \land \mi{term}$.} \label{fig:devicectx:diff-cnv-jmpout-term}
    \end{subfigure}
    ~
    \begin{subfigure}[b]{0.5\textwidth}
        \centering
        \begin{tikzpicture}[>=stealth',shorten >=1pt, node distance=2.8cm]
            \tikzstyle{ndS} = [state, fill]
            \tikzstyle{empty} = [draw=none,fill=none]
            \node[state] (dS)      {$\delta_L$};
            \node[ndS] (d0) [right of=dS]  { };

            \path[->] (dS)  edge [loop above] node {$\epsilon$} (dS);
            \path[->] (dS)  edge [loop below] node {\footnotesize $\mi{wr}(\_)$} (dS);
            \path[->] (dS)  edge [sloped, anchor=center,above] node {{\footnotesize $\mi{rd}(\mtt{A\_HALT})$}} (d0);
        \end{tikzpicture}

        \caption{The case of $\tb_i = \cnv \land \tb'_i = \jmpout{\dt}{\R} \land \lnot\mi{term}$.} \label{fig:devicectx:diff-cnv-jmpout-noterm}
    \end{subfigure}

    \caption{Graphical representations of the updates performed by Algorithm~\ref{algo:devicectx} to the transition function of the device.}
\end{figure}
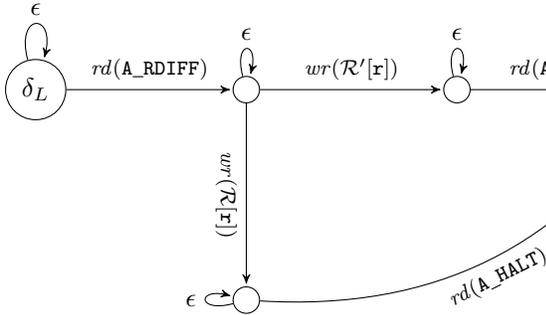
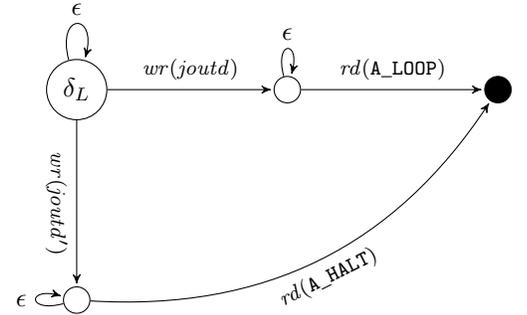
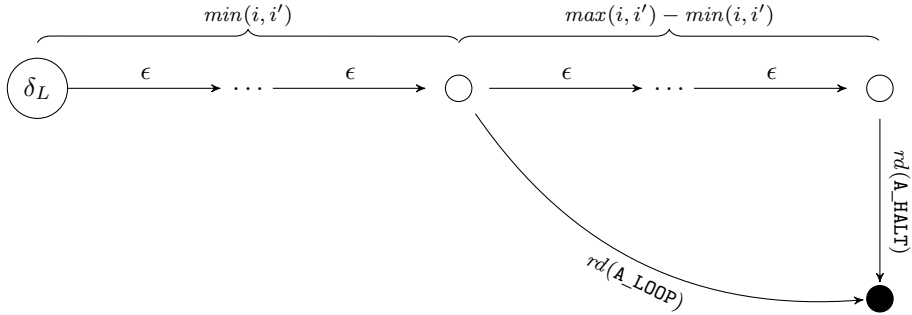
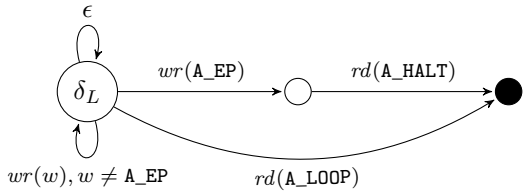
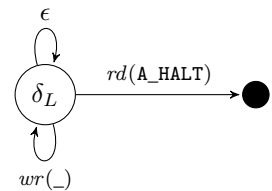

The first property about $\textsc{BuildDevice}$ states that, under the right conditions, it always produces an actual I/O device:
\begin{property}\label{prop:builddevice-isdevice}
    Let $\M_M \nttreq \M_{M'}$, $\btr, \btr'$ be distinguishing traces of $\M_M$ and $\M_{M'}$ originated by some context $C^L$ and let $\mi{term}$ and $\mi{term'}$ be any pair of booleans,
    then
    $\D = \textsc{BuildDevice}(\btr, \btr', \mi{joutd}, \mi{joutd'}, \mi{term}, \mi{term'}, C^L) \neq \bot$ and $\D$ is an I/O device.
\end{property}
\begin{proof}
    We first show that $\textsc{BuildDevice}$ never returns $\bot$ when $\btr$ and $\btr'$ are distinguishing traces.
    For that, let $\btr = \btr_s  \cdot \tb \cdot \btr_e$ and $\btr' = \btr_s \cdot \tb' \cdot \btr'_e$, and note that the only cases for which $\bot$ is returned are the following:
    \begin{itemize}
        \item \emph{Case $\tb = \tb' = \cnv$.}
            Since $\tb \neq \tb'$ by hypothesis, this case never happens.
        \item \emph{Case $\tb = \jmpout{\dt}{\R}$ and $\tb' = \jmpin{\R'}$ (or vice versa).}
            This case never happens due to Proposition~\ref{prop:tb-samesort}.
        \item \emph{Case $\{\cnv, \jmpin{\R} \} \ni \tb \neq \tb' \in \{\cnv, \jmpin{\R'} \}$.}
            Roughly, this means that the \emph{same} context performed two different actions upon observation of the same trace ($\btr_s$).
            Formally, we know by hypothesis that for the context $C^L = \langle \M_C, \D^L \rangle$
            \begin{align*}
                &\btrsemstar{\D^L}{\initconf{C^L}{\M_M}}{\btr_s}{c_1}\\
                &\btrsemstar{\D^L}{\initconf{C^L}{\M_{M'}}}{\btr_s}{c_2}.
            \end{align*}
            with $\moderelUM{c_1}$ and $\moderelUM{c_2}$.
            Property~\ref{prop:ueq-pres-btr} guarantees that $c_1 \ueq c_2$, thus by Property~\ref{prop:ueq-um-samebeta} the same observable must originate from both $c_1$ and $c_2$, but that is against the hypothesis that $\tb \neq \tb'$.
    \end{itemize}
    Finally, it is easy to see that $\D$ returned by $\textsc{BuildDevice}$ is an actual device.
    Indeed, its set of states $\Delta$ is finite (the algorithm always terminates in a finite number of steps and each step adds a finite number of state); its initial state $0$ belongs to $\Delta$;
 since a sink state is assumed to exist, no $\mi{int?}$ transitions are ever added and a single $\mi{rd}(w)$ transition outgoes from any given state: thus the transition relation respects the definition of I/O devices.
\end{proof}
Before stating and proving the reflection itself, we need some further definitions and properties.

The following property states that the context built by joining together the results of the two algorithms above is a distinguishing one:
\begin{property}\label{prop:build-distinguishing}
    Let $\M_M \nttreq \M_{M'}$; let $C^L = \langle \M_C, \D^L \rangle$; let
    \begin{align*}
        &\btrsemstar{\D^L}{\initconf{C^L}{\M_M}}{\btr_s}{c'_1 \btrarrow{\tb} c_1}\\
        &\btrsemstar{\D^L}{\initconf{C^L}{\M_{M'}}}{\btr_s}{c'_2 \btrarrow{\tb'} c_2}
    \end{align*}
    be such that $\btr = \btr_s \cdot \tb \cdot \btr_e$ and $\btr' = \btr_s \cdot \tb' \cdot \btr_e$ distinguishing traces of $\M_M$ and $\M_{M'}$;
    and let
    \begin{align*}
        &\mi{term} \iff \slmainstar{\nI{\D^L}}{c'_1}{\haltconf}\\
        &\mi{term'} \iff \slmainstar{\nI{\D^L}}{c'_2}{\haltconf}.
    \end{align*}

    If $(\M_C, \mi{joutd}, \mi{joutd'}) = \textsc{BuildMem}(\btr, \btr')$, $\D = \textsc{BuildDevice}(\btr, \btr', \mi{joutd}, \mi{joutd'}, \mi{term}, \mi{term'})$ and $C^H = \langle \M_C, \D \rangle$, then $\sconv{C^H[\M_M]}$ and $\nsconv{C^H[\M_{M'}]}$ (or vice versa).
\end{property}
\begin{proof}
    Assume wlog that $\tconv{C^L[\M_M]}$ and $\ntconv{C^L[\M_{M'}]}$.
    By Lemma~\ref{lemma:nointeq}
    \begin{align*}
       \sconv{C^H[\M_M]} \iff \tconv{\nI{C^H}[\M_M]} \quad\text{ and }\quad \sconv{C^H[\M_{M'}]} \iff \tconv{\nI{C^H}[\M_{M'}]}
    \end{align*}
    It suffices thus proving that $\nI{C^H}$ distinguishes $\M_M$ and $\M_{M'}$, i.e., $\tconv{\nI{C^H}[\M_M]}$ and $\ntconv{\nI{C^H}[\M_{M'}]}$ or vice versa.

    We show by induction on the length $2n + 1$ of $\btr_s$ that if
    \begin{align*}
        &\btrsemstar{\D^L}{\initconf{C^L}{\M_M}}{\btr_s}{c'_1}\\
        &\btrsemstar{\D^L}{\initconf{C^L}{\M_{M'}}}{\btr_s}{c'_2}
    \end{align*}
    then $\exists \btr'_s$ s.t.
    \begin{align*}
        &\btrsemstar{\nI{D^H}}{\initconf{\nI{C^H}}{\M_M}}{\btr'_s}{c_3} \text{ and}\\
        &\btrsemstar{\nI{D^H}}{\initconf{\nI{C^H}}{\M_{M'}}}{\btr'_s}{c_4} \text{ with }  \btr'_s \tapprox \btr_s \text{ (see Definition~\ref{def:btr-approx-btrp}).}
    \end{align*}
    Note that the length of $\btr_s$ must be odd as a consequence of Properties~\ref{prop:ueq-pres-btr} and~\ref{prop:ueq-um-samebeta} and no $\cnv$ appears in it since otherwise it would mean that $\btr = \btr'$.
    \begin{itemize}
        \item \emph{Case $n = 0$.}
            Then, $\btr_s$ is $\jmpin{\R}$.
            Thus, Algorithm~\ref{algo:memoryctx} guarantees that the current instruction is $\IN \rpc$ (at address $\mtt{A\_EP}$) and its execution leads to address $\mtt{A\_JIN}$ (by Algorithm~\ref{algo:devicectx}) and the same $\jmpin{\R}$ is observed starting from both $\initconf{\nI{C^H}}{\M_M}$ and $\initconf{\nI{C^H}}{\M_{M'}}$ and also $\btr'_s \tapprox \btr_s$.
        \item \emph{Case $n = n' + 1$.}
            If
            \begin{align*}
                \btrsemstar{\D^L}{\initconf{C^L}{\M_M}}{\btr''_s}{c''_1} &\land \btrsemstar{\D^L}{\initconf{C^L}{\M_{M'}}}{\btr''_s}{c'''_2}\\
                &\Downarrow\\
                \btrsemstar{\nI{D^H}}{\initconf{\nI{C^H}}{\M_M}}{\btr'''_s}{c'_3} &\land \btrsemstar{\nI{D^H}}{\initconf{\nI{C^H}}{\M_{M'}}}{\btr'''_s}{c'_4}\ \land \btr'''_s \tapprox \btr''_s \text{ (IHP)}
            \end{align*}
            then
            \begin{align*}
                \btrsemstar{\D^L}{\initconf{C^L}{\M_M}}{\btr''_s}{c''_1 \btrarrowstar{\btr''} c'_1} &\land
                \btrsemstar{\D^L}{\initconf{C^L}{\M_{M'}}}{\btr''}{c''_2  \btrarrowstar{\btr''} c'_2} \\
                &\Downarrow\\
                \btrsemstar{\nI{D^H}}{\initconf{\nI{C^H}}{\M_M}}{\btr'''_s}{c'_3 \btrarrowstar{\btr'''} c_3} &\land
                \btrsemstar{\nI{D^H}}{\initconf{\nI{C^H}}{\M_{M'}}}{\btr'''_s}{c'_4 \btrarrowstar{\btr'''} c_4}\ \land
                \btr'''_s \cdot \btr''' \tapprox \btr''_s \cdot \btr''.
            \end{align*}
            Note that it must be that $\btr'' = \jmpout{\dt}{\R}\cdot\jmpin{\R'}$ by Proposition~\ref{prop:tb-samesort} and because we never observe $\cnv$ in the common prefix.
            By (IHP) and Property~\ref{prop:peq-pres-btr} we have $c''_1 \peq c'_3$ and $c''_2 \peq c'_4$.
            Thus, by Properties~\ref{prop:peq-pres-handlereti} and~\ref{prop:peq-swap-ctx}, it must be that $\jmpout{\dt'}{\R}$ is observed when starting in $c'_3$ and $\jmpout{\dt''}{\R}$ is observed when starting in $c'_4$ (for some $\dt'$ and $\dt''$).

            By definition of coarse-grained traces, each of the computations above is generated by fine-grained trace in the form (we write $\_$ to denote a generic configuration):
            \begin{align*}
                &\atrsem{\D^L}{\_}{\jmpin{\R''}}{c''_1 = c^{(0)}_1 \atrarrow{\ta^{(0)}_1} \cdots \atrarrow{\ta^{(n_1-1)}_1} c^{(n_1)}_1 \atrarrow{\jmpout{k^{(n_1)}_1}{\R}} c^{(n_1)+1} \atrarrowstar{\usilent\cdots\usilent\jmpin{\R'}} c'_1}\\
                &\atrsem{\D^L}{\_}{\jmpin{\R''}}{c''_2 = c^{(0)}_2 \atrarrow{\ta^{(0)}_2} \cdots \atrarrow{\ta^{(n_2-1)}_2} c^{(n_2)}_2 \atrarrow{\jmpout{k^{(n_2)}_2}{\R}} c^{(n_2)+1} \atrarrowstar{\usilent\cdots\usilent\jmpin{\R'}} c'_2}\\
                &\atrsem{\nI{D^H}}{\_}{\jmpin{\R''}}{c'_3 = c^{(0)}_3 \atrarrow{\ta^{(0)}_3} \cdots \atrarrow{\ta^{(n_3-1)}_3} c^{(n_3)}_3 \atrarrow{\jmpout{k^{(n_3)}_3}{\R}} c^{(n_3+1)}_3}\\
                &\atrsem{\nI{D^H}}{\_}{\jmpin{\R''}}{c'_4 = c^{(0)}_4 \atrarrow{\ta^{(0)}_4} \cdots \atrarrow{\ta^{(n_4-1)}_4} c^{(n_4)}_4 \atrarrow{\jmpout{k^{(n_4)}_4}{\R}} c^{(n_4+1)}_4}.
            \end{align*}
            Thus, due to Property~\ref{prop:btr-timings} and by hypothesis, it holds that $\dt = \sum_{i=0}^{n_1} \ilen{c^{(i)}_1} + (11 + \MT) \cdot |\mb{I}_{\ta^{(0)}_1\cdots\ta^{(n_1)}_1}| = \sum_{i=0}^{n_2} \ilen{c^{(i)}_2} + (11 + \MT) \cdot |\mb{I}_{\ta^{(0)}_2\cdots\ta^{(n_2)}_2}|$.
            Also, since by (IHP) and Properties~\ref{prop:ueq-pres-btr} and~\ref{prop:ueq-um-samebeta} it follows that $c^{(0)}_1 = c''_1 \ueq c''_2 = c^{(0)}_2$, we know $|\mb{I}_{\ta^{(0)}_1\cdots\ta^{(n_1)}_1}| = |\mb{I}_{\ta^{(0)}_2\cdots\ta^{(n_2)}_2}|$ (by Property~\ref{prop:ueq-pres-handlereti}) and thus $\sum_{i=0}^{n_1} \ilen{c^{(i)}_1} = \sum_{i=0}^{n_2} \ilen{c^{(i)}_2}$.
            Moreover, by (IHP) and Property~\ref{prop:peq-pres-btr}, we get $c^{(0)}_1 = c''_1 \peq c'_3 = c^{(0)}_3$ and $c^{(0)}_2 = c''_2 \peq c'_4 = c^{(0)}_4$.
            Now, as a consequence of Properties~\ref{prop:peq-samedecode},~\ref{prop:peq-pres-handlereti} and~\ref{prop:peq-swap-ctx} we know that $\dt' = \sum_{i=0}^{n_3} \ilen{c^{(i)}_3} = \sum_{i=0}^{n_1} \ilen{c^{(i)}_1} = \sum_{i=0}^{n_2} \ilen{c^{(i)}_2} = \sum_{i=0}^{n_3} \ilen{c^{(i)}_3} = \dt''$.
            By (IHP) and since the first observable after $c'_3$ and $c'_4$ is the same, by Property~\ref{prop:ueq-pres-btr} it follows $c^{(n_3+1)}_3 \ueq c^{(n_4+1)}_4$.
            Thus, due to Property~\ref{prop:ueq-um-samebeta}, we get that the same coarse-grained observable $\jmpin{\R'''}$ is observed after  $c^{(n_3+1)}_3$ and $c^{(n_4+1)}_4$.
            Finally, $\R'''$ is equal to $\R'$ since after any $\jmpout{\cdot}{\cdot}$ a $\IN \rpc$ instruction is executed and its execution leads to address $\mtt{A\_JIN}$ (by Algorithm~\ref{algo:devicectx}) that performs $\jmpin{\R}$, and the thesis follows.
    \end{itemize}

    \bigskip
    Since we proved that
    \begin{align*}
        &\btrsemstar{\nI{D^H}}{\initconf{\nI{C^H}}{\M_M}}{\btr'_s}{c_3} \text{ and}\\
        &\btrsemstar{\nI{D^H}}{\initconf{\nI{C^H}}{\M_{M'}}}{\btr'_s}{c_4}
    \end{align*}
    we also have that $c_3 \ueq c_4$ by Properties~\ref{prop:ueq-pres-btr} and~\ref{prop:ueq-um-samebeta}.

    Let $\btrsemstar{\nI{D^H}}{c_3}{\btr_3}{c''_3}$ and $\btrsemstar{\nI{D^H}}{c_4}{\btr_4}{c''_4}$, with $\btr_3$ and $\btr_4$ either empty or made of a single observable (either $\cnv$ or $\jmpout{\cdot}{\cdot}$, since no difference cannot be observed upon $\jmpin{\cdot}$ as observed above). By exhaustive cases on $\tb$ and $\tb'$ we have:
    \begin{itemize}
        \item \emph{Case $\tb = \cnv$ and $\tb' = \jmpout{\dt'''}{\R''}$.}
            Note that, since $\mi{term} \iff \slmainstar{\nI{\D^L}}{c'_1}{\haltconf}$ and $c'_1 \peq c_3$ (by Properties~\ref{prop:init-jmpin} and~\ref{prop:peq-pres-btr}), we get $\mi{term} \iff \slmainstar{\nI{\D^H}}{c_3}{\haltconf}$ by Property~\ref{prop:peq-swap-ctx} and since neither $\nI{\D^L}$ nor $\nI{\D^H}$ raise any interrupt.
            Thus, by definition of $\D^L$ (cf. Algorithm~\ref{algo:devicectx}) the context $C^H$ distinguishes the two modules.
        \item \emph{Case $\tb = \jmpout{\dt'''}{\R''}$ and $\tb' = \emptystr$.}
            Similar to the previous case (with $\mi{term'}$ in place of $\mi{term}$).
        \item \emph{Case $\tb = \jmpout{\dt'''}{\R''}$ and $\tb' = \jmpout{\dt'''}{\R'''}$ with $\R'' \neq \R'''$.}
            Since $c'_1 \peq c_3$ and $c'_2 \peq c_4$, it must be that $\btr_3 = \jmpout{\dt^{v}}{\R''}$ and $\btr_4 = \jmpout{\dt^{vi}}{\R''}$.
            Thus, by Algorithms~\ref{algo:memoryctx} and~\ref{algo:devicectx}, $C^H$ distinguishes the two modules.
        \item \emph{Case $\tb = \jmpout{\dt'''}{\R''}$ and $\tb' = \jmpout{\dt^{iv}}{\R''}$.}
            In this case it holds that $\btr_3 = \jmpout{\dt^{v}}{\R''}$ and $\btr_4 = \jmpout{\dt^{vi}}{\R''}$ with the same timings of the instructions (by Property~\ref{prop:btr-timings}).
            Since $c_3 \ueq c_4$, the two times must differ one from each other otherwise, by the counterpositive of Property~\ref{prop:ttreq-ilen}, we would get $\M_M \ttreq \M_{M'}$.
            Again, by definition of Algorithms~\ref{algo:memoryctx} and~\ref{algo:devicectx}, one computation converges and one diverges, hence $C^H$ distinguishes the two modules.
    \end{itemize}
\end{proof}

Finally, we can use the above algorithms and results to prove that if two modules are contextually equivalent in \SH, then they are also contextually equivalent in \SL.
\begin{lemma}\label{lemma:pres-tr}
    If $\M_M \seq \M_{M'}$ then $\M_M \teq \M_{M'}$.
\end{lemma}
\begin{proof}
    We prove the contrapositive, i.e., if $\M_M \nteq \M_{M'}$ then $\M_M \nseq \M_{M'}$.
    Since $\M_M \nteq \M_{M'}$, assume wlog that $\tconv{C^L[\M_M]}$ and $\ntconv{C^L[\M_{M'}]}$.
    By Property~\ref{prop:distinguishing-exist} we know that a pair of distinguishing traces for $\M_M$ and $\M_{M'}$ exist.
    Algorithm~\ref{algo:memoryctx} and~\ref{algo:devicectx} witness the existence of a context $C^H$ that -- due to Properties~\ref{prop:builddevice-isdevice} and~\ref{prop:build-distinguishing} (with the right $\mi{term}$ and $\mi{term'}$) -- is an actual context and is guaranteed to differentiate $\M_M$ from $\M_{M'}$, i.e., $\sconv{C^H[\M_M]}$ and $\nsconv{C^H[\M_{M'}]}$ (or vice versa).
    Thus, by definition of contextually equivalent modules in \SH, we get  $\M_M \nseq \M_{M'}$ as requested.
\end{proof}

\paragraph{Full abstraction.}

Finally, we can restate the original full abstraction theorem and prove it.
\fullabstractionrep*{}
\begin{proof} \hfill
    \begin{itemize}
        \item Direction $\Rightarrow$ follows from Lemma~\ref{lemma:reflection}.
        \item Direction $\Leftarrow$ (i.e., $(iii)$ in~\figurename~\ref{fig:strategy-rep}), follows directly from Lemma~\ref{lemma:pres-tr}.
    \end{itemize}
\end{proof}


\FloatBarrier

\end{document}